\documentclass[12pt]{article}

\advance\voffset by -2.0cm \advance\hoffset by -1.25cm
\usepackage{cite}
\usepackage{amsmath}
\usepackage{amssymb}
\usepackage{amsthm}
\usepackage[shortlabels]{enumitem}
\usepackage{authblk}

\usepackage{lscape}
\usepackage{booktabs}
\usepackage{tabularx}
\usepackage{ltablex}
\usepackage{caption}
\usepackage{array}

\usepackage{multirow}
\usepackage[table]{xcolor}
\usepackage[colorlinks,linkcolor=blue]{hyperref}

\def\eatcell#1\unskip{}
\newcommand{\CC}{\mathbb{C}} 
\newcommand{\RR}{\mathbb{R}} 
\newcommand{\QQ}{\mathbb{Q}} 
 
\newcommand{\HH}{\mathbb{H}}
\newcommand{\ZZ}{\mathbb{Z}} 
 
\newcommand{\PP}{\mathbb{P}}

\newcommand{\C}{\mathcal{C}}
\newcommand{\E}{\mathcal{E}}
\newcommand{\F}{\mathcal{F}}

\newcommand{\Hh}{\mathcal{H}}

\newcommand{\M}{\mathcal{M}}
\newcommand{\N}{\mathcal{N}}

\newcommand{\calS}{\mathcal{S}}

\newtheorem{theorem}{Theorem}
\newtheorem{proposition}[theorem]{Proposition}
\newtheorem{lemma}[theorem]{Lemma}

\newtheorem{claim}[theorem]{Claim}

\DeclareMathOperator{\Hom}{Hom}
\DeclareMathOperator{\Sym}{Sym}
\DeclareMathOperator{\Aut}{Aut}

\DeclareMathOperator{\rk}{rk}

\DeclareMathOperator{\Tr}{Tr}
\DeclareMathOperator{\sgn}{sgn}
\DeclareMathOperator{\Div}{div}
\DeclareMathOperator{\Hilb}{Hilb}

\numberwithin{equation}{section}

\def\be{\begin{equation}}
\def\ee{\end{equation}}

\allowdisplaybreaks[3]
\title{Some comments on symmetric orbifolds of K3}
\author{Roberto Volpato\thanks{volpato@pd.infn.it}}
\date{}
\affil{\small{Dipartimento di Fisica e Astronomia `Galileo Galilei' and\\ INFN sez. di Padova\authorcr
		Via Marzolo 8, 35131 Padova (Italy)}}
\begin{document}
	\maketitle

	\abstract{We consider two dimensional $\mathcal{N}=(4,4)$ superconformal field theories in the moduli space of symmetric orbifolds of K3. We complete a classification of the discrete groups of symmetries of these models, conditional to a series of assumptions and with certain restrictions.  Furthermore, we provide a partial classification of the set of twining genera, encoding the action of a discrete symmetry $g$ on a space of supersymmetric states in these models. These results suggest the existence of a number of surprising identities between seemingly different Borcherds products, representing Siegel modular forms of degree two and level $N>1$. We also provide a critical review of various properties of the moduli space of these superconformal field theories, including the groups of dualities, the set of singular models and the locus of symmetric orbifold points, and describe some puzzles related to our (lack of) understanding of these properties.

	}
	\pagebreak
	\tableofcontents
	

	\section{Introduction}
	
	Non-linear sigma models (NLSM) with target space a symmetric orbifold $\Sym^n(M_4)$, where $M_4=T^4$ or $K3$, have been the subject of intense study in the last two decades. The interest in these models is mostly due to the role they play in the AdS$_3$/CFT$_2$ correspondence as the holographic dual of AdS$_3\times S^3\times M_4$ \cite{Maldacena:1997re}. The symmetric orbifold models represent some very special points in the moduli space of two dimensional $\N=(4,4)$ superconformal field theories (SCFTs) at central charge $c=6n$. A generic point in this moduli space (at least in the connected component containing symmetric orbifolds) can be thought of as a non-linear sigma model on a hyperk\"ahler manifold of real dimension $4n$, which is a deformation of $\Hilb^n M_4$, the Hilbert scheme of $n$-points on the manifold $M_4$. In the context of type IIB string theory, these models arise as the infrared effective theories describing the worldsheet dynamics of a system of branes wrapping (some cycles in) the manifold $M_4$. The most studied example of this construction involves a system of $Q_5$ D5-branes wrapping $M_4\times S^1$ and $Q_1$ D1-branes wrapping $S^1$, with $Q_1Q_5=n$.
	
	This article is focused on the case $M_4=K3$. In this case, the relevant target spaces are known in algebraic geometry as hyperk\"ahler manifolds of K3$^{[n]}$ type. The symmetric orbifold $\Sym^nK3$ can be obtained as a suitable limit of such manifolds, where the geometry becomes singular. In analogy with the terminology in algebraic geometry, we will call the SCFTs in the moduli space of $Sym^n K3$ `of $K3^{[n]}$ type'. These SCFTs are generically not free, even in the symmetric orbifold limit, and this makes explicit computations particularly challenging, at least for finite $n$, where holographic methods do not apply -- for example, the partition function of a generic such model is not known. 
%

 The first major result of this article is a classification (up to isomorphisms) of the possible groups of discrete symmetries (satisfying certain conditions) of these SCFTs, for any $n$ and any point in the moduli space. This is the generalization to manifolds of K3$^{[n]}$ type of the analogous result for non-linear sigma models on a single K3 surface \cite{K3symm}. More precisely, we will consider symmetries acting by linear transformations on the states (or equivalently, on the fields) of the CFT, such that the $\N=(4,4)$ supercurrents are invariant, and commuting with the spectral flow transformation relating the Neveu-Schwarz and the Ramond sector, both on the left-moving and on the right-moving side (see conditions (1) and (2) in section \ref{s:callifsymm}). We will also require some technical restrictions on these symmetries, see condition (3) in section \ref{s:callifsymm}. These restrictions are clearer if one interprets the SCFT as the world-sheet theory for a system of branes in type IIB string theory, for example a bound state of D1-D5 branes. In this context, a first condition is that we focus on symmetries whose action extends to the whole type IIB string theory. The second constraint can be described very roughly as the condition that the symmetries (or rather their extension to the full type IIB string theory) do not to exchange branes with antibranes. In fact, it is quite surprising that these requirements are meaningful at all, i.e. that there might exist some symmetry that does \emph{not} satisfy them.  We will argue that this is the case in section \ref{s:dualities}. Finally, our analysis will exclude the models corresponding to a certain set of points in the moduli space. Conjecturally, this excluded set might be empty, i.e. the `models' we are excluding might be singular limits in the moduli space where the SCFT is not really well defined. However, while we will provide some reasonable arguments in favour of this conjecture (see section \ref{s:singular}), assuming that this is true seems to lead to some puzzles, related in particular to the locus of symmetric orbifolds. As long as these puzzles are not clarified, it seems safer to claim explicitly that these models, if they exist, are not part of our analysis. With all caveats in mind, we can now formulate our classification result (proposition \ref{th:main}, section \ref{s:callifsymm}): we will show that the group $G$ of symmetries of a (non-singular) SCFT of K3$^{[n]}$ type is isomorphic to a subgroup of the Conway group $Co_0$ (the group of automorphisms of the Leech lattice $\Lambda$), which fixes a sublattice of $\Lambda$ of dimension at least $3$. Conversely, given  any subgroup $G$ of $Co_0$ with this property and given any $n>1$, there exists a (non-singular) SCFT of $K3^{[n]}$ type having exactly $G$ as a group of symmetries. This is very similar to the result for NLSM on a single K3 surface, the only difference there being that the groups $G\subset Co_0$ were required to fix a $4$-dimensional, rather than $3$-dimensional, sublattice \cite{K3symm}. The fact that $K3^{[n]}$-models for $n>1$ admit symmetry groups that did not occur for NLSMs on a single K3 (for example, the Mathieu groups $M_{11}$ and $M_{22}$) might play a role in the context of Umbral moonshine.
 
 In order to obtain this result, the main input from physics is the global form of the moduli space of SCFTs of K3$^{[n]}$-type. This moduli space is believed to be obtained by taking the quotient of a certain Grassmannian by a discrete group of dualities; from this quotient, one needs to exclude a subset corresponding to `singular' limits where the SCFT is not well defined  \cite{Dijkgraaf:1998gf,Seiberg:1999xz}. The reader should be warned that our results are conditional to some assumptions (described in section \ref{s:callifsymm}) about the precise structure of the duality group and the location of the singular models. While these assumptions are quite standard, the arguments supporting them are not as stringent as, for example, the case of the moduli space of NLSMs on K3.
 
 If one relaxes our conditions  and considers also symmetries acting by (roughly speaking) charge conjugation on the D-branes, then the corresponding groups can, in some cases, get larger, and become $\ZZ_2$-extensions of the groups described above. We were not able to determine exactly at which points in the moduli space these enlargement occurs and the precise group structure of the extensions.
 
There are analogous attempts in algebraic geometry to classify the groups of symplectic automorphisms of hyperk\"ahler manifolds of K3$^{[n]}$ type (see  \cite{Mongardi2015,Mongardi2016,huybrechts,HohnMason2019}), which are subgroups of the groups of symmetries of the corresponding NLSM. For a single K3 surface, the problem was solved a long time ago by a famous theorem by Mukai \cite{mukai1988finite}. The classification of these `geometric' symmetries seems more complicated than the one of symmetries of SCFTs. In particular, certain groups $G$ seem to appear only for some manifolds of K3$^{[n]}$-type only for $n$ `large enough'; by contrast, for SCFTs, a certain group $G$ either occurs for all $n>1$ or does not occur at all. The results of \cite{K3symm}, valid for non-linear sigma models on K3, were subsequently interpreted in \cite{huybrechts} in terms of a classification of the groups of autoequivalences of the derived category of coherent sheaves on K3 surfaces. It is tempting to conjecture that the results of this article admit analogous interpretations in terms of derived categories of hyperk\"ahler manifolds of K3$^{[n]}$ type.

\bigskip

The second result of this article concerns the so called twining genera. Given a SCFT of K3$^{[n]}$-type with a group of symmetries $G$, one can define a $g$-twining elliptic genus (twining genus, for short) $\phi_g^{K3^{[n]}}(\tau,z)$ obtained from the elliptic genus by inserting the symmetry $g$ in the trace (see the definition \eqref{twindef}). One case where a twining genus can be easily computed is when the SCFT of $K3^{[n]}$ type is the $n$-th symmetric orbifold of a NLSM $\C_{K3}$ on a K3 surface, and the symmetry $g$ is induced by a symmetry $g'$ of the seed theory $\C_{K3}$. In this case, provided that the twining genera $\phi_{g'^r}^{\C_{K3}}$ of $g'$ and of all its powers $g'^r$ are known for $\C_{K3}$, one can compute a generating function  (sometimes known as the `second quantized twining genus')
\be \Psi_{g'}(\sigma,\tau,z)=\sum_{n=0}^\infty p^n \phi^{Sym^n\C_{K3}}_{g'}\ ,\qquad\qquad p:=e^{2\pi i\sigma}\ ,
\ee
 for the elliptic genera $\phi^{Sym^n\C_{K3}}_{g'}$ of the induced symmetries in the  symmetric orbifold $\Sym^n\C_{K3}$, for all $n$. Upon including an `automorphic correction', the second quantized twining genus becomes a (meromorphic) Siegel modular form $1/\Phi_g$ of genus $2$ under a certain congruence subgroup of $Sp(4,\ZZ)$; explicitly, \be 1/\Phi_{g}(\sigma,\tau,z)=\frac{\Psi_g(\sigma,\tau,z)}{p\psi_g(\tau,z)}\ .\ee
 Both $\Psi_g$ and $\Phi_g$ can be defined explicitly in terms of infinite products, whose exponents depend on the Fourier coefficients of  $\phi_{g'^r}^{\C_{K3}}$ for all powers $g'^r$. The set of possible twining genera of NLSM on K3 are known \cite{Cheng:2016org,Paquette:2017gmb} (although, for some symmetries $g$, there are more than one `candidate' twining genus, and it is not known which one is actually realized in a given NLSM, see \cite{Paquette:2017gmb}), so the computation can be effectively performed.
 
Standard arguments (reviewed in section \ref{s:conjclass}) show that a  twining genus $\phi_g^{K3^{[n]}}$ is invariant under deformations of the moduli that preserve the symmetry $g$, i.e. deformations generated by exactly marginal operators that are $g$-invariants. This means that there are connected families of models where the twining genus $\phi_g^{K3^{[n]}}$ is well-defined and constant -- for NLSMs on K3, an analogous argument was developed in \cite{Cheng:2016org}. Therefore, it is sufficient to compute the twining genus at a suitable point (for example, a symmetric orbifold), in order to obtain the genus for all models in the corresponding family. 

For this reason, it is interesting to study what are the possible families of models with a given symmetry $g$, and in particular families containing a symmetric orbifold point. We prove that, for each $n$, there are only finitely many distinct families of SCFTs of K3$^{[n]}$ type with a non-trivial symmetry $g$ (once we identify families related by dualities). We were not able to determine exactly how many such families exist, in general. However, we provide a complete classification of families containing symmetric orbifold points, and such that the corresponding symmetry $g$ is induced by a symmetry of the seed theory (see Table \ref{tab:big}). We also argue that there exist families of SCFTs of K3$^{[n]}$ that are not of this kind, i.e. whose symmetry $g$ is never the  symmetry of a seed theory $\C_{K3}$ in a symmetric orbifold point $\Sym^n\C_{K3}$. For these families, we do not know any method for computing the corresponding twining genus. 
 
 A surprising result is the following (see section \ref{s:elevenandfriends}). Consider a pair of NLSMs on K3, $\C_{K3}$ and $\C'_{K3}$,   with symmetries, respectively, $g$ and $g'$ of the same order, but not belonging to the same family. In general, one expects the corresponding twining genera $\phi^{\C_{K3}}_g$ and $\phi^{\C'_{K3}}_{g'}$ to be different. Our classification, however, shows that for some $n$ (sometimes for all $n>1$), the symmetric orbifolds $Sym^n \C_{K3}$ and $Sym^n \C'_{K3}$ with the induced symmetries $g$ and $g'$ belong to the same family. This implies that, while the twining genera of the seed theory are different, the ones of the $n$-th symmetric orbifold are the same. This phenomenon happens, for example, for symmetries of order 11: there are $4$ isolated points in the moduli space of NLSMs on K3 having a symmetry of order $11$. There are two possible twining genera for symmetries of order $11$ \cite{Paquette:2017gmb}. While one cannot determine directly which twining genus is realized in each of the four isolated points, it is reasonable to expect that not all of them are the same; this expectation also fits nicely with a conjecture in \cite{Cheng:2016org}, related to the so-called Umbral Moonshine phenomenon \cite{EOT,Cheng:2012tq,Cheng:2013wca,Cheng:2014zpa}. On the other hand, in this article we show that, for each $n>1$, the $n$-th symmetric products of the four NLSM on K3 all sit in the same family with a symmetry of order $11$. This means that the second quantized twining genera $\Psi_g$ of all these four K3 NLSM must have the same coefficients for all powers $p^n$ with $n>1$. On the other hand, the second quantized twining genera are given by infinite products whose exponents are determined the twining genera of the seed theory. If the twining genera of the seed theory are really different, we would have different infinite products that only differ for the $p=1$ term, thus implying an infinite number of non-trivial cancellations occurring. This seems even more puzzling if one consider the `automorphic corrected' Siegel modular forms. In this case, one would have that the difference of the two Siegel modular forms contains only a term with power $p^0$, which seems to be incompatible with the fact that this difference must be modular. All these arguments seem to falsify the conjecture that there are two distinct twining genera of order $11$ for NLSM on K3. However, miracles seem to happen. Some numerical experiments suggest that the two different second quantized twining genera do have the same $p^n$ terms for $n>1$, at least up to $n=12$. And even the modularity argument is avoided: indeed, for order $11$ symmetries, the corresponding Siegel modular forms turn out to have weight $0$. Furthermore, a direct calculation shows that the coefficients of the $p^0$ term in the corresponding Siegel modular forms, which in general are functions of $\tau$ and $z$, differ only by a constant. Thus, one would have that the difference of two meromorphic Siegel modular forms of weight $0$ is just a constant, which is perfectly compatible with modularity. In fact, this argument might be used to prove that the two meromorphic Siegel forms are the same up to a constant. Indeed, the location of the poles is strongly constrained by wall crossing \cite{Paquette:2017gmb}, so it is not unlikely that the two forms have exactly the same poles. If this is the case, then the difference must be a holomorphic Siegel form of weight zero, which is necessarily a constant.  While we focused here in the introduction on symmetries of order $11$, similar phenomena occur for several other symmetries: in total, we have $13$ pairs of distinct Borcherds products that differ only by a constant. In \emph{all} these cases, the corresponding Siegel modular forms have weight $0$ and the $p^0$-term in the difference is independent of $\tau$ and $z$. While such coincidences, if true, might be explained by modularity, it is still remarkable that there are so many examples where different infinite products lead to the same modular forms up to a constant. This is not an uncommon phenomenon for genus one modular forms, in particular products of eta functions, but we were not aware of any such example at genus two.
	
The paper is organized as follows. In section \ref{s:SCFTs}, we describe some general properties of the moduli space of SCFTs of $K3^{[n]}$ type. Most of these results are well-known, but we propose a new perspective on various aspects that are relevant for our subsequent analysis. In  particular, in section \ref{s:dualities} we discuss the group of dualities of these superconformal field theories, and its relationship with the well-known group of dualities to the type IIB string theory on $K3$. In section \ref{s:singular}, we discussed the limits in the moduli space where the models become singular; this is mostly based on \cite{Seiberg:1999xz}, but we will suggest the possibility that those results might be modified. In section \ref{s:symmorbmoduli}, we discuss the locus of symmetric orbifold points in the moduli space. The analysis of this section, together with \ref{s:dualities} and \ref{s:singular}, leads to some puzzles, whose solution is not clear to us. These puzzles suggest that our (at least, the author's) understanding of the moduli space of SCFTs of $K3^{[n]}$-type is not as watertight as, for example, the one of NLSM on a single K3 surface. In section \ref{s:mainres}, we describe our new results. In section \ref{s:callifsymm} we classify the groups of symmetries of $K3^{[n]}$ models. In section \ref{s:twining}, we review the main properties of twining genera, whose classification is described in section \ref{s:conjclass}. Finally, in section \ref{s:elevenandfriends}, we discuss the consequences of these results for the Borcherds products corresponding to the second quantized twining genera. Section \ref{s:conclusions} contains some suggestions for future research. Most of the technical results are contained in the appendices, where we also review the basic properties of the moduli space of type IIB on K3 and some basic facts about lattices that are needed in the rest of the paper.

%
%
	
	\section{Superconformal field theories of $K3^{[n]}$-type}\label{s:SCFTs}

	The moduli space $\M_n$ of the $N=(4,4)$ models at $c=6n$, $n>1$, in the connected component of $\Sym^n(K3)$, is believed to be an open subset in the quotient \cite{Dijkgraaf:1998gf}
	\be \M_n\subseteq O^+_n\backslash Gr^+(4,21) \ ,\qquad n>1\ ,
	\ee where
	\be Gr^+(a,b):= O^+(a,b,\RR)/(SO(a,\RR)\times O(b,\RR))\ ,\ee is the Grassmannian parametrising oriented positive definite $a$-dimensional subspaces of the space $\RR^{a,b}$. Here, $O^+(a,b,\RR)\equiv O^+(a,b)$ denotes the subgroup of the real orthogonal group $O(a,b,\RR)$ preserving the orientation of a positive definite $a$-dimensional subspace in $\RR^{a,b}$. 
	
	The discrete group $O^+_n$, which depends on $n$, will be discussed at length in section \ref{s:dualities}. We will argue the $O^+_n$ is a subgroup of the group of automorphisms $O^+(\Gamma^{5,21})=O(\Gamma^{5,21})\cap O^+(5,21,\RR)$ of an even unimodular lattice $\Gamma^{5,21}\subset \RR^{5,21}$ of signature $(5,21)$.   In particular, $O^+_n$ contains the subgroup $Stab^+(v)$ of $O(\Gamma^{5,21})$ that fixes a primitive vector $v\in \Gamma^{5,21}$ of squared norm $v^2=2n-2$. The space $Gr^+(4,21)$ is interpreted as the Grassmannian of oriented positive definite four-dimensional subspaces $\Pi$ in the space $v^\perp \cap \RR^{5,21}\cong \RR^{4,21}$, where $\RR^{5,21}=\Gamma^{5,21}\otimes \RR$. With this identification, there is a natural action of $Stab^+(v)$ on $\RR^{4,21}$, which induces an action on $Gr^+(4,21)$.
	
	As will be discussed in section \ref{s:singular}, the open subspace $\M_n$ is the complement in $ O^+_n\backslash Gr^+(4,21)$ of some `singular loci' of real codimension at least $4$; as a consequence, $\M_n$ is connected.  In the limit where we approach such a singular point, the superconformal field theory is supposed to show some pathological behaviour, e.g. some correlators might diverge. 
	
	As reviewed in section \ref{s:stringsonK3}, in the context of type IIB string theory compactified on a K3 surface, $\Gamma^{5,21}$ can be interpreted as the lattice of charges carried by string-like objects in the six dimensional theory, $O^+(\Gamma^{5,21})$ is the group of U-dualities, and $v$ is the charge of one of these strings, whose dynamics in the infrared (or near-horizon) limit, is described by an $\N=(4,4)$ SCFT of central charge $6n$. For example, the world-sheet dynamics of a system of $Q_1$ D1-branes and $Q_5$ D5-branes wrapping the K3 surface is described by a NLSM on a deformation of $\Sym^n(K3)$, with  $n=Q_1Q_5-1$. We denote by \be
	L_v\equiv L_n:=v^\perp \cap \Gamma^{5,21}\ee the sublattice of $\Gamma^{5,21}$ orthogonal to such a vector $v$. Since any two primitive vectors of the same length in $\Gamma^{5,21}$ are related by an automorphism in $O^+(\Gamma^{5,21})$ (see appendix \ref{a:lattices}), it follows that $L_v$ is uniquely determined by the integer $n$ up to isomorphisms (for this reason, we will sometimes use the notation $L_n$). In particular, $L_n$ is isomorphic to the orthogonal sum $\langle w\rangle\oplus \Gamma^{4,20}$ of a $1$-dimensional lattice $\langle w\rangle$ with a generator $w$ of squared norm $2-2n$ and the even unimodular lattice $\Gamma^{4,20}$ of signature $(4,20)$. 
	
	The moduli space $\M_n$ can be understood as a subset of the moduli space
	\be \M_{IIB}=O^+(\Gamma^{5,21})\backslash Gr^+(5,21)\ ,
	\ee  of type IIB string theory on a K3 surface  (see appendix \ref{s:stringsonK3}). Notice that $\M_{IIB}$ is a quotient of the Grassmannian parametrizing  positive definite $5$-dimensional subspaces $Z$ within $\Gamma^{5,21}\otimes\RR\cong \RR^{5,21}$, taken modulo automorphisms of the lattice $\Gamma^{5,21}$. The subspace $\M_n$ can be interpreted as the space of attractor moduli for a string of primitive charge $v\in \Gamma^{5,21}$, with $v^2=2n-2$. The attractor conditions constrain the space $Z$ to contain the vector $v$, so that the Grassmannian $Gr^+(4,21)$ parametrizes the four dimensional subspaces $\Pi=Z\cap v^\perp$ inside $\RR^{4,21}=v^\perp\cap \RR^{5,21}$ \cite{Dijkgraaf:1998gf}.
	
	Apart from the interpretation in terms of type IIB moduli, the space $\M_n$ parametrizes the possible hyperk\"ahler metric and B-field on a hyperk\"ahler manifold of $K3^{[n]}$-type, i.e. deformation equivalent to the Hilbert scheme of $n$-points in a K3 surface (see appendix \ref{a:Hilbertscheme}).
	
	The space $\M_n$ contains a sublocus $\M_n^{sym}$ given by the $n$-th symmetric orbifolds of  NLSM on K3, which has codimension $4$ and is isomorphic to the moduli space $\M_1$ of NLSM on K3.  $\M_1$ can be described as an open subset in the quotient of a Grassmannian
	\be  \M_1\subset O^+(\Gamma^{4,20})\backslash Gr^+(4,20)\ ,
	\ee obtained by excluding from the Grassmannian the points corresponding to $4$-dimensional subspaces $\Pi\subset \Gamma^{4,20}\otimes\RR\cong \RR^{4,20}$ that are orthogonal to a root $\Pi\subset r^\perp$, where $r\in \Gamma^{4,20}$, $r^2=-2$ \cite{Aspinwall:1996mn,Nahm:1999ps}. In other words, there is an injective map $\Sym^n:\M_1\to \M_n$ with image $\M_n^{sym}$. Clearly, $O^+(\Gamma^{4,20})$ is a subgroup of $O^+ (L_n)$. In section \ref{s:symmorbmoduli}, we will try to identify the locus $\M_n^{sym}$ within $\M_n$.

\subsection{The duality group}\label{s:dualities}

In this section, we analyze the group of dualities $O^+_n$ of SCFT of $K3^{[n]}$ type. The duality group must necessarily preserve the Zamolodchikov metric on the moduli space, and therefore it will be a subgroup of the orthogonal group $O(4,21)$. Furthermore, an analysis similar to \cite{Cheng:2016org}, section 3.2, shows that a duality in $O(4,21)$ preserves the world-sheet orientation of the SCFT if and only if it is contained in the $O^+(4,21)$ subgroup that preserves the orientation of $4$-dimensional subspaces in $\RR^{4,21}$. We will denote by $O_n$ the group of dualities of SCFTs of $K3^{[n]}$-type that might or might not preserve the world-sheet orientation, and by $O_n^+=O_n\cap O^+(4,21)$ the subgroup that does preserve the orientation. It is a matter of definitions whether two point in the moduli space related by a $O_n\setminus O^+_n$ transformation should be identified or not; in this article, we find it convenient to distinguish such points, so only the subgroup $O^+_n$ will be relevant for the definition of the moduli space.

As mentioned in the previous sections, these SCFTs describe the infrared dynamics of a stringy-like object of charge $v\in \Gamma^{5,21}$, with $v^2=2n-2$, in type IIB on K3 (or K3$\times S^1$).  The U-duality group of type IIB string theory on K3 is the group $O(\Gamma^{5,21})$ of automorphisms of the lattice $\Gamma^{5,21}$. This U-duality group contains a subgroup $Stab(v)$ that fixes the vector $v$. Dualities in the group $Stab(v)$ map an object of charge $v$ into itself, and preserve (setwise) the space $v^\perp\subset \Gamma^{5,21}\otimes \RR$ and therefore the Grassmannian $Gr(4,21)$ of positive definite $4$-dimensional subspaces inside $v^\perp\cong \RR^{4,21}$. Therefore, any two points in $Gr(4,21)$ related by a transformation in $Stab(v)$ must be dual to each other, i.e. they define equivalent superconformal field theories, so that $Stab(v)\subseteq O_n$. The subgroup of such dualities that preserve the world-sheet orientation is given by
\be Stab^+(v)=Stab(v)\cap O^+(4,21)\ ,
\ee and only this group will be included in $O^+_n$. 

Can the group $O^+_n$ be strictly larger than $Stab^+(v)$? Clearly, if a duality of the  SCFTs of $K3^{[n]}$-type outside of $Stab(v)$ extends to a duality of the full type IIB, then it must act by automorphisms on the lattice of charges $\Gamma^{5,21}$, so it must be an element of $O(\Gamma^{5,21})$. Furthermore, it must preserve the space $v^\perp \subset \Gamma^{5,21}\otimes \RR$, in order to have a well-defined action on the moduli space of SCFTs of $K3^{[n]}$-type. The only possibilities are that the duality either acts trivially on $v$ (but in this case, it would be in $Stab(v)$), or it acts by $v\mapsto -v$. The latter transformation, if extended to the whole type IIB string theory, would exchange the effective string we are considering with its charge conjugate. In fact, as explained in section \ref{s:symmorbmoduli}, the existence of such a duality follows from the existence of a certain symmetry at the points in the moduli space corresponding to symmetric orbifold $\Sym^n K3$. This symmetry acts non-trivially on the exactly marginal operators deforming the SCFT away from the symmetric orbifold locus. If two exactly marginal operators are related by such a symmetry, the corresponding deformed models must be dual to each other. In section \ref{s:symmorbmoduli}, we will argue that such a duality can only be extended to a $O(\Gamma^{5,21})$ transformation that maps $v$ into $-v$ (at least for $n>2$). So, it seems that this kind of dualities have to be included in $O^+_n$, which therefore must contain the group
\be Stab^+(\ZZ v):= \{h\in O(\Gamma^{5,21})\mid h(v)=\pm v\} \cap O^+(4,21)\ ,
\ee (as the notation suggests, this is the setwise stabilizer of the sublattice $\ZZ v\subset \Gamma^{5,21}$ in $O(\Gamma^{5,21})\cap O^+(4,21)$).  

Finally, let us discuss the possibility that $O^+_n$ contains some duality $h\in O^+(4,21)$ of the SCFTs that cannot be extended to an element in $O(\Gamma^{5,21})$. Recall that one of the arguments showing that $O(\Gamma^{5,21})$ is the full group of dualities of type IIB is that it is a maximal discrete subgroup of $O(5,21)$, so that any larger group of dualities would make the moduli space of type IIB on K3 non-Hausdorff \cite{Aspinwall:1995td,Aspinwall:1996mn}. In the case we are considering, the groups $Stab^+(v)$ and $Stab^+(\ZZ v)$ are subgroups of the group $O^+(L_v)$ of automorphisms of the lattice $L_v$. For generic $n$, however, they are strictly smaller than $O^+(L_v)$.\footnote{In particular $Stab^+(v)$ is a proper subgroup of $Stab^+(\ZZ v)$, and therefore of $O^+(L_v)$, for all $n>2$, while for $n=2$ one has $Stab^+(v)=Stab^+(\ZZ v)=O^+(L_v)$. The group $Stab^+(\ZZ v)$ is a proper subgroup of $O^+(L_v)$ whenever $v^2=2n-2$ is not a prime power. } Therefore, even requiring the moduli space of SCFTs of $K3^{[n]}$ type to be Hausdorff does not rule out the possibility that the duality group is a subgroup of $O^+(L_v)$ larger than $Stab^+(\ZZ v)$. It might be possible to exclude these kind of dualities by showing explicit examples (maybe of symmetric orbifolds) of models that are related by such automorphisms of $L_v$ but are not equivalent to each other. In the absence of a full proof in this sense, in the following we will keep in mind the possibility that the duality group $O^+_n$ of the SCFTs are subgroups of $O^+(L_v)$ larger than $Stab^+(\ZZ v)$. 

Could the group of dualities include also elements of $O^+(4,21)$ that are not automorphisms of the lattice $L_v$? While we were not able to prove this rigorously, it seems reasonable to expect that any such transformation would generate a group that is not discrete, and such that the moduli space is not Hausdorff. From a different perspective, recall that for a NLSM on a single K3 surface, the lattice $\Gamma^{4,20}$ can be interpreted as the lattice of R-R charges of D-branes, and can be constructed within the SCFT itself by looking at the couplings of R-R ground fields with $1/2$ BPS boundary states. One can consider boundary states also in SCFTs of $K3^{[n]}$ type; it is reasonable to expect $L_v$ to be the lattice of charges for a suitable class of boundary states. If this is the case, then a symmetry of the SCFT should preserve this lattice, and therefore act on it by automorphisms. These arguments seem to suggest that the duality group $O^+_n$ must be a subgroup of $O^+(L_v)$. In the following sections, we will assume that this is the case, and that $O^+_n$ is some intermediate subgroup between $Stab^+(\ZZ v)$ and $O^+(L_v)$
\be Stab^+(\ZZ v)\subseteq O^+_n\subseteq O^+(L_v)\ .
\ee
From the physics viewpoint, elements of $O^+_n\subseteq O^+(L_v)$  that are not contained in $Stab^+(\ZZ v)$ would correspond to `accidental' equivalences between $K3^{[n]}$-models, that do not descend from dualities of the full string theory. For definiteness, let us think about such an accidental equivalence as an element of $O(5,21,\RR)$ that fixes the vector $v$ and relates two distinct points of $\M_{IIB}$, both satisfying the attractor condition for a string of charge $v$. Then, these two points are physically inequivalent from the point of view of the full string theory, but  cannot be distinguished  just looking at the low energy dynamics of the string of charge $v$.

\subsection{Singular limits}\label{s:singular}

As explained in section \ref{s:stringsonK3}, the moduli space $\M_{IIB}$ of type IIB superstrings compactified on K3 is a Grassmannian of oriented positive definite $5$-planes $Z$ within $\Gamma^{5,21}\otimes \RR$, modulo $O^+(\Gamma^{5,21})$. For each point in $\M_{IIB}$, one has an orthogonal decomposition $v\equiv (v_+,v_-)$ of any vector $v\in \Gamma^{5,21}$ into a (positive norm) component $v_+$ along $Z$ and a (negative norm) component $v_-$ along $Z^\perp$. Let $v\in \Gamma^{5,21}$ be a primitive vector with $v^2=2n-2> 0$, representing the charge of some BPS object (the case $v^2=0$ is a bit peculiar, so let us not consider it here). Then, the attractor manifold for an object of charge $v$ is the locus of subspaces $Z $ containing $v$, so that $v=(v_+,0)$ in the decomposition above. This locus can be identified with the moduli space of non-linear sigma models on hyperk\"ahler manifolds of $K3^{[n]}$ type, i.e. deformations of the Hilbert scheme of $n$ points on the K3 manifold $X$ (or the moduli space of sheaves on $X$ with Mukai vector $v$). More precisely,  a point $Z$ in the attractor manifold of $v$ determines an oriented positive definite $4$-plane $\Pi:=Z\cap v^\perp$ in $v^\perp \cong L_v\otimes \RR$, where $L_v=v^\perp \cap \Gamma^{5,21}$. 

Seiberg and Witten \cite{Seiberg:1999xz} described the conditions on the attractor point $Z$ such that the corresponding SCFT be singular. This only happens when the BPS object of charge $v$ can decay into two BPS objects of charges $v'$ and $v''$.  Formally, this translates into the existence of two vectors $v',v''\in \Gamma^{5,21}$, that are not proportional to $v$ (otherwise $v$ would not be primitive), and such that
\begin{align}
&(v')^2,(v'')^2\ge -2\ ,\label{cond1}\\
&v=v'+v''\ ,\label{cond2}\\
&|v_+|=|v'_+|+|v''_+|\ .\label{cond3}
\end{align} (note that $v_+=v$, because we are at the attractor manifold). Condition \eqref{cond2} is just charge conservation. The mass of a BPS object of charge $v$ is proportional to the length $|v_+|$ of its component along $Z$. For $v=v'+v''$, one has $|v_+|\le |v'_+|+|v''_+|$ and the decay can occur only when the equality is satisfied, which is condition \eqref{cond3}. The condition \eqref{cond1} ensures the existence of BPS objects with charges $v'$ and $v''$. 

To be precise, the condition \eqref{cond1} is slightly different from the condition $(v')^2,(v'')^2\ge 0$ given in \cite{Seiberg:1999xz}, which seems to be appropriate only for compactifications on a torus $T^4$. For compactifications on K3, there are BPS objects carrying charge with square norm $-2$: the easiest such example is a single D5-brane wrapping the K3 manifold, with carries $Q_5=+1$ five brane charge a geometrically induced $Q_1=-1$ one-brane charge, so that $v^2=2Q_1Q_5=-2$.  Notice that the near horizon geometry of such objects is rather complicated \cite{Johnson:1999qt}. It seems appropriate to exclude the points in the moduli space where the system of branes that we are considering can decay into one of these objects (plus some other decay product).  The conditions \eqref{cond1}--\eqref{cond3} also lead us to exclude from the moduli space the points where some non-perturbative (at the attractor point) strings become tensionless \cite{Witten:1995zh}. This phenomenon can occur, for example, when the K3 surface contains a holomorphic $2$-cycle with self-intersection $-2$ (a $\PP^1$) whose volume is shrunk to a point, and with no B-field flux through it. A D3-brane wrapping such a cycle becomes tensionless at such a point in the moduli space. It is known that the SCFT describing the worldsheet dynamics of a fundamental type IIB string on K3, which is a NLSM on a single copy of K3, becomes singular at this point in the moduli space, due to the presence of such tensionless non-perturbative strings \cite{Witten:1995zh,Aspinwall:1996mn}. It is reasonable to expect that, at such a point in the moduli space, the same will happen for the superconformal field theory of type $K3^{[n]}$ describing the infrared dynamics of an effective string.

It is sometimes useful to repackage the conditions \eqref{cond1}--\eqref{cond3} in a different way.

Let us first consider the case where $(v')^2,(v'')^2\ge 0$. The conditions \eqref{cond2},\eqref{cond3} imply that $v'_+=\alpha v$ and $v''_+=(1-\alpha)v$, with $0< \alpha< 1$, and $v'_-=-v''_-\neq 0$ (if $v'_-=0$ then $v'$ and $v''$ would be proportional to $v$). Up to an exchange of $v'$ and $v''$, we can assume that $\alpha\le 1/2$, so that $|v'_+|\le | v''_+|$ and $(v')^2\le (v'')^2$ ($(v')^2=(v'_+)^2-(v'_-)^2= (v'_+)^2-(v''_-)^2\le (v''_+)^2-(v''_-)^2=(v'')^2$). Notice that
\be (v,v')=\alpha v^2=\alpha(2n-2)\in \ZZ_{> 0}\ ,
\ee because $(v,v')\in \ZZ$ for any two vectors $v,v'\in \Gamma^{5,21}$. This also implies that $(2n-2)v_-=(2n-2)v'-\alpha(2n-2)v\in \Gamma^{5,21}$.
The sublattice spanned by $v,v'$ has  signature $(1,1)$ (it contains a negative direction corresponding to $v_-$), so that \be \det \left(\begin{smallmatrix} v^2 & (v,v')\\ (v,v') & (v')^2\end{smallmatrix}\right)<0\qquad \Leftrightarrow\qquad 
(v,v')^2>(v')^2 v^2 \ .\ee Thus, the conditions \eqref{cond1}, \eqref{cond2}, \eqref{cond3} imply
\be 0\le (v')^2v^2<(v,v')^2\le \left(\frac{v^2}{2}\right)^2\ ,\label{condeq}
\ee where the last inequality follows from $\alpha=\frac{(v,v')}{v^2}\le 1/2$. Vice versa, suppose that there is $v'\in \Gamma^{5,21}$, such that \eqref{condeq} holds. Up to an exchange $v'\leftrightarrow -v'$, we can assume that $(v,v')\ge 0$. Eq. \eqref{condeq} implies that the primitive sublattice $M\subset \Gamma^{5,21}$ of rank $2$ containing both $v$ and $v'$ has signature $(1,1)$. Let $x\in M$ be a generator of $v^\perp\cap M$, so that $x^2<0$. If $Z\perp x$, then $v'$ and $v'':=v-v'$ satisfy all conditions \eqref{cond1}, \eqref{cond2}, \eqref{cond3}. Indeed, $(v')^2\ge 0$ follows from \eqref{condeq} and $v^2\ge 0$. Furthermore, $v'$ and $v''$ can be written as $v'=\alpha v+\beta x$ and $v''=(1-\alpha)v-\beta x$, with $\alpha= \frac{(v,v')}{v^2}$ and $\beta=\frac{(x,v')}{x^2}$, and by \eqref{condeq} and the condition $(v,v')\ge 0$ one has $0\le \alpha \le 1/2$. Thus, $(v'')^2\ge (v')^2\ge 0$, and $v'_+=\alpha v$, $v''_+=(1-\alpha)v$, with $\alpha,(1-\alpha)\ge 0$, so that also \eqref{cond3} holds.

Notice that $x\in v^\perp\cap \Gamma^{5,21} =L_v$, and that the condition $Z\perp x$ is equivalent to $\Pi\perp x$. 

Let us now consider the case where one of the charges $v'$ or $v''$ has squared norm $-2$. Up to an exchange, we can suppose that $(v')^2=-2$. In this case, the primitive  sublattice $M$ containing both $x$ and $v$ must contain $v'$ satisfying  
\be\label{condeq2} (v')^2=-2\qquad \text{and}\qquad 0\le (v',v)\le \frac{v^2}{2}\ .
\ee The latter equations are equivalent to \eqref{cond1}, \eqref{cond2} and \eqref{cond3} for $v'$ and $v''=v-v'$ with $(v')^2=-2$.

Assuming that the conditions in \cite{Seiberg:1999xz} are necessary and sufficient for singularity, we conclude that the SCFT corresponding to $\Pi$ is singular if and only if $\Pi$ is orthogonal to $x\in L_v$, such that the primitive sublattice $M\subset \Gamma^{5,21}$ contains a vector $v'$ satisfying either \eqref{condeq} (case $(v')^2\ge 0$) or \eqref{condeq2}.

The conditions \eqref{condeq} and \eqref{condeq2} are very similar to the condition describing the K\"ahler cone of a hyperkahler manifold $X$ of type $K3^{[n]}$ \cite{Mongardi2015,Mongardi2016}. Recall that the lattice $H^2(X,\ZZ)$ is isomorphic to the orthogonal complement in $\Gamma^{4,20}$ of a primitive vector $v$ of norm $2n-2$. The K\"ahler cone, i.e. the region of $H^2(X,\RR)$ spanned by the K\"ahler classes of $X$, is delimited by some walls, which are the hyperplanes orthogonal to certain`wall divisors' $D\in H^2(X,\ZZ)$.  Then, $D\in H^2(X,\ZZ)$ is a wall divisor if the primitive lattice of rank $2$ $T_D\subset \Gamma^{4,20}$ containing both $v$ and $D$ has a generator $r_D$ satisfying either 
\be\label{geomsing1} (r_D)^2=-2\qquad \text{and}\qquad 0\le (r_D,v)\le \frac{v^2}{2}\ .
\ee or
\be\label{geomsing2} 0<(r_D)^2v^2<(v,r_D)^2\le \left(\frac{v^2}{2}\right)^2\ ,
\ee which have the same form as the conditions \eqref{condeq} and \eqref{condeq2}. This is not completely surprising. In the limit where a K\"ahler class reaches the boundary of the K\"ahler cone, one of the holomorphic cycles shrinks to zero volume, and the geometry of the hyperk\"ahler manifold becomes singular (for example, an orbifold). In general, a non-linear sigma model whose target space is such a singular geometry can be perfectly well-defined -- orbifolds are the main examples of these phenomena. On the other hand, what usually happens is that this consistent  orbifold has some fixed non-zero B-field. If we turn off the B-field, then one  expects the non-linear sigma model to become singular as well. Therefore, for zero B-field there should be a precise correspondence between singular geometries and singular CFT.  The similarity of the equations \eqref{geomsing1}--\eqref{geomsing2} with \eqref{condeq}--\eqref{condeq2}, therefore, gives further support to our analysis of singularities.

\bigskip

On the other hand, as we will see in section \ref{s:symmorbmoduli}, the analysis of this section seems to be in contradiction with the properties of symmetric orbifolds. In section \ref{s:symmorbmoduli} we will try to determine the locus of symmetric orbifolds models within the moduli space $\M_n$, for all $n$. The outcome is puzzling: all symmetric orbifolds occur at points in the moduli space that, according to the analysis above, should be singular! Since we are not sure what the resolution of this puzzle is, the results of the present section cannot be completely trusted.

\subsection{The symmetric orbifold locus}\label{s:symmorbmoduli}

In this section, we will discuss the locus $\M^{sym}_n\subset \M_m$ in the Hilbert space of SCFT of  $K3^{[n]}$ type corresponding to symmetric orbifolds.

If $\C_S$ denotes the NLSM on the K3 surface $S$, corresponding to a certain choice of hyperk\"ahler structure and B-field, then the symmetric orbifold $\Sym^n\C_S$ (where orbifold is understood in a CFT sense) is a NLSM on a hyperk\"ahler manifold $X$ of K3$^{[n]}$ type. In particular, one expects the geometry of the target space to be the one of  the symmetric orbifold (in a geometric sense) $\Sym^n S$, i.e. with $\alpha_1=\alpha_2=\lambda=0$ in the notation of appendix \ref{a:Hilbertscheme} 

The space of states of the orbifold CFT $\Sym^n\C_S$ is the sum over $[\sigma]$-twisted sectors for each conjugacy class $[\sigma]$ of $S_n$; the untwisted sector corresponds to $[\sigma]$ being the identity class. In turn, each conjugacy class for the symmetric group $S_n$ can be labeled by the cycle shape  $\prod n_i$, with $\sum_i n_i=n$ of the class. There are $84$ exactly marginal operators in $\Sym^n\C_S$, preserving the $\N=(4,4)$ superconformal algebra, corresponding to deformations of the CFT in the $84$ dimensional moduli space of NLSM of K3$^{[n]}$ type. Out of these $84$ deformations, $80$ are contained in the untwisted sector and can be naturally identified with deformations of the `seed' CFT $\C_S$ -- they are literally obtained by symmetrization $\sum_{\sigma\in S_n}\sigma(\chi\otimes1\otimes \ldots\otimes 1)$ of the exactly marginal operators $\chi$ on $\C_S$. The remaining $4$ moduli are contained in a twisted sector. Recall that the $\N=(4,4)$ preserving exactly marginal operators are fields of conformal weight $(1,1)$ and neutral under $SU(2)_L\times SU(2)_R$ R-symmetry, that are obtained as supersymmetric descendants of NS-NS fields with conformal weights $(h,\bar h)=(1/2,1/2)$ and R-symmetry charges $(q,\bar q)=(1/2,1/2)$ (we normalize the charges so that they are half-integral).  It was shown in \cite{Lunin:2001pw}   that the lowest chiral operators (i.e. with $h=q$) in the sector corresponding to a cycle shape $\prod_i n_i$ has
\be h=q=\sum_i \frac{n_i-1}{2}\ .
\ee It is clear then that fields with $h=q=1/2$ and $\bar h=\bar q=1/2$ can only occur in the untwisted sector or in the $[\sigma]$-twisted sector where $\sigma$ is a single transposition (i.e. with cycle shape $1^{n-2}2^1$). In particular, there is a unique such field in the $[\sigma]$-twisted sector, and it coincides with the $[\sigma]$-twisted ground field.\footnote{For general $[\sigma]$ the $[\sigma]$-twisted ground field is usually not chiral.} This field has four supersymmetric descendants of weights $(1,1)$ and neutral under R-symmetry, which are the twisted sector exactly marginal operators.
Therefore, the $4$ marginal operators not associated with deformations of the seed theory come from the $[\sigma]$-twisted sector where $[\sigma]$ is the class of a single transposition, and can be identified with deformations of the remaining moduli $\alpha_1,\alpha_2,\lambda,\beta$ of the NLSM on $X$ (see appendix \ref{a:Hilbertscheme}). 

Notice that every symmetric orbifold has a $\ZZ_2$ symmetry (that we will call a \emph{quantum symmetry}) acting by multiplication by $(-1)^{\sgn(\sigma)}$ on the $[\sigma]$-twisted sector. Indeed, the OPE of a $[\sigma_1]$- and a $[\sigma_2]$-twisted fields can contains a field in the $[\sigma_3]$-twisted sector only if $\sigma_3$ is conjugate to $\sigma_1h\sigma_2h^{-1}$ for some $h\in S_n$; then it is clear that $\sgn(\sigma_3)=\sgn(\sigma_1h\sigma_2h^{-1})=\sgn(\sigma_1)\sgn(\sigma_2)$, so that the symmetry preserves the OPE. This $\ZZ_2$ symmetry acts trivially on the $80$ marginal operators in the untwisted sector and changes the sign of the $4$ twisted sector marginal operators, since they are contained in a $[\sigma]$-twisted sector with $\sigma$ an odd permutation. 
	

The presence of the quantum symmetry will be our main tool in trying to identify the locus $\M_{n}^{sym}\subset \M_n$. Suppose that a certain $K3^{[n]}$ model has a symmetry $t$ acting on exact marginal operators in the same way as the quantum symmetry of the symmetric orbifold. This means that, with respect to the action of $t\in O(21)\subset  SO(4)\times O(21)\subset O^+(4,21)$, the $(4,21)$ representation of $SO(4)\times O(21)$ splits as the sum of a $4$-dimensional subspace (a four-dimensional representation of $SO(4)$) with $t$-eigenvalue $-1$ and an $80$-dimensional subspace ($20$ four-dimensional representations of $SO(4)$) with eigenvalue $+1$. As a consequence, in the $4\oplus 21$ representation of $SO(4)\times O(21)$ (i.e. on $\RR^{4,21}\equiv L_v\otimes \RR$), $t$ has a single eigenvector with eigenvalue $-1$ in the $21$-dimensional component, and fixes the  $24$-dimensional  orthogonal complement. Let $s\in \RR^{21}\subset \RR^{4,21}$ be a generator of the $-1$ eigenspace (this implies that $s^2<0$). Then, $t$ acts on any $x\in \RR^{4,21}$ by a reflection with respect to the hyperplane $s^\perp$
\be x\mapsto t(x)=x-\frac{2(s,x)}{(s,s)} s\ .
\ee Notice that this formula is invariant under rescalings of $s$, $s\to \alpha s$, $\alpha\in \RR\setminus\{0\}$. It is easy to check that the transformation $t$ is in $O(4,21)$. In order to be an automorphism of the lattice $L_v$, it must be that $x-t(x)\in L_v$ for every $x\in L_v$. This implies, in particular, that $s$ is proportional to a lattice vector; without loss of generality, we can rescale it so that $s$ is a primitive vector in $L_v$. With this choice, the condition for $t$ to be a lattice automorphism is
\be\label{condint} \frac{2(s,x)}{(s,s)}\in \ZZ\qquad \forall x\in L_v\ .
\ee
Let $\Div(s)$ be the greatest common divisor of $(s,x)$, $x\in L_v$; for the lattice $L_v$, the possible values of $\Div(s)$ are the divisors of $2n-2$ (see appendix \ref{a:lattices}). Then, the condition \eqref{condint} can be written as $(s,s)|2\Div(s)$. On the other hand, we also have $\Div(s)|(s,s)$, from which we find the two possibilities $\Div(s)=-(s,s)$ or $\Div(s)=-(s,s)/2$ (note that $(s,s)<0$ and $\Div(s)>0$).

In order to identify the vector $s$, we need some additional information. The  locus $\F_t$ of $K3^{[n]}$-models with symmetry the reflection $t$ is an $80$-dimensional connected orbifold: it is given by the quotient $\F_t=H_s\backslash Gr_{s^\perp}$ of the Grassmannian $Gr_{s^\perp}\cong Gr^+(4,20)$ of four dimensional oriented positive definite subspaces in $s^\perp\cong \RR^{4,20}$, quotiented by the subgroup $H_s$ of the duality group that fixes $s$ up to a sign. On the other hand, the locus $\M^{sym}_n$ of symmetric orbifolds is the image of an isometry $\M_1\to \M_n$, where  $\M_1=O^+(\Gamma^{4,20})\backslash Gr^+(4,20)$. This isometry can be `lifted' at the level of Teichm\"uller spaces: there exists an isometric embedding $i:Gr^+(4,20)\to Gr^+(4,21)$ such that $i(Gr^+(4,20))$ maps to $\M_n^{sym}$ under the projection $Gr^+(4,21)\to Gr^+(4,21)/O^+_n$. Such an embedding is not unique, but the different choices are related to each other by dualities. The image $i(Gr^+(4,20))$ forms a single orbit under the action of a subgroup $SO^+(4,20,\RR)\subset SO^+(4,21,\RR)$ acting (not faithfully) on $Gr^+(4,21)$ by isometries. Let us choose a point in $i(Gr^+(4,20))$, corresponding to a four-plane $\Pi\subset \RR^{4,21}$ and to a non-singular symmetric orbifold in $\M_n^{sym}$. This symmetric orbifold has the symmetry $t$, which is unbroken by infinitesimal deformations of the model within $\M_n^{sym}$. This means that $\Pi$, and each of the four-planes in a neighborhood of $\Pi$ within $i(Gr^+(4,20))$,  must be orthogonal to some vector in $O^+_n\cdot s$,  the $O^+_n$-orbit of $s\in L_v$. If $\Pi\in i(Gr^+(4,20))$ is generic enough, there cannot be more than one vector $s'\in O^+_n\cdot s$ satisfying $s'\perp \Pi$ (the Grassmannian of four planes orthogonal to two or more vectors in $O^+_n\cdot s$ has dimension smaller than $80$), and by continuity all four-planes in a neighborhood of $\Pi$ must be orthogonal to the same $s'$. We can choose the embedding $i$ in such a way that $s'=s$. As a consequence, the infinitesimal action of the subgroup $SO^+(4,20)\subset SO^+(4,21)$ generating $i(Gr^+(4,20))$ preserves setwise the subspace $s^\perp\subset \RR^{4,21}$. This condition is sufficient to identify the subalgebra of generators of $SO^+(4,20)$ within the Lie algebra of $SO^+(4,21)$, and therefore the group $SO^+(4,20)$ itself as the    subgroup of $SO^+(4,21)$ that acts trivially on  $s$. Therefore,  $i(Gr^+(4,20))$ is exactly the Grassmannian $Gr_{s^\perp})$ of subspaces $\Pi\subset \RR^{4,21}$ orthogonal to $s$. This means that also the images of $i(Gr^+(4,20))=Gr_{s^\perp}$ under $Gr^+(4,21)\to O^+_n\backslash Gr^+(4,21)$ must be the same, and that the discrete groups we quotient by are isomorphic $H_s\cong O^+(\Gamma^{4,20})$.

It is easy to check that this isomorphism of discrete groups holds only if $\Div(s)=2n-2$. Indeed, $H_s$ is a group of automorphisms of the sublattice $s^\perp \cap L_v$. Up to dualities, one can always take both $v$ and $s$ to be contained in some $\Gamma^{2,2}\subset \Gamma^{5,21}$; therefore, $s^\perp \cap L_v$ contains a copy of $\Gamma^{3,19}$ as a primitive sublattice. By acting by $H_s\cong O^+(\Gamma^{4,20})$ on this $\Gamma^{3,19}$, one obtains a lattice isomorphic to $\Gamma^{4,20}$, which must be contained in $s^\perp \cap L_v$. But $s^\perp \cap L_v$ has signature $(4,20)$ and is even, and these conditions imply that $s^\perp \cap L_v\cong \Gamma^{4,20}$. One has the lattice inclusions $\Gamma^{4,20}\oplus \langle s\rangle \subseteq L_v\subseteq (\Gamma^{4,20}\oplus \langle s\rangle )^*$, but $(\Gamma^{4,20}\oplus \langle s\rangle )^*=\Gamma^{4,20}\oplus \langle \frac{s}{\Div(s)}\rangle$, and since $s$ is primitive we must have the equality $L_v=\Gamma^{4,20}\oplus \langle s\rangle$. This isomorphism implies that $\Div(s)=-(s,s)=2n-2$. 

We conclude that the symmetric orbifold locus $\M_n^{sym}$ corresponds to the Grassmannian of $4$-dimensional subspaces orthogonal to a vector $s$ with $\Div(s)=-(s,s)=2n-2$. As we explain below, for generic $n$ this analysis is not sufficient to identify the vector $s$ (and therefore $\M_n^{sym}$) up to dualities. However, it is enough to reach a puzzling conclusion: for any $s\in L_v$ with $\Div(s)=-(s,s)=2n-2$, a four-dimensional subspace $\Pi$  orthogonal to $s$ corresponds to a singular model, according to the arguments of \cite{Seiberg:1999xz} and of section \ref{s:singular}. Indeed, since $\langle s,v\rangle^\perp \cap \Gamma^{5,21}\cong s^\perp\cap L_v\cong  \Gamma^{4,20}$, the vectors $s$ and $v$ are contained in a sublattice $\Gamma^{1,1}\subset \Gamma^{5,21}$. Let $u,u^*$ be generators of this $\Gamma^{1,1}$, such that $(u,u)=(u^*,u^*)=0$, $(u,u^*)=1$. Then, $s$ and $v$ can be written as
\be v=au+bu^*\ ,\qquad \qquad  s=-bu+au^*\ ,
\ee for some $a,b\in \ZZ$ coprime and such that $ab=2n-2$ (in fact, the different duality orbits of vectors $s$ are labeled by this pair of integers). We can always choose $u,u^*$ in such a way that $a,b>0$. Then, whenever $Z=\langle v\rangle\oplus \Pi$ is orthogonal to $s$, we have that $v':=u$ and $v'':=(a-1)u+bu^*$ satisfy the conditions \eqref{cond1}--\eqref{cond3} (and,  in fact, even the stricter conditions in \cite{Seiberg:1999xz}). There are a few possible resolutions of this puzzle:
\begin{itemize}
	\item The symmetric orbifold model is not a well defined CFT. This seems quite a radical departure from what we know from conformal field theory and string theory.
	\item We have misidentified the locus of symmetric orbifolds $\M_n^{sym}$. This means that there must be a hole in one of the arguments we used.  One of the ingredients is the existence of the quantum symmetry of symmetric orbifolds. This is just based on the standard properties of non-abelian orbifolds, so it seems quite a safe statement, although  there might be some potential subtleties with the definition of the symmetry in sectors other than the NS-NS. Another ingredient we need is the fact that two models obtained by deforming the symmetric orbifold by exactly marginal operators related  by a symmetry are dual to each other. This should be true, provided that conformal perturbation theory correctly reproduces the correlators of the deformed models in some neighborhood of each point in the moduli space. This is clearly a delicate assumption, but without this assumption the whole construction of the moduli space $\M_n$ should be reconsidered. One of the most delicate ingredients is the assumption that the group of dualities acts by automorphisms of the lattice $L_v$. In section \ref{s:dualities} we discussed why it is unlikely that the duality group is larger. It would be nice to have a rigorous argument showing that this assumption follows from the requirement that the moduli space $\M_n$ be Hausdorff. Notice that, for the arguments of this section to work, we don't need to assume that the duality group is a \emph{proper} subgroup of $O^+(L_v)$.
	\item We have misidentified the locus of singular models in $Gr(4,21)$. This seems to us as the most likely possibility. The arguments described in section \ref{s:singular} are certainly reasonable, but do not seem to be conclusive. In particular, these arguments are not formulated within the framework of the conformal field theory itself, but follow from analysing the dynamics of the system of branes, which is described by the $K3^{[n]}$ model only in the low energy limit. It could in principle happen that, in the limit where we reach a singular point in the moduli space, the $K3^{[n]}$-model is still a well-defined CFT, even if it is not a good description of the low energy dynamics of the brane system anymore. Another (related) possibility is that, in some occasions, the additional `non-compact' degrees of freedom that appear at the singular point are decoupled from the SCFT itself, so that the latter remains consistent. In order to confirm or to rule out similar scenarios, one would need a more detailed understanding of the dynamics of the brane system at singular point.  Unfortunately, a precise description of the set of singular models is needed for the classification of symmetries of SCFTs of $K3^{[n]}$ type that we will attempt in the following sections. This means that our results will be conditional to a long list of assumptions and restrictions. We hope to be able to drop some of these assumptions in future works.
\end{itemize}

Finally, let us show that our analysis is not sufficient to unambiguously determine $s$ up to dualities. Indeed, while there is a unique orbit of vectors $s$ with $ \Div(s)=-(s,s)=2n-2$ with respect to the full automorphism group $O(L_v)$ of the lattice $L_v$, the duality group $O_n^+$ is the (generally, proper)  subgroup of $O(L_v)$ acting by $\pm 1$ on the discriminant group $A_{L_v}:=L_v^*/L_v$. By Eichler principle (see appendix \ref{a:lattices}), since $L_v$ contains a $\Gamma^{2,2}$ sublattice, then any two primitive vectors $s,s'$  are related by the group $O^0(L_v)$ acting trivially on $L_v^*/L_v$ if and only if they have the same norm, the same divisor, and $\frac{s}{\Div(s)}\equiv \frac{s'}{\Div(s')}\mod L_v$. In particular, we have $L_v^*/L_v\cong \ZZ_{2n-2}$, and we are interested in vectors $s$ with $\Div(s)=-(s,s)=2n-2$, so that $\frac{s}{\Div(s)}$ has norm $\frac{1}{2n-2} \mod 2\ZZ$ and order $2n-2$ in $L_v^*/L_v$, so it is a generator.  On the other hand, in general, there are several different generators with the same norm and order in $L_v^*/L_v$, and one can show that all of them can be lifted to elements of the form $\frac{\tilde s}{2n-2}$, where $\tilde s$ has $\Div(\tilde s)=-(\tilde s,\tilde s)=2n-2$.\footnote{This is because, for a lattice containing a copy of $\Gamma^{1,1}$, the natural map $O(L_v)\to O(L_v^*/L_v)$ is surjective. Here, $O(L_v^*/L_v)$ is the group of automorphisms of $L_v^*/L_v\cong \ZZ_{2n-2}$ preserving the quadratic form (defined modulo $2\ZZ$). Every two generators of $\ZZ_{2n-2}$ with norm $\frac{1}{2n-2}$ are related by some element of $O(L_v^*/L_v)$, and this element lifts to some automorphism in $O(L_v)$. The latter automorphism maps $s$ to a vector $\tilde s$ with the same norm and divisor.}  All such $\tilde s$ are in the same $O(L_v)$-orbit, but are not related by $O^0(L_v)$ transformations.

\section{Symmetries and twining genera}\label{s:mainres}
	
In this section we will discuss the classification (up to isomorphisms) of the possible groups of symmetries of a (non-singular) $\N=(4,4)$ superconformal field theories with central charge $c=6n$, $n>1$, in the moduli space non-singular NLSM on hyperk\"ahler manifolds of $K3^{[n]}$ type. The results and the methods are very similar to the ones relevant for non-linear sigma models on K3 \cite{K3symm}.

\subsection{A (restricted) classification of symmetries}\label{s:callifsymm}

Let us consider $v\in \Gamma^{5,21}$, with $v^2=2n-2$, and suppose that the five plane $Z$ is at an attractor point for $v$, i.e. $v\in Z$. As usual, let $L_n\equiv L_v:=v^\perp \cap \Gamma^{5,21}$ and $\Pi=Z\cap v^\perp$, so that $\Pi$ is a positive definite four plane in $L_n\otimes\RR$ and determines a point $\C$ in the moduli space $\M_n$ of NLSM on $K3^{[n]}$. We assume that this point is non-singular, in the sense of section \ref{s:singular}, so that we expect the corresponding CFT to be well defined. Our goal is to determine the possible group of symmetries of the CFT $\C$.

 Let us try to be more precise. By symmetries of a conformal field theory we mean a linear transformation of the fields that preserves the OPE, fixes the vacuum vector and the stress-energy tensor. We will focus on symmetries that satisfy some additional properties:
\begin{itemize}
	\item[(1)] They commute with the full $\N=(4,4)$ superconformal algebra (and not only the Virasoro algebra).
	\item[(2)] They commute with left- and right-moving spectral flow isomorphism relating the Neveu-Schwarz with the Ramond sector of the theory;
\end{itemize}
While there are certainly interesting symmetries that \emph{do not} satisfy these conditions, there are some good reasons to require them. The first reason is practical: without these conditions, the classification problem is much more complicated, and in fact it is not solved even in the case of non-linear sigma models on K3. Secondly,  these properties assure that the action on the low-energy six dimensional type IIB string compactification preserves space-time supersymmetry. Finally, these conditions are the exact analogue of the ones required in \cite{K3symm} for NLSM on K3 and in \cite{Volpato:2014zla} for NLSM on $T^4$.

Let $G$ be the group of symmetries of the CFT $\C$ satisfying the conditions (1) and (2) above. Then, $G$ has an action on the $84$ dimensional space of exactly marginal operators of the model, which can be identified with the tangent space to the Grassmannian $Gr^+(4,21)$ at the point corresponding to $\C$.\footnote{To be precise, $\C$ determines a point in the quotient $Gr^+(4,21)/O_n^+$, and one has to choose a lift of this point to $Gr^+(4,21)$. What follows is independent of this choice.} These fields are suitable supersymmetric descendants of the $84$-dimensional space of superconformal primaries of weights $(1/2,1/2)$, which decomposes into $21$ representations of the superconformal algebra.  By (1), the symmetry $G$ will act by a $O(21)$ transformation on these $21$ representations, i.e. the group of $O(4,21)$ transformations that leave the positive definite $4$-dimensional subspace $\Pi$ point-wise fixed. If two exactly marginal operators are related by a symmetry $g$, then the models obtained by the corresponding deformations will be equivalent, i.e. related by a duality. This means that the action on the $84$-dimensional space must be induced by a transformation in the duality group $O^+_n$. We conclude that there must be a homomorphism $G\to O^+_n\cap O(21)$. This homomorphism must be surjective: every duality in $O^+_n$ fixing the $4$-dimensional space $\Pi$ will map  the theory $\C$ into itself in a non-trivial way (as can be seen from the action on the exactly marginal operators), so it must lift to a symmetry of the theory satisfying (1) and (2). One can also argue that the homomorphism is injective. Indeed, any element $g$ in the kernel would act trivially on the whole space of exactly marginal operators of the theory. As a consequence, the kernel of this homomorphism is preserved by any deformation, and since the moduli space is connected, it must be a symmetry of every SCFT in the moduli space. Therefore, in order to exclude a non-trivial kernel for every model in the moduli space, it is sufficient to choose one $K3^{[n]}$-model $\C$ and show that $\C$ has no symmetries satisfying (1) and (2) and acting trivially on all exactly marginal operators. A suitable model for this analysis is given by $\C=\Sym^n \C'$, where $\C'$ is the NLSM on K3 with `the largest symmetry group', described in \cite{Gaberdiel:2013psa}. The details are relegated to appendix \ref{a:nokernel}. We conclude that the group of symmetries $G$ satisfying (1) and (2) is isomorphic to $O(21)\cap O_n^+$, i.e. it is the subgroup of $O_n^+\subset O(4,21)$ fixing pointwise the $4$-dimensional subspace $\Pi\subset \RR^{4,21}$.

\bigskip

As discussed in section \ref{s:symmorbmoduli}, the group $O^+_n\subseteq O^+(L_n)$ contains a normal subgroup $O^{0+}(L_n)$ of $O^+(L_n)$ fixing the discriminant group $L_n^*/L_n$. This group can be identified with  the subgroup  $Stab^+(v)$ of $O(\Gamma^{5,21})$ fixing the vector $v$.   For the purpose of the classification, it is convenient to restrict ourselves to symmetries that satisfy (1), (2) and
\begin{itemize}
	\item[(3)] They are contained in the subgroup $O^{0+}(L_n)$ of $O_n^+$. 
\end{itemize}  Physically, the symmetries in the subgroup $Stab^+(v)$ of $O^+_n$ are the ones that lift to dualities of the full type IIB string theory and that fix the charge $v$ of the string-like object we are considering. This is quite a natural restriction from a space-time viewpoint. One can expect $O^+_n$ to be strictly larger than $Stab^+(v)$. In particular, for $n>2$, the set $O^+_n\setminus Stab^+(v)$ should contain elements in $O(\Gamma^{5,21})$ that flip the sign of $v$: if our analysis of section \ref{s:symmorbmoduli} is correct, then for $n>2$ every symmetric product $\Sym^n \C'$ of a NLMS $\C'$ on a K3 surface has a self-duality in this set, which in turn satisfies (1) and (2). 
 From the perspective of the non-linear sigma model,  it would be interesting to consider the larger group of symmetries where (3) is not necessarily satisfied. We will not attempt a classification of these groups in this work.

\bigskip
 We will also put some restriction on the set of models we will consider. As argued in section \ref{s:singular}, the set $S_n^{roots}\subset O_n^+\backslash Gr^+(4,21)$ corresponding to $4$-planes $\Pi$ that are orthogonal to some root $r\in \Gamma^{5,21}$, $r^2=-2$, should be  contained in the set of singular models, and therefore should be excluded from the moduli space $\M_n$. On the other hand, the analysis of symmetric orbifolds (in particular, for $n=2$), seems to be in contradiction with this conclusion. For this reason, it is preferable to take an agnostic view  on this argument and explicitly exclude such models from our classification. This restriction greatly simplifies the classification problem we are considering -- in fact, we were not able to generalize this classification so as to include these models. 

\bigskip

Finally, we will assume that the set of singular models is not larger than the set $S$ described in section \ref{s:singular}. This assumption will be needed for the second part of theorem \ref{th:main}, in order to show the existence of a non-singular model with a given symmetry group $G$.\footnote{The assumption  was also used to argue that the moduli space $\M_n$ is connected, and show that there is no symmetry acting trivially on all exactly marginal deformations.  However, the assumption about connectedness of $\M_n$ is much weaker than the one we need for the second part of theorem \ref{th:main}.} 

Our discussion leads us to conclude:

\begin{claim}
	Let $\C_\Pi$ be a non-singular SCFT of $K3^{[n]}$ type, $n\ge 1$, corresponding to a $4$-subspace $\Pi\subset L_n\otimes \RR\cong \RR^{4,21}$,  such that $\Pi$ is not orthogonal to any $r\in L_n$, $r^2=-2$. The group of symmetries $G_\Pi$ satisfying the conditions (1), (2), and (3) above is isomorphic to the subgroup of $O^+(L_n)$ that acts trivially on $L_n^*/L_n$ and fixes $\Pi\subset L_n\otimes \RR$ pointwise.
\end{claim}

While this statement in principle provides a complete characterization of the symmetry groups $G_\Pi$, in practice it is not a trivial task to determine which groups can actually arise. The following propositions provide a much more useful characterization of the groups $G_\Pi$.

\begin{proposition}\label{th:lemma}
	The group $G_\Pi$ is isomorphic to the group $O^0(L_\Pi)$ of automorphisms of the lattice $L_\Pi:=L_n\cap \Pi^\perp$ that act trivially on the discriminant group $L_\Pi^*/L_\Pi$.
\end{proposition}

The proof is in appendix \ref{a:pflemma}.
 
 This proposition shows, in particular, that the group $G_\Pi$ only depends on the isomorphism class of the lattice $L_\Pi$. In particular, different models $\C_\Pi,\C_{\Pi'}$ corresponding to isomorphic lattices $L_{\Pi}\cong L_{\Pi'}$ have isomorphic groups $G_\Pi\cong G_{\Pi'}$, even if they are not related by a duality.
 
 The most useful characterization of the groups $G_\Pi$ is given by the following theorem, which is a direct analogue of the main result of \cite{K3symm}. Let $\Lambda$ be the negative definite Leech lattice, i.e. the unique negative definite even unimodular lattice of rank $24$ without roots, i.e. vectors of square norm $-2$. Its group of automorphisms $O(\Lambda)$ is the finite group $Co_0\cong \ZZ_2.Co_1$, which in turn is a central $\ZZ_2$ extension of a finite simple group $Co_1$.

\begin{proposition}\label{th:main} Let $\C_\Pi$ be a non-singular SCFT of $K3^{[n]}$-type, $n> 1$, corresponding to a $4$-subspace $\Pi\subset L_n\otimes \RR\cong \RR^{4,21}$, and such that $\Pi$ is not orthogonal to any $r\in L_n$, $r^2=-2$. Let $L_\Pi=L_n\cap \Pi^\perp$ and $G_\Pi=O^0(L_\Pi)$ be the group of symmetries of $\C_\Pi$ satisfying the conditions (1), (2), and (3) above. Then, $L_\Pi$ is isomorphic to a primitive sublattice $\Lambda_\Pi$ of the Leech lattice $\Lambda$ of rank $\rk \Lambda_\Pi \le 21$, and $G_\Pi$ is isomorphic to the subgroup of $\Aut(\Lambda)\cong Co_0$ fixing pointwise the sublattice $\Lambda\cap \Lambda_\Pi^\perp$ of rank at least $3$. 
	
	Vice versa, if $\tilde G\subset Co_0$ is the pointwise stabilizer of a sublattice $\Lambda^{\tilde G}$ of the Leech lattice $\Lambda$ of rank at least $3$, then for every $n>1$ there exists  a non-singular NLSM $\C_\Pi$ on $K3^{[n]}$ whose group of symmetries $G_\Pi$ is isomorphic to $\tilde G$ and such that $L_\Pi\cong (\Lambda^{\tilde G})^\perp \cap \Lambda$.
\end{proposition}

The proof is in appendix \ref{a:pfthmain}.

Triples $(\tilde G,\Lambda^{\tilde G},\Lambda_{\tilde G})$ where $\tilde G\subset Co_0$ is the stabilizer of a sublattice of $\Lambda$, and $\Lambda^{\tilde G}$ and $\Lambda_{\tilde G}=(\Lambda^{\tilde G})\cap \Lambda$ are the corresponding invariant and coinvariant lattices, were classified in \cite{HohnMason} up to isomorphisms. Out of the $290$ groups listed in \cite{HohnMason}, $221$ have an invariant sublattice $\Lambda^{\tilde G}$ of rank at least $3$. As expected from the analysis of section \ref{s:symmorbmoduli}, the quantum symmetry of symmetric orbifolds is not included in the group $G_\Pi$ classified in this theorem. Indeed, there is  no element of $Co_0$ of order $2$ fixing a $3$-dimensional sublattice of $\Lambda$ and acting by $-1$ on the orthogonal complement. This means that, as argued in \ref{s:symmorbmoduli}, either the $n$-th symmetric orbifold corresponds to a $4$-plane $\Pi$ orthogonal to a root (which seems to be the case for $n=2$), or the quantum symmetry does not satisfy the conditions (1), (2), and (3). In fact, it is easy to see that the extension of the quantum symmetry to $O(\Gamma^{5,21})$ must act on $v$ by $v\mapsto -v$, so that (3) is not satisfied.

It is quite remarkable that the list of groups appearing as possible groups of symmetries of NLSM of $K3^{[n]}$ type is essentially independent of $n$, at least for $n>1$. When the fixed sublattice has rank at least $4$, a partial  explanation of this phenomenon is given by the symmetric orbifold construction. Indeed, in this case, we know from \cite{K3symm} that there exists a NLSM on K3 $\C'$ with symmetry group $G$, so that in all the symmetric orbifold models $\Sym^n \C'$ the symmetry group must contain $G$, and it is easy to take a deformation for which the group is \emph{exactly} $G$. From the perspective of type IIB string theory on K3, what happens is that there is a family of models in the moduli space having  a self-duality group $G\subset O^{\Gamma^{5,21}}$ fixing a sublattice $\Gamma^{1,1}\subset \Gamma^{5,21}$. Therefore, all NLSM describing the world-sheet dynamics of stringy-like objects whose charge is in this $G$-fixed sublattice will have $G$ as a group of symmetries. We do not see any similar explanation for the case where the rank of the fixed sublattice is exactly $3$.
	
	Let us briefly comment on the problem of classifying groups of symmetries $G^*_\Pi$ that satisfy property (1) and (2), but not necessarily (3). Let us assume that $O^+_n$ contains only elements of $O(\Gamma^{5,21})$ acting by $v\mapsto \pm v$ of the vector $v$.   For a given model $\C_\Pi$ of $K3^{[n]}$ type, $G^*_\Pi$ might be either equal to $G_\Pi$ (if there are no symmetries satisfying (1) and (2), but not (3)), or of the form $G^*_\Pi=G_\Pi.\ZZ_2$, i.e. it contains a normal subgroup $G_\Pi$ such that $G^*_\Pi/G_\Pi\cong\ZZ_2$. Furthermore, it is clear from our construction that $G_\Pi^*$ must be a subgroup of the group $O(L_\Pi)$ of the lattice $L_\Pi:=\Gamma^{5,21}\cap Z$. It seems difficult to find a practical criterion to determine whether $G^*_\Pi$ is larger than $G_\Pi$ or not. In particular, this seems to depend not only on the isomorphism class of $L_\Pi$, but also on $n$, on $v$, and in the way $L_\Pi$ is embedded in $\Gamma^{5,21}$. It also seems difficult to determine the precise structure of the group $G_\Pi^*$ in the cases where it is different from $G_\Pi$ -- in general, there can be many non-isomorphic groups $G_\Pi^*$ such that $G^*_\Pi/G_\Pi\cong \ZZ_2$. In particular, the example of the symmetric orbifolds shows that $G_\Pi^*$ is not always a subgroup of $Co_0$. We hope we will be able to address these issues in future works.

\subsection{Twining genera}\label{s:twining}

Given a NLSM $\C$ of $K3^{[n]}$ type with a group of symmetries $G$ (satisfying the conditions (1), (2) and (3) of section \ref{s:callifsymm}), one is interested in finding how the group $G$ acts on the states of the model, in the sense of determining how the space of states of $\C$ decomposes into irreducible representations of $G$. This information can be partially recovered if one knows all \emph{twining genera} $\phi_g(\tau,z)$, that are functions on $\Hh\times \CC$ defined by
\be\label{twindef} \phi_g(\tau,z):=\Tr_{RR}(g\, q^{L_0-\frac{c}{24}}  \bar q^{\bar L_0-\frac{\bar c}{24}} y^{J_0^3} (-1)^{F+\bar F})\ ,\qquad q=e^{2\pi i\tau}, y=e^{2\pi i z},  
\ee for all $g\in G$. Here, the trace is taken over the Ramond-Ramond sector of the theory $\C$,  $L_0$ and $\bar L_0$ are the holomorphic and anti-holomorphic Virasoro generators, $J_0^3$ is the zero mode of a Cartan generator in the $su(2)_k$ subalgebra of the holomorphic $\N=4$ superconformal algebra, $F$ and $\bar F$ are the holomorphic and anti-holomorphic fermion numbers. For $g=1$, $\phi_g$ reduces to the elliptic genus of the model $\C$. As for the elliptic genus, only states with $L_0-\frac{\bar c}{24}=0$ (right-moving ground states) can give a non-vanishing contribution to this trace: the contributions of states with $L_0-\frac{\bar c}{24}>0$ cancel each other due to supersymmetry. This implies that $\phi_g(\tau,z)$ is a holomorphic function of $(\tau,z)\in \HH\times\CC$

In a path integral formulation of the NLSM, the twining genera can be described in terms of a  path integral on a world-sheet $\Sigma$ of genus $1$ (a torus) with modular parameter $\tau$, with the insertion of the operator $y^{J_0^3}$, and where all the fields are required to be periodic around one of the cycles in a basis of $H_1(\Sigma,\ZZ)$, and twisted by $g$ as one goes around the other cycle. From  this description, and from spectral flow invariance, one can argue that $\phi_g(\tau,z)$ must transform as a Jacobi form of weight $0$ and index $n$ (see \cite{eichler_zagier} for the relevant definitions)
\begin{align}
\label{twinmod}\phi_{g} \left ({a \tau + b\over c\tau +d}, {z \over c\tau + d}\right )&=  e^{2 \pi i m {c z^2\over c\tau + d}} \chi_g\left(\begin{array}{cc}
	a & b  \\
	c & d  \end{array}\right)\phi_{g}(\tau, z)  ~~~&& \forall 
\left(\begin{array}{cc}
	a & b  \\
	c & d  \end{array}\right) \in \Gamma_g\subseteq SL_2(\mathbb Z),
\\
\phi_{g}(\tau, z + \ell \tau + \ell')&=e^{-2 \pi i n(\ell^2\tau + 2\ell z)}\,  \phi_{g}(\tau,z)~ ~~ &&\forall (\ell, \ell') \in \mathbb Z^2, 
\end{align}
for a suitable subgroup $\Gamma_g\subset SL(2,\ZZ)$. We have allowed for the possibility of a non-trivial multiplier $\chi_g:\Gamma_g\to \CC^\times$, which cannot be excluded by path integral arguments (and can be shown to exist in some explicit examples). 
%

Once the twining genera $\phi_g$ are given for all $g\in G$, using standard group theory arguments one can determine how every simultaneous eigenspace for $L_0,\bar L_0,J_0^3$ (which is always finite dimensional) decomposes as a sum $\oplus_i n_i R_i$ over the irreducible representations $R_i$ of $G$, where $n_i$ are $\ZZ_2$-graded multiplicities, with the grading given by the total fermion number $(-1)^{F+\bar F}$. This means that the twining genera are not sufficient to detect whether a certain $(L_0,\bar L_0,J_0^3)$-eigenspace contains the sum of two copies of the same $G$-representation with opposite fermion number, since their contribution cancels exactly. Of course, we know that huge cancellations do occur in every model $\C$ due to supersymmetry, since  all contributions from states with $\bar L_0-\frac{\bar c}{24}>0$ cancel each other. In this sense, from the twining genera one can only hope to recover information about the subspace states with $\bar L_0-\frac{c}{24}=0$, i.e. the ones that are supersymmetric (BPS) with respect to the anti-holomorphic $\N=4$ superconformal algebra. It is believed  (though, to the best of my knowledge, no rigorous proof exists) that at a generic point in the moduli space there are no cancellations among the contributions of the BPS states to the elliptic genus. If this is the case, then for a generic model the twining genera are sufficient to unambiguously determine the decomposition of the space of BPS states into irreducible $G$-representations. On the other hand, models with a non-trivial symmetry group $G$ form a subset of measure zero in the moduli space, so it might very well be that only a small fraction of them (or none at all!) is `generic', in the sense above. 

Nevertheless, the twining genera are useful for two reasons. First,  they do provide interesting information about the action of $G$ on the space of BPS states -- even if the $G$-representation is not determined  unambiguously, it is certainly strongly constrained by the twining genera. Secondly, as we will explain in the rest of this section, one can determine a large number of such twining genera, whereas by contrast the analogous computation of `twining partition functions' without the inclusion of the fermion number $(-1)^{F+\bar F}$ is in general out of reach (except for a few very special points in the moduli space).

The property that makes the twining genus $\phi_g$ computable is the fact that it is invariant under continuous deformations of the moduli that preserve the symmetry $g$. The argument for this is completely analogous to the one leading to the invariance of the elliptic genus (see for example \cite{Cheng:2016org}). By the analysis in section \ref{s:callifsymm}, any symmetry $g$ can be identified with an element in $Stab^+(v)\subset O^+(\Gamma^{5,21})$ fixing a sublattice $\Gamma^g\subset \Gamma^{5,21}$ of signature $(5,d)$, $d\ge 0$.  An infinitesimal deformation of the NLSM preserves the symmetry $g$ if and only if it is generated by a $g$-invariant exactly marginal operator. It follows that, for each symmetry $g\in O^+(\Gamma^{5,21})$, there is a family of non-singular models with symmetry $g$, consisting of all non-singular (in the sense of section \ref{s:singular}) positive definite oriented $5$-planes $Z$ containing $v$, with $Z\subset \Gamma^g\otimes \RR$. Equivalently, it can be seen as an element of $O^{0+}(L_v)\cong Stab^+(v)$ acting trivially on a sublattice $(L_v)^g$ of signature $(4,d)$, $d\ge 0$. This family is necessarily connected, because it is the quotient by a discrete group of dualities of a space of the form $Gr(4,d)\setminus \calS$, where $Gr(4,d)$ is the Grassmannian of positive definite four planes $\Pi\in (L_v)^g\otimes \RR\cong \RR^{4,d}$, and $\calS$ is the locus of singular NLSM, described in section \ref{s:singular}, which has codimension at least $4$. This means that the twining genus $\phi_g$ depends only on the element $g\in Stab^+(v)$, and not on the particular model $\C\in \F_g$ in which it is computed. Therefore, in order to determine the twining genus $\phi_g$ for the whole family $\F_g$, it is sufficient to compute it at any point in $\F_g$. Furthermore, if two elements $g,g'\in Stab^+(v)$ are conjugate in $Stab^+(v)$, i.e. $g'=hgh^{-1}$ for some $h\in Stab^+(v)$, then the models of the family $\F_g$ are dual to the ones in the family $\F_{g'}$, and the corresponding twining genera $\phi_g$ and $\phi_{g'}$ are the same.\footnote{This is true more generally for conjugation by any $h\in O^+_n$. Since it is not completely clear to us what this group is, in this section and in the following we will be conservative and only consider dualities by $Stab^+(v)$.  } It follows that the twining genus $\phi_g$ only depends on the conjugacy class of $g$ in $Stab^+(v)$. Finally, invariance of the $\N=4$ characters under charge conjugation implies
\be \phi_{g}(\tau,z)=\phi_{g^{-1}}(\tau,-z)=\phi_{g^{-1}}(\tau,z)\ .
\ee
To summarize, while the definition of the twining genera $\phi_g$ refers to a specific CFT of K3$^{[n]}$ type, the twining genus itself only depends on the element $g\in Stab^+(v)$ up to conjugation in $Stab^+(v)$ and is invariant under charge conjugation $g\leftrightarrow g^{-1}$. Each such class determines a connected family $\F_g$ of non-singular CFTs of K3$^{[n]}$ type, such that $\phi_g$ is well-defined at each point in $\F_g$ and is constant along $\F_g$. In section \ref{s:conjclass} we will discuss the classification of such conjugacy classes of symmetries.

\subsection{Conjugacy classes of symmetries}\label{s:conjclass}

Motivated by the analysis of section \ref{s:twining}, we will now consider a classification of all conjugacy classes of elements $g\in Stab^+(v)$ fixing a positive definite $4$-dimensional subspace in $L_v\otimes \RR\cong \RR^{4,21}$. A first rough classification follows by considering the possible eigenvalues of $g$ in the defining $25$-dimensional representation. By construction $g$ acts non-trivially only on a negative definite sublattice $\Gamma_g\subseteq L_v\cap \Pi^\perp $  of rank $21-d$, $d\ge 0$, so that there are only $21-d$ non-trivial eigenvalues. By theorem \ref{th:main}, the lattice $\Gamma_g$ can be primitively embedded in the Leech lattice $\Lambda$, and the non-trivial eigenvalues of $g$ are the same as for an element $g'\in Co_0$, such that the cyclic group $\langle g'\rangle\subset Co_0$ is the pointwise stabilizer of a $3+d$ dimensional sublattice $\Lambda^{\langle g'\rangle}\subset \Lambda$. There are $42$ conjugacy classes $[g']$ of $Co_0$ with the property that its elements $g'$ stabilize a sublattice of rank at least $3$ (see \cite{Conway:1985vn}). As a matter of fact, all such classes actually stabilize a lattice of dimension at least $4$, so that $d\ge 1$.\footnote{More generally, the fixed subspace of any element in $SO(24,\RR)$ has even (possibly $0$) dimension. Indeed, any real orthogonal matrix is diagonalizable over $\CC$ with eigenvalues having modulus $1$. Its characteristic polynomial has real coefficients, so the non-real eigenvalues must come in complex conjugate pairs, while the multiplicity of the $-1$-eigenvalue must be even in order for the determinant to be positive. Therefore, the multiplicity of the $+1$-eigenvalue is also even. We stress that this argument has only implications for \emph{cyclic} subgroups of $Co_0$; non-cyclic subgroups of $Co_0$ fixing a sublattice of rank exactly $3$ do exist, see for example \cite{HohnMason}.} Furthermore, for any two distinct $Co_0$ classes, the corresponding elements have different sets of eigenvalues, and non-isomorphic lattices $\Lambda_{g'}=(\Lambda^{g'})^\perp\cap \Lambda$ (see for example \cite{HaradaLang1990,HohnMason}). The sets of eigenvalues are most easily encoded in the  Frame shape of $g'$, {\it i.e.} a symbolic product
\be \pi_{g'} := \prod_{\ell|N}\ell^{k_\ell}\ , \ee
where $N=o(g')$ is the order of $g'$. The integers $k_\ell\in \ZZ$ are defined by
\be \det (t{\bf 1}_{24}- \rho_{ 24}(g)) =\prod_{\ell|N} (t^\ell-1)^{k_\ell}\ .
\ee If $g'$ acts as a permutation of the vectors in some basis of the $24$ dimensional representation of $Co_0$, then all $k_\ell$ are non-negative and the Frame shape coincides with the cycle shape of the permutation. By theorem \ref{th:main}, we conclude that there are $42$ possible sets of eigenvalues for symmetries $g\in Stab^+(v)$, each corresponding to a certain isomorphism class of lattices $\Gamma_g$.

While two elements $g_1,g_2\in Stab^+(v)$ with different Frame shapes obviously belong to distinct $Stab^+(v)$ conjugacy classes, the inverse is not necessarily true: it could happen that $g_1,g_2$ with the same Frame shape are not conjugate in $Stab^+(v)$. We are left with the problem of determining the possible classes for each of the $42$ Frame shapes. The following lemma is useful in this sense.
\begin{lemma}
	Let $L$ be an even lattice and let $G\subseteq O(L)$ be a subgroup of its group of automorphisms. Let $g_1,g_2\in G$ be such that the coinvariant sublattice $L_{g_k}:=(L^{g_k})^\perp\cap L$, $k=1,2$, are both isomorphic to a given lattice $M$. Let $i_1,i_2:M\hookrightarrow L$ be primitive embeddings and let  $g\in O(M)$ be an automorphism of $M$, such that $i_k(M)=L_{g_k}$  and $g_k\circ i_k=i_k\circ g$, $k=1,2$.  
Then, $g_1$ and $g_2$ are conjugate in $G$ if and only if there exist $h\in G$ and $s\in C_{O(M)}(g)$ (the centralizer of $g$ in $O(M)$) such that 
	\be h\circ i_1=i_2\circ s\ ,
	\ee and in this case $g_2=hg_1h^{-1}$.
\end{lemma}
\begin{proof}
This is an	immediate generalization of the proof of Lemma 8 in \cite{Cheng:2016org}.
\end{proof}

In the case we are interested in, $L$ is the lattice $L_v=v^\perp\cap \Gamma^{5,21}$, $G=Stab^+(v)$ is the subgroup of $O(\Gamma^{5,21})$ preserving $v$ and preserving the orientation of positive definite $4$-subspaces in $L_v\otimes \RR$, and the lattice $M$ is the coinvariant sublattice $\Lambda_g=\Lambda\cap (\Lambda^g)^\perp$, i.e. the orthogonal complement of the $g$-fixed sublattice in the Leech lattice, for some $g\in Co_0$ with the given Frame shape $\pi_g$. 
Therefore, the  symmetries with Frame shape $\pi_g$ up to $Stab^+(v)$ transformations are in one to one correspondence with the double cosets
\be\label{upbound}  \{\text{primitive }i:\Lambda_g\hookrightarrow L_v\}/Stab^+(v)\ ,
\ee where $Stab^+(v)$ is the stabilizer of $v$ in $O^+(\Gamma^{5,21})$. In principle, one might want to consider the number of classes of symmetries up to $O^+_n$ transformation, where $O^+_n$ might contain (at least) elements that flip the sign of the vector $v$. Furthermore, as follows from the lemma above, the number of conjugacy classes of symmetries with Frame shape $\pi_g$ is obtained by quotienting also by the centralizer $C_{\Lambda_g}(g)$.  In this paper, we will only compute the number of cosets in \eqref{upbound} (and only in certain cases, see below). This gives an upper bound on the number of $O_n^+$-classes of symmetries. On the other hand, for most Frame shapes $\pi_g$, this upper bound is either $1$ or equals some lower bound that can be obtained by other arguments (in particular, by the number of known twining genera), so it actually equals the number of classes.

The single cosets \eqref{upbound} are in one to one correspondence with the cosets in
\be \label{upboundeq} \{(\hat v,\hat i)\mid \hat i:\Lambda_g\hookrightarrow \Gamma^{5,21}\text{ primitive},\ \hat v\in \Gamma^{5,21}\cap \hat i(\Lambda_g)^\perp\text{ primitive},\  \hat v^2=2m-2\}/O^+(\Gamma^{5,21})\ .
\ee Indeed, there is a map from \eqref{upbound} to \eqref{upboundeq} given by assigning to the coset $[i]$ with representative $i$ the coset $[(v,i)]$ with representative $(v,i)$. If we choose a different representative $i'$ in the same coset $[i]$, then $(v,i)$ and $(v,i')$ are related by $Stab^+(v)\subset O^+(\Gamma^{5,21})$, so they belong to the same $[(v,i)]$. It follows that the map is well defined. Vice versa, given any coset in \eqref{upboundeq}, using the fact that any two primitive vectors of same length in $\Gamma^{5,21}$ are related by an $O^+(\Gamma^{5,21})$ transformation, we can choose a representative of the form $(v,\hat i)$, and the embedding $\hat i$ is determined up to automorphisms in $Stab^+(v)$. Therefore, this determines a coset $[\hat i]$ in \eqref{upbound}, and gives a well defined map from \eqref{upboundeq} to \eqref{upbound} that is clearly the inverse of the previous one. 

The quotient in \eqref{upboundeq} admit two alternative useful descriptions. The first is given by noticing that a pair $(\hat v,\hat i)$ as in \eqref{upboundeq} determines an embedding $\tilde i:\Lambda_{g,n}\hookrightarrow \Gamma^{5,21}$, where
\be \Lambda_{g,n}:=\Lambda_g\oplus \langle 2n-2\rangle\ .
\ee More precisely, the cosets in \eqref{upboundeq} are in one to one correspondence with
\be\label{cosetsemb} \{\tilde i:\Lambda_{g,n}\hookrightarrow \Gamma^{5,21}\mid \tilde i(\Lambda_g), \tilde i(\langle 2n-2\rangle)\text{ primitive in }\Gamma^{5,21}\}/O^+(\Gamma^{5,21})
\ee

The second one uses the fact that the lattices $\Lambda_g$ we are considering are the same appearing in the classification of symmetries of NLSM on K3. Using the results of \cite{Nikulin}, one can show that all such lattices can be primitively embedded in an even unimodular lattice $\Gamma^{4,20}$ (see \cite{K3symm} for details).  For any decomposition $\Gamma^{5,21}=\Gamma^{1,1}\oplus \Gamma^{4,20}$, this gives a primitive embedding of $\Lambda_g$ in $\Gamma^{5,21}$, such that the orthogonal complement $K$ is of the form  $K=\Gamma^{1,1}\oplus K'$ for some lattice $K'$. 
For lattices $K$ of this form, theorem 1.14.2 of Nikulin \cite{Nikulin} then implies that the primitive embedding of $\Lambda_g$ in $\Gamma^{5,21}$ is unique up to $O^+(\Gamma^{5,21})$ transformations. For each $\Lambda_g$ let us choose one such primitive embedding $ i_g:\Lambda_g \to \Gamma^{5,21}$, and set $K:=i_g(\Lambda_g)^\perp\cap\Gamma^{5,21}$. Then,  each coset in \eqref{upboundeq} has a representative of the form $(\hat v,i_g)$, where $\hat v\in K\subset \Gamma^{5,21}$ is determined up to $O^+(\Gamma^{5,21})$ transformations acting trivially on $i_g(\Lambda_g)$. The latter transformations act on $K$ by an automorphism in $O^{0+}(K)$, fixing the orientation of positive definite $5$-dimensional subspaces and acting trivially on $K^*/K$. In fact, every element in $O^{0+}(K)$ extends to an $O^+(\Gamma^{5,21})$ automorphism acting trivially on $K^\perp\equiv i_g(\Lambda_g)$. We conclude that the cosets in \eqref{upboundeq} are in one to one correspondence with
\be\label{cosetsprim} \{\text{primitive } v\in K,\\ v^2=2n-2\}/O^{0+}(K)\ ,
\ee i.e. $O^{0+}(K)$-orbits of primitive $v\in K$ of length $2n-2$.

It remains to count these orbits. First of all, given a primitive $v\in K$, let $\Div(v)$ be the maximal positive integer such that $\frac{v}{\Div(v)}\in K^*$. One has $\Div(v)$ divides $\gcd(N,2n-2)$, and that for any $f\in O(k)$, $\Div(f(v))=\Div(v)$ (see appendix \ref{a:lattices}). If $f\in O^{0+}(K)$, then $\frac{1}{\Div(v)} f(v)\equiv \frac{1}{\Div(v)} v\mod K$, so that a necessary condition for two primitive vectors $v,v'\in K$ of same length $2n-2$ to be in the same $O^{0+}(K)$-orbit is that $\frac{v}{\Div(v)}$ and $\frac{v'}{\Div(v')}$ are in the same $K^*/K$ coset (in particular, $\Div(v)=\Div(v')$; this is true more generally for vectors in the same $O(K)$-orbit). Therefore, there is at least one $O^{0+}(K)$-orbit for each generator $x\in K^*/K$ (elements $x\in K^*/K$ that are not generators cannot be written as $v/\Div(v)$ for some $v\in K$).

The case where $\Div(v)=1$ is particularly interesting. This condition corresponds to the trivial coset $x\equiv 0 \in K^*/K$, since $\frac{v}{\Div(v)}=v\in K$. This is the only possible value for $\Div(v)$ when the order $N$ of $g$ and the norm $v^2=2n-2$ are coprime. Furthermore, as we will argue in section \ref{s:elevenandfriends}, based on the analysis of section \ref{s:symmorbmoduli}, this is also the case when the NLSM is a symmetric orbifold $\Sym^n(\C_{K3})$ of some NLSM on K3 $\C_{K3}$, and $g$ is induced by a symmetry of $\C_{K3}$. It is easy to see that $\Div(v)=1$ if and only if  $\langle v\rangle \oplus \Lambda_g$ is a \emph{primitive} sublattice in $\Gamma^{5,21}$.\footnote{Suppose $\Div(v)=1$. Every primitive vector of $\langle v\rangle \oplus \Lambda_g$ is of the form $av+b\lambda$, with $\lambda$ primitive and $\gcd(a,b)=1$. Suppose that $\frac{1}{k}(av+b\lambda)\in \Gamma^{5,21}\subset K^*\oplus (\Lambda_g)^*$ for some integer $k$. Since $\frac{av}{k}\in K^*$ and $\Div(v)=1$, it must be $k|a$ and $\frac{av}{k}\in K$, i.e. $\frac{av}{k}$ belongs to the trivial coset of $K^*/K\cong (\Lambda_g)^*/\Lambda_g$. Since $\frac{av}{k}+\frac{b\lambda}{k}\in \Gamma^{5,21}$, the gluing conditions imply that also $\frac{b\lambda}{k}\in \Lambda_g$. As a consequence, $k$ divides $b$, because $\lambda$ is primitive. But then $k$ divides $\gcd(a,b)=1$, so that $k=1$. Thus, every primitive vector of $\langle v\rangle \oplus \Lambda_g$ is primitive in $\Gamma^{5,21}$ as well. For the vice versa, notice that, for a general $\Div(v)$, the gluing construction implies that $\Gamma^{5,21}\subset K^*\oplus (\Lambda_g)^*$ contains a vector of the form $\frac{1}{\Div(v)}(v+\lambda)$, for some $\lambda\in \Lambda_g$, so if $\Div(v)>1$ then $\langle v\rangle \oplus \Lambda_g$ is not primitively embedded in $\Gamma^{5,21}$. 
 } 

Therefore, the subset of cosets in \eqref{cosetsprim} where $\Div(v)=1$ corresponds to the subset of cosets in \eqref{cosetsemb} where the $\tilde i:\Lambda_{g,n}\hookrightarrow \Gamma^{5,21}$ are primitive. Since $\Lambda_{g,n}=\langle v\rangle\oplus \Lambda_g$ is indefinite of signature $(1,d)$, $0\le d\le 20$, one can apply the theorems by Miranda and Morrison \cite{MirandaMorrison1,MirandaMorrison2} to compute the number of such embeddings. This number depends essentially on the discriminant form of $\langle v\rangle\oplus \Lambda_g$, which is the direct sum of the discriminant form on $\Lambda_g^*/\Lambda_g$ plus the discriminant form $(A,q)$, where $A\cong \ZZ_{2n-2}$ has a generator $x$ such that $q(x)=\frac{1}{2n-2}$. The necessary information about the discriminant forms is reported in appendix \ref{a:discriminants}. The results of this calculation are contained in table \ref{tab:big}.

More generally, if $\Div(v)>1$, then $\frac{v}{\Div(v)}$ determines a non-trivial element $x\in K^*/K\cong \Lambda_g^*/\Lambda$, so that one has a primitive embedding in $\Gamma^{5,21}$ of an overlattice $N\supset \langle v\rangle\oplus \Lambda_g$, generated by $\langle v\rangle\oplus \Lambda_g$ together with a vector of the form $\frac{1}{\Div(v)}(v+\lambda)$, where $\lambda\in \Lambda_g$ is such that $\frac{\lambda}{\Div(v)}\equiv\frac{\lambda}{\Div(\lambda)}$  is in the same class $x\in \Lambda_g^*/\Lambda_g$. Thus, $N$ is an indefinite lattice whose discriminant group is a quotient of the discriminant group of $\langle v\rangle\oplus \Lambda_g$ by a cyclic subgroup. In principle, one might be able to work out all the possibilities, but we will not do this here.

\section{Second quantized twining genera}\label{s:elevenandfriends}
	Let $\C$ be a NLSM on a single K3 surface, and let us consider the symmetric orbifold $\Sym^n(\C)$. The elliptic genus is given by the $p^m$ coefficient in the infinite product \cite{Dijkgraaf:1996it,Dijkgraaf:1996xw}
\be \sum_{n=1}^{\infty} p^n\phi_{\Sym^n(K3)}(\tau,z)=\prod_{\substack{m,n,l\in \ZZ\\ m>0,n\ge 0}} (1-p^nq^my^l)^{-c(mn,l)}\ ,
\ee where $c(m,l)$ are the Fourier coefficients of the elliptic genus of K3
\be \phi_{K3}(\tau,z)=\sum_{\substack{m,l\in \ZZ\\n\ge 0}}c(m,l)q^my^l\ .
\ee
Let $G$ be the symmetry group of the `seed' model $\C$, commuting with the $\N=4$ SCA and the spectral flow generators. The action of $G$ lifts to a group of symmetries $\tilde G$ of  $\Sym^n(\C)$ satisfying analogous properties (conditions (1) and (2) of section \ref{s:callifsymm}), and preserving the twisted and untwisted sectors. More precisely, the group $\tilde G$ fits in an exact sequence
\be 1\to H\to \tilde G\to G\to 1\ ,
\ee where $H$ is a group acting trivially on the untwisted sector. Since $\Sym^n\C$ is generated by the untwisted sector and by the ground state of the $\sigma$-twisted sector, with $\sigma\in S_n$ a single transposition, the group $H$ can only act by phases on the $\sigma$-twisted ground state. The only such symmetry is, in fact, the quantum symmetry, so that $H\cong \ZZ_2$, and $\tilde G$ is a $\ZZ_2$ central extension of $G$. 

 In the language of the previous sections, the model $\Sym^n(\C)$ will correspond to a positive definite four-plane $\Pi\subset L_v\otimes \RR$, where $L_v=v^\perp \cap \Gamma^{5,21}$ for a primitive $v\in \Gamma^{5,21}$ of length $2n-2$. Since $G$ is a symmetry of the fundamental string world-sheet theory, one can argue that it will lift to a symmetry of the whole string theory, so that $\tilde G$ will act on $L_v$ by lattice automorphisms.  The analysis of section \ref{s:symmorbmoduli} shows that $\Pi$ is orthogonal to a vector $s\in L_v$, with $s^2=2-2n$, and such that $L_v\cong \langle s\rangle\oplus_\perp \Gamma^{4,20}$.  The quantum symmetry of the model, that acts by $s\mapsto -s$ and fixes $L_v\cap s^\perp\cong \Gamma^{4,20}$, is central in  $\tilde G$, since $\tilde G$ does not mix the untwisted and the twisted sectors. This means that the action of $\tilde G$ on $L_v$ preserves setwise the sublattice $s^\perp \cap L_v\cong \Gamma^{4,20}$. The group $\tilde G$ has a normal subgroup of index $2$ that fixes the vector $s$ and acts faithfully on $\Gamma^{4,20}$. This means that this normal subgroup is isomorphic to $\tilde G/H\cong G$; therefore, the $\ZZ_2$ extension of $G$ is split, and $\tilde G\cong \ZZ_2\times G$. Thus, the lift of $G$ to $\Sym^n\C$ can be chosen to act trivially on the twisted ground state and to be isomorphic to $G$ itself. With a certain abuse of notation, we will denote this lift again by $G$. Since $G$ fixes $s$, it must act trivially on $L_v^*/L_v$. Therefore, the action of $G$ on $L_v$ extends to an action on $\Gamma^{5,21}$ by automorphisms that fix $v$. Therefore, conditions (1), (2), and (3) of section \ref{s:callifsymm} are satisfied. 
 
 Let $\Gamma^G$ be the $G$-fixed sublattice of $\Gamma^{5,21}$ and $\Gamma_G=(\Gamma^G)^\perp\cap \Gamma^{5,21}$ its orthogonal complement. One has $\Gamma_g\subset \Gamma_G$ for all $g\in G$. Recall that $v$ and $s$ are contained in a sublattice $\Gamma^{1,1}\subset \Gamma^{5,21}$. Since both $v$ and $s$ are $G$-fixed, then $\Gamma^{1,1}$ is a sublattice of $\Gamma^G$, and $v$ has divisor $1$ in this sublattice. Therefore, the sublattice $\Gamma_G\oplus_\perp \langle v\rangle$ is primitive in $\Gamma^{5,21}$, and the same is true for the sublattices $\Gamma_g\oplus_\perp \langle v\rangle$ for all $g\in G$. We conclude that, as claimed in section \ref{s:conjclass}, the $Stab^+(v)$-conjugacy class of any symmetry $g$ of $\Sym^n(\C)$, inherited from a symmetry of $\C$, is such that $v$ has divisor $1$ in $\Gamma^{5,21}\cap \Gamma_g$. Therefore, these classes are the ones described in table \ref{tab:big}.

 The twining genera $\phi^{\Sym^n(K3)}_g$ can be obtained from the generating functions  \be \Psi_g= \sum_{n=1}^{\infty} p^n\phi^{\Sym^n(K3)}_g(\tau,z)=\prod_{\substack{m,n,l\in \ZZ\\ n>0,m\ge 0}}\prod_{t\in \ZZ/N\ZZ}  (1-e^{\frac{2\pi it}{N}}p^nq^my^l)^{-\hat c_{t}(mn,l)}\ ,
\ee where $N$ is the order of $g$, and $\hat c_{t}$ are the Fourier coefficients of the discrete Fourier transforms of $\phi_{g^i}$
\be \hat \phi_t(\tau,z)=\sum_{m=0}^\infty \sum_{l\in \ZZ} \hat c_{t}(m,l)q^my^l=\frac{1}{N}\sum_{k\in \ZZ/N\ZZ} e^{-\frac{2\pi i tk}{N}}\phi_{g^k}(\tau,z)\ .
\ee
Therefore, the twining genera of the symmetric orbifold $\Sym^n(\C)$ are completely determined in terms of the twining genera of the `seed' K3 model $\C$. The complete list of the possible twining genera for a K3 model $\C$ can be found in \cite{Paquette:2017gmb}. The twining genera $\phi_g$ can be written as
\be\label{twinformula} \phi_g(\tau,z)=A_g \chi_{0,1}(\tau,z)+F_g(\tau)\chi_{-2,1}(\tau,z)\ ,
\ee where $\chi_{0,1}$ and $\chi_{-2,1}$ are the standard weak Jacobi forms given in terms of Jacobi theta functions and Dedekind eta series as
\be \chi_{0,1}(\tau,z)=4\sum_{i=2}^4\frac{\vartheta_i(\tau,z)^2}{\vartheta_i(\tau,0)^2}\ ,\qquad \qquad \chi_{-2,1}(\tau,z)=\frac{\vartheta_1(\tau,z)^2}{\eta(\tau)^6}\ ,
\ee $A_g$ is a constant depending on the the Frame shape of $g$ (specifically, if $\pi_g=\prod_{\ell|N}\ell^{k_\ell}$ is the Frame shape of $g$, then $A_g=\frac{1}{12}\sum_{\ell|N}k_{\ell}$), and $F_g(\tau)$ is a modular form of weight $2$ for a congruence subgroup of $SL(2,\ZZ)$, which is given in table \ref{tab:big} (see \cite{Paquette:2017gmb} for notation).

The generating function $\Psi_g$ can be `completed' to a function 
\be\label{PhiSieg} \Phi_g=\frac{p\psi_g(\tau,z)}{\Psi_g(\sigma,\tau,z)}=\prod_{{(m,n,l)}}\prod_{t\in \ZZ/N\ZZ}  (1-e^{\frac{2\pi it}{N}}p^nq^my^l)^{-\hat c_{t}(mn,l)}
\ee which is a meromorphic Siegel modular form for a congruence subgroup of $Sp(4,\ZZ)$ of weight $(d-4)/2$, where $d$ is the dimension of the $g$-fixed subspace in the $24$-dimensional representation. Here, the product is over $m,n,l\in \ZZ$ with $m,n\ge 0$, and with $l<0$ whenever $m=n=0$. The function
\be \psi_g(\tau,z)= qy\prod_{t\in \ZZ/N\ZZ}  (1-e^{\frac{2\pi it}{N}}y^l)^{-\hat c_{t}(0,l)}\prod_{\substack{m,l\in \ZZ\\ n> 0}}(1-e^{\frac{2\pi it}{N}}q^my^l)^{-\hat c_{t}(0,l)}
\ee is a Jacobi form of weight $4-d$ and index $1$ for a subgroup of $SL(2,\ZZ)$, and it only depends on the Frame shape of $g$.

This leads to a very surprising phenomenon. Suppose we have two different NLMS on K3 $\C$ and $\C'$, with symmetries $g$ and $g'$ having the same Frame shape, but belonging to different $O^+(\Gamma^{4,20})$ conjugacy classes, and giving rise to different twining genera $\phi_g$ and $\phi_{g'}$. 
As a consequence, we have different generating functions $\Psi_g$ and $\Psi_{g'}$, and one would expect the twining genera $\phi^{\Sym^n(\C)}_g$ and $\phi^{\Sym^n(\C)}_{g'}$ to be different for generic $n$. However, as can be seen by a quick inspection of table \ref{tab:big}, there are many cases where, for any $n>1$, there is a unique conjugacy class of symmetries with a given Frame shape (e.g., this happens for the Frame shapes $1^211^2$, $1^12^17^114^1$, $1^13^15^115^1$, and many more). This means that there exists a continuous deformation from the  model $\Sym^n(\C)$ to the model $\Sym^n(\C')$, such that the symmetry $g$ of $\Sym^n(\C)$ is preserved and mapped to the symmetry $g'$ of $\Sym^n(\C')$. This deformation must move outside of the symmetric orbifold locus, otherwise it would  exist already at the level of the seed theories $\C$ and $\C'$, i.e. for $n=1$.

It follows that, while $\Psi_g$ and $\Psi_{g'}$ are defined in terms of totally different infinite products, they actually differ only for the $p^1$ coefficient, while all the $p^n$ coefficients are the same for $n>1$. For this to be true, the exponents of the two infinite products must conspire in order to give infinitely many cancellations.

This phenomenon is even more striking if we consider the `completions' $\Phi_g$ and $\Phi_{g'}$, because their inverse $1/\Phi_g$ and $1/\Phi_{g'}$ must differ only for the $p^0$ term (notice that $\psi_g=\psi_{g'}$, since the automorphic corrections only depend on the Frame shape). Thus, the difference $1/\Phi_g-1/\Phi_{g'}$ should be a function of $\tau$ and $z$ only. But it should also be a Siegel modular form, which seems impossible for a function of $\tau$ and $z$ only! In fact, there is only one possibility for this to be true: the modular weight must be zero and the functions $1/\Phi_g$ and $1/\Phi_{g'}$ only differ by a constant (i.e., independent also of $\tau$ and $z$). Indeed, one can check that this phenomenon only occurs for Frame shapes such that the $g$-fixed subspace is $4$-dimensional so that the modular weights of $\Phi_g$ and $\Phi_{g'}$ (or their inverse) is $0$. Furthermore, quite amazingly, it turns out that the difference $\phi_g-\phi_{g'}$ is always proportional to $\psi_g=\psi_{g'}$ (which has weight $0$ for these cases)! Thus, the $p^0$ term in the difference $1/\Phi_g-1/\Phi_{g'}$ is indeed a constant. 

These observations give strong support  to our statement that $\phi^{\Sym^n(\C)}_g=\phi^{\Sym^n(\C)}_{g'}$ for all $n>1$. For the two twining genera $\phi_g$ and $\phi_{g'}$ related to the Frame shape $1^211^2$, we verified these identities up to $n=12$.\footnote{We thank Max Zimet for help with these calculations.} It would be interesting to give a rigorous mathematical proof of these identities for all $n>1$. A possible strategy for such a proof is the following. One knows that $\Phi_g-\Phi_{g'}$ is a meromorphic Siegel modular form of weight $0$. Using some theorems by Borcherds (e.g., theorem 13.3 in \cite{Borcherds98}), one can in principle determine the location of the poles of $1/\Phi_g$ and $1/\Phi_{g'}$. In many cases, one also knows  the coefficients of these poles. If one could prove that all poles of $1/\Phi_g$ and $1/\Phi_{g'}$ are in the same location and have the same coefficients, then the difference should be a holomorphic Siegel modular form of weight $0$, which is necessarily a constant.

\newcolumntype{C}{>{$}c<{$}}

\newcolumntype{E}{>{\eatcell}c@{}}

\begin{landscape}
	\begin{tabularx}{\linewidth}{CCCCC}
		\pi_g
		& (\Gamma^{4,20})^g 
		&
		\begin{matrix}
			 \text{\# Cosets}
		\end{matrix} 
		&   F_{g}(\tau) 
		&\\
		\midrule\endhead
	1^{24} &  \Gamma^{4,20}
		&\begin{matrix}
		\circ
		\end{matrix} 
		& 0 
		&
		\\
		\rowcolor{gray!11}{} 
	1^82^8 &  \Gamma^{4,4}\oplus E_8(-2) 
		&\begin{matrix}
			\circ 
		\end{matrix}
		& -\frac{4}{3}\E_2
		&\\
	1^{-8}2^{16}&  \Gamma^{4,4}(2) 
		&\begin{matrix}
			\circ
			\end{matrix} 
		& -\frac{8}{3}\E_2
		&\\
		\rowcolor{gray!11}{} 
	2^{12}& \ZZ(2)^4\oplus \ZZ(-2)^{\oplus 8} 
		&\begin{matrix}
			\circ
		\end{matrix}
		& 2\E_2-\frac{4}{3}\E_4
		&\\
	1^63^6& \Gamma^{2,2}\oplus \Gamma^{2,2}(3)\oplus (A_2(-1))^{\oplus 2}  
		& 
		\begin{matrix}
			\circ
		\end{matrix} 
		& -\frac{3}{4}\E_3
		&\\
		\rowcolor{gray!11}{} 
	1^{-3}3^{9}&  \Gamma^{2,2}(3)\oplus A_2  
		& \begin{matrix}
			\circ
		  	\end{matrix}  
	  	& -\frac{9}{8}\E_3
		&\\
	3^{8}  &  \Gamma^{4,4}(3)  
		& \begin{matrix} n\notin 3\ZZ & \updownarrow\\ n\in 3\ZZ & \circ\end{matrix}
		& \frac{1}{2}\E_3-\frac{3}{8}\E_9\pm 9\eta[1^33^{-2}9^3]
		\\
		\rowcolor{gray!11}{} 
	1^42^24^4&  \Gamma^{2,2}\oplus\Gamma^{2,2}(4)\oplus \ZZ(-2)^{\oplus 2} 
		&\begin{matrix}
			\circ
			\end{matrix} 
		& \frac{1}{3}\E_2-\frac{2}{3}\E_4
		&\\
	1^82^{-8}4^8&  \Gamma^{4,4}(2) 
		&\begin{matrix}
			\circ
		   \end{matrix}  
	   	&  -\frac{4}{3}\E_2
		\\
		\rowcolor{gray!11}{} 
	1^{-4}2^64^4& \Gamma^{2,2}(4)\oplus\ZZ(2)^{\oplus 2}  
		&\begin{matrix}
			\circ
			\end{matrix}  
		& -\frac{1}{3}\E_2-\frac{2}{3}\E_4
		\\
	2^{-4}4^8  & D_4(2) 
		& \begin{matrix} n\notin 2\ZZ & \updownarrow\\ n\in 2\ZZ & \circ \end{matrix}  
		& 
		\begin{matrix}
			2\E_2-\frac{4}{3}\E_4\\
			-2\E_2
		\end{matrix}
		\\
		\rowcolor{gray!11}{} 
	2^{4}4^4&  D_4(2)\oplus D_4(-2)  
		&\begin{matrix}
			\circ
			\end{matrix}  
		&   -\frac{1}{3}\E_2+\E_4-\frac{2}{3}\E_8
		\\
	4^6 &  \ZZ(4)^{\oplus 4}\oplus \ZZ(-4)^{\oplus 2}
		&   \begin{matrix} n\notin 2\ZZ & \updownarrow\\ n\in 2\ZZ & \circ \end{matrix} 
			& -\frac{1}{6}\E_4+\frac{1}{2}\E_8-\frac{1}{3}\E_{16}\pm 8 \eta[2^44^{-4}8^4]
		\\
		\rowcolor{gray!11}{} 
	1^45^4 &  \Gamma^{2,2}\oplus \Gamma^{2,2}(5)
		&\begin{matrix}
			\circ
			\end{matrix} 
		& -\frac{5}{12}\E_5
		\\
	1^{-1}5^5 & A_4^*(5)  
		& \begin{matrix}n\in 1+5\ZZ & \updownarrow\\ n\notin 1+5\ZZ & \circ \end{matrix}  
		& 
		-\frac{25}{48}\E_5\mp\frac{25\sqrt{5}}{2}\eta[1^{-1}5^5]
		\\
		\rowcolor{gray!11}{}
	1^22^23^26^2 & \Gamma^{2,2}\oplus \Gamma^{2,2}(6)
		&\begin{matrix}
			\circ
		\end{matrix}  
		& \frac{1}{6}\E_2+\frac{1}{4}\E_3-\frac{1}{2}\E_6
		\\
	1^{4}2^13^{-4}6^5& \Gamma^{2,2}(2)\oplus A_2(2) 
		& \begin{matrix}
			\circ
		\end{matrix}   
		& \frac{1}{12}\E_2-\frac{1}{4}\E_3-\frac{1}{4}\E_6
		\\
		\rowcolor{gray!11}{} 
	1^{5}2^{-4}3^16^4 &\Gamma^{2,2}(3)\oplus A_2
		& \begin{matrix}
			\circ
		\end{matrix}   
		& -\frac{7}{12}\E_2+\frac{1}{8}\E_3-\frac{1}{4}\E_6
		\\
	1^{-2}2^43^{-2}6^4 & A_2(2)^{\oplus 2}
		& \begin{matrix}n=1 &  \updownarrow\\ n>1 & \circ \end{matrix}  
		& \begin{matrix}
			\frac{1}{3}\E_2+\frac{5}{4}\E_3-\E_6\\
			-\frac{2}{3}\E_2-\frac{3}{4}\E_3
		\end{matrix}
		\\
		\rowcolor{gray!11}{} 
	1^{-1}2^{-1}3^36^3& D_4(3)  
		& \begin{matrix} n\notin 3\ZZ & \updownarrow\\ n\in 3\ZZ & \circ \end{matrix}
		&  \begin{matrix}
			\frac{11}{12}\E_2+\frac{3}{8}\E_3-\frac{3}{4}\E_6\\
			-\frac{4}{3}\E_2-\frac{3}{8}\E_3
		\end{matrix}
		\\
	1^{-4}2^53^46^1& \Gamma^{2,2}(6)\oplus A_2(2)  
		& \begin{matrix}
		n\in 1+3\ZZ &	\circ,\circ^{*}\\
		n\notin 1+3\ZZ & \circ
		\end{matrix} 
		& -\frac{7}{12}\E_2-\frac{1}{4}\E_3-\frac{1}{4}\E_6
	\end{tabularx}
	\newpage
	\addtocounter{table}{-1}
	
	\begin{tabularx}{\linewidth}{CCCC}
		\pi_g
		& \Gamma^g 
		&
		\begin{matrix}
			\#\text{ Cosets}
		\end{matrix} 
		&   F_{g}(\tau) 
		\\
		\midrule\endhead
	2^36^3& A_2(2)^{\oplus 2}\oplus A_2(-2) 
		& \begin{matrix}
			\circ
			\end{matrix}   
			& -\frac{1}{4}\E_2-\frac{1}{4}\E_3+\frac{1}{6}\E_4+\frac{3}{4}\E_6-\frac{1}{2}\E_{12}
		\\ 
		\rowcolor{gray!11}{} 
	6^{4}&  D_4(3)  
		& \begin{matrix} n=1 & \updownarrow,\updownarrow \\\\ n>1,\ n\notin 3\ZZ & \updownarrow\\\\ n\in 3\ZZ & \circ \end{matrix}   
			& \begin{matrix}
			2\eta[1^22^23^26^{-2}]\\
			\hline
			2\eta[1^52^{-1}3^16^{-1}]\\
			2\eta[1^52^{-1}3^16^{-1}]+36\eta[6^4]
		\end{matrix}
		\\
	1^37^3&  \Gamma^{1,1}\oplus \Gamma^{1,1}(7)\oplus \left[\begin{smallmatrix}
			4 & 1\\ 1& 2
		\end{smallmatrix}\right] 
		&\begin{matrix}
			\circ
			\end{matrix}   
			& -\frac{7}{24}\E_7
		\\
		\rowcolor{gray!11}{} 
	1^22^14^18^2&  \Gamma^{1,1}\oplus \Gamma^{1,1}(8)\oplus \left[\begin{smallmatrix}
			2 & 0\\ 0& 4
		\end{smallmatrix}\right] 
		&\begin{matrix}
			\circ
			\end{matrix}   
		& \frac{1}{6}\E_4-\frac{1}{3}\E_8
		\\
	1^42^{-2}4^{-2}8^4&  D_4(2) 
		&  \begin{matrix} n\notin 2\ZZ & \updownarrow\\ n\in 2\ZZ & \circ\end{matrix}  
		&  \begin{matrix} \frac{1}{3}\E_2-\frac{2}{3}\E_4 \\ -\frac{5}{6}\E_2+\frac{1}{2}\E_4-\frac{1}{3}\E_8\end{matrix}
		\\
		\rowcolor{gray!11}{} 
	1^{-2}2^34^18^2 & \ZZ(4)\oplus A_3^*(8)
		&  \begin{matrix} n\in 1+8\ZZ &\updownarrow\\ n\notin 1+8\ZZ & \circ \end{matrix}  
		&  \begin{matrix} n\neq 1+6\ZZ\\ n\neq 1+8\ZZ\end{matrix}  
		\begin{matrix}
			-\frac{1}{3}\E_2+\frac{1}{6}\E_4-\frac{1}{3}\E_8\mp 16\sqrt{2}\eta[1^{-2}2^34^18^2]
		\end{matrix}
		\\
	2^44^{-4}8^4 & \ZZ(4)^{\oplus 4} 
		& \begin{matrix} n\notin 2\ZZ & \circ,\circ \\\\ n\in 2\ZZ & \circ\end{matrix}  
			& -\frac{1}{6}\E_4+\frac{1}{2}\E_8-\frac{1}{3}\E_{16}\pm 
		8\eta[2^44^{-4}8^4]
		\\
		\rowcolor{gray!11}{} 
	4^28^2  & \left[\begin{smallmatrix}
			4 & 0 & 0 & 0 \\
			0 & 4 & 0 & 0 \\
			0 & 0 & 8 & 0 \\
			0 & 0 & 0 & 8
		\end{smallmatrix}\right] 
		& \begin{matrix} n=1 & \updownarrow,\updownarrow\\\\ n>1,\ n\notin 4\ZZ & \updownarrow\\\\ n\in 4\ZZ & \circ \end{matrix}  
		& \begin{matrix}
			16\eta[4^48^{-4}16^4]+2\eta[2^44^28^{-2}] - 8\eta[4^28^2]\\
			16\eta[4^48^{-4}16^4]+2\eta[2^44^28^{-2}] + 24\eta[4^28^2]\\
			\hline
			2\eta[2^44^28^{-2}]
		\end{matrix}
		\\
	1^33^{-2}9^3 &  A_2\oplus A_2(3)
		& \begin{matrix} n=1 & \circ, \circ \\ n>1 & \circ\end{matrix}  
		& -\frac{1}{8}\E_3-\frac{3}{16}\E_9\pm\frac{9}{2}\eta[1^33^{-2}9^3]
		\\
		\rowcolor{gray!11}{} 
	1^{2}2^15^{-2}10^3 &A_4(2)   
		& \begin{matrix} n=1 & \updownarrow\\ n>1 & \circ\end{matrix} 
		& \frac{1}{24}\E_2-\frac{5}{24}\E_{10}\pm 2\sqrt{5}\eta[1^22^15^{-2}10^3]
		\\
	1^{3}2^{-2}5^110^2  &  A_4^*(5) 
		& \begin{matrix} n\in 1+5\ZZ & \updownarrow\\ n\notin 1+5\ZZ & \circ \end{matrix} 
			& -\frac{7}{24}\E_2+\frac{5}{48}\E_5-\frac{5}{24}\E_{10}\mp\frac{5\sqrt{5}}{2}\eta[1^32^{-2}5^110^2]
		\\
		\rowcolor{gray!11}{} 
	1^{-2}2^35^210^1 & A_4^*(10) 
		& \begin{matrix} n\in 1+5\ZZ & \updownarrow\\ n\notin 1+5\ZZ & \circ \end{matrix}  
			& -\frac{7}{24}\E_2-\frac{5}{24}\E_{10}\mp 10\sqrt{5}\eta[1^{-2}2^35^210^1]
		\\
	2^210^2 &  \left[\begin{smallmatrix}
					6 & 4 & 0 & 0\\  4 & 6 & 0& 0\\ 0&0& 6&4\\ 0& 0& 4& 6
				\end{smallmatrix}\right],\ \left[\begin{smallmatrix}
					2 & 0 & 0 & 0\\  0 & 2 & 0& 0\\ 0&0& 10&0\\ 0& 0& 0& 10
				\end{smallmatrix}\right]
		& \begin{matrix} n=1 &  \updownarrow, \circ,  \circ, \circ\\ n>1 & \circ \end{matrix} 
		&
		\begin{matrix}
			-\frac{1}{12}\E_2+\frac{1}{18}\E_4-\frac{5}{36}\E_5+\frac{5}{12}\E_{10}-\frac{5}{18}\E_{20}-\frac{20}{3}
			\eta[2^210^2]\\
			-\frac{1}{12}\E_2+\frac{1}{18}\E_4-\frac{5}{36}\E_5+\frac{5}{12}\E_{10}-\frac{5}{18}\E_{20}+\frac{40}{3}\eta[2^210^2]
		\end{matrix}  
		\\
		\rowcolor{gray!11}{} 
	1^211^2&  \left[\begin{smallmatrix}
			4 & 2 & 1 & 1 \\
			2 & 4 & 0 & 1 \\
			1 & 0 & 4 & 2 \\
			1 & 1 & 2 & 4
		\end{smallmatrix}\right],\left[\begin{smallmatrix}
			2 & 0 & 1 & 0 \\
			0 & 2 & 0 & 1 \\
			1 & 0 & 6 & 0 \\
			0 & 1 & 0 & 6
		\end{smallmatrix}\right],\left[\begin{smallmatrix}
			2 & 1 & 1 & 1 \\
			1 & 2 & 0 & 1 \\
			1 & 0 & 8 & 4 \\
			1 & 1 & 4 & 8
		\end{smallmatrix}\right] 
		& 
		\begin{matrix} n=1 & \updownarrow, \circ ,  \circ \\ n>1 & \circ\end{matrix}
		& \begin{matrix}
			-\frac{11}{60}\E_{11} -\frac{22}{5}\eta[1^211^2]\\
			-\frac{11}{60}\E_{11} + \frac{33}{5}\eta[1^211^2]
		\end{matrix}
		\\
	1^{2}2^{-2}3^24^{2}6^{-2}12^2  &  A_2(2)\oplus A_2(2)
		& \begin{matrix} n=1 & \updownarrow\\ n>1 & \circ\end{matrix}  
		& \begin{matrix} \frac{1}{6}\E_2+\frac{1}{4}\E_3-\frac{1}{2}\E_6 \\ 
			-\frac{13}{12}\E_2-\frac{1}{4}\E_3+\frac{1}{2}\E_4+\frac{3}{4}\E_6-\frac{1}{2}\E_{12}\end{matrix}
		\\
		\rowcolor{gray!11}{} 
	1^{1}2^23^14^{-2}12^2  &  A_2\oplus \ZZ(6)^{\oplus 2}
		& \begin{matrix}  n=1 & \circ, \circ\\ n>1 & \circ\end{matrix}  
			& -\frac{1}{24}\E_2+\frac{1}{12}\E_4+\frac{1}{8}\E_6-\frac{1}{4}\E_{12}\pm 3\sqrt{3}\eta[1^12^23^14^{-2}12^2]
		\\
	1^{2}3^{-2}4^16^212^1  &  A_2(4)\oplus \ZZ(2)^{\oplus 2}
		& \begin{matrix}
			n=1 & \circ, \circ\\ n>1 & \circ
			\end{matrix}  
			& -\frac{1}{12}\E_2-\frac{1}{4}\E_3+\frac{1}{12}\E_4+\frac{1}{4}\E_6-\frac{1}{4}\E_{12}\pm 4\sqrt{3}\eta[1^23^{-2}4^16^212^1]
		\\
		\rowcolor{gray!11}{} 
	1^{-2}2^23^24^112^1 & \ZZ(6)^{\oplus 2}\oplus A_2(4) 
		& \begin{matrix} n\in 1+4\ZZ & \updownarrow\\ n\notin 1+4\ZZ & \circ\end{matrix}  
		& 
		-\frac{5}{12}\E_2-\frac{1}{4}\E_3+\frac{1}{12}\E_4+\frac{1}{4}\E_6-\frac{1}{4}\E_{12}\mp 12\sqrt{3}\eta[1^{-2}2^23^24^112^1]
		\\
	2^14^16^112^1& \left[\begin{smallmatrix}
			4 & 2 & 0 & 0\\ 2 & 4 & 0 & 0\\ 0 & 0 & 8 & 4\\ 0& 0& 4& 8
		\end{smallmatrix}\right] 
		& 
		\begin{matrix} n=1 & \updownarrow,\circ,\circ\\ n>1 & \circ \end{matrix}
			& \begin{matrix}
			\frac{1}{24}\E_2-\frac{1}{8}\E_4-\frac{1}{8}\E_6+\frac{1}{12}\E_8+\frac{3}{8}\E_{12}-\frac{1}{4}\E_{24} - 6\eta[2^14^16^1 12^1]\\
			\frac{1}{24}\E_2-\frac{1}{8}\E_4-\frac{1}{8}\E_6+\frac{1}{12}\E_8+\frac{3}{8}\E_{12}-\frac{1}{4}\E_{24} +18\eta[2^14^16^1 12^1]
		\end{matrix}
		\\
		\rowcolor{gray!11}{} 
	1^12^17^114^1&  \left[\begin{smallmatrix} 4 & 1 & 1 & 0 \\
			1 & 4 & 0 & 1 \\
			1 & 0 & 4 & -1 \\
			0 & 1 & -1 & 4\end{smallmatrix}\right],\left[\begin{smallmatrix}
			2 & 0 & 1 & 1 \\
			0 & 2 & 1 & 1 \\
			1 & 1 & 8 & 1 \\
			1 & 1 & 1 & 8
		\end{smallmatrix}\right],\left[\begin{smallmatrix}
			2 & 1 & 0 & 0 \\
			1 & 4 & 0 & 0 \\
			0 & 0 & 4 & 2 \\
			0 & 0 & 2 & 8
		\end{smallmatrix}\right] 
		&\begin{matrix} n=1 & \updownarrow,  \circ ,  \circ \\ n>1 & \circ \end{matrix}
		& 
		\begin{matrix}
			\frac{1}{36}\E_2+\frac{7}{72}\E_7-\frac{7}{36}\E_{14}-\frac{14}{3}\eta[1^12^17^114^1]\\
			\frac{1}{36}\E_2+\frac{7}{72}\E_7-\frac{7}{36}\E_{14}+\frac{28}{3}\eta[1^12^17^114^1]
		\end{matrix}
		\\
	1^13^15^115^1 & \left[\begin{smallmatrix} 4 & 2 & 1 & 1 \\
			2 & 4 & -1 & 2 \\
			1 & -1 & 6 & 2 \\
			1 & 2 & 2 & 6 \end{smallmatrix}\right],
		\left[\begin{smallmatrix}
			2 & 1 & 0 & 0 \\
			1 & 2 & 0 & 0 \\
			0 & 0 & 10 & 5 \\
			0 & 0 & 5 & 10
		\end{smallmatrix}\right],
		\left[\begin{smallmatrix}
			2 & 0 & 0 & 1 \\
			0 & 4 & 1 & 0 \\
			0 & 1 & 4 & 0 \\
			1 & 0 & 0 & 8
		\end{smallmatrix}\right]  
		& \begin{matrix}n=1 &  \updownarrow,  \circ ,  \circ\\ n>1 & \circ \end{matrix}
		& 
		\begin{matrix}
			\frac{1}{32}\E_3+\frac{5}{96}\E_5-\frac{5}{32}\E_{15}-\frac{15}{4}\eta[1^13^15^115^1]\\
			\frac{1}{32}\E_3+\frac{5}{96}\E_5-\frac{5}{32}\E_{15}+\frac{45}{4}\eta[1^13^15^115^1]
		\end{matrix}
		\\
		\bottomrule \caption{\small The first column contains all possible Frame shapes $\pi_g$ of a symmetry $g$ of a NLSM on K3. For each of them, in the second column we report the possible fixed sublattices $(\Gamma^{4,20})^g$ in the action of $g$ on $\Gamma^{4,20}$, that were derived in \cite{Persson:2015jka}.  The third column contains the number of primitive embeddings of the lattice $\langle v\rangle\oplus \Gamma_g$, with $v^2=2n-2$, in $\Gamma^{5,21}$ up to $O^+(\Gamma^{5,21})$ automorphisms. The techniques used for this computation are the theorems in \cite{MirandaMorrison1,MirandaMorrison2}, and the necessary information for the calculation is contained in 
			in appendix \ref{a:discriminants}. The notation is as follows: for each possible value of $n$, we report a number of circles $\circ$ and double arrows $\updownarrow$; each circle represents a class that is self-conjugate under $O(\Gamma^{5,21})\setminus O^+(\Gamma^{5,21})$ transformations, while the double arrows represent a couple of $O^+(\Gamma^{5,21})$-classes giving rise to a unique $O(\Gamma^{5,21})$-class. When we don't specify the value of $n$, it means that the result is valid for all $n$. The reason for counting these classes is explained  in section \ref{s:conjclass}. The last column contains the possible modular forms $F_g(\tau)$ that determine the twining genus $\phi_g$, see eq.\eqref{twinformula}. The modular forms are expressed in terms of Eisenstein series and products of $\eta$-series; see \cite{Paquette:2017gmb} for the precise notation. In some cases, the possible twining genera $\phi_g$ for a given Frame shape are Jacobi forms with a non-trivial complex-valued multiplier. This happens if and only if the length of the shortest cycle in the Frame shape is $>2$, see \cite{Cheng:2016org}. In this case, there are always at least two different twining genera at $n=1$, with complex conjugate multipliers \cite{Cheng:2016org}. For the Frame shapes $6^4$ and $4^28^2$, where there are more than two candidate twining genera, we separate by a horizontal line the twining genera with different multipliers. For the Frame shape $1^{-4}2^53^46^1$, in the case $n\in 1+3\ZZ$, we were not able to determine whether there are one or two cosets; in any case, if there are two cosets, they correspond to symmetries $g$ being one the inverse of the other, so that the twining genera are the same (see the analogous problem in \cite{Cheng:2016org}). }
		\label{tab:big}
	\end{tabularx}
\end{landscape}

\section{Discussion}\label{s:conclusions}

In this final section, we discuss some possible directions of investigation for future works.

\begin{itemize}
	\item As discussed in section \ref{s:dualities}, it would be useful to have a complete and rigorous characterization of the duality groups $O^+_n$ of SCFTs of $K3^{[n]}$-type, analogous to the one valid for NLSMs on K3 (see \cite{Aspinwall:1996mn}). In particular, it should be possible to prove (based on very basic assumptions on the moduli space $\M_n$, such as the fact that it is Hausdorff), that the group $O_n^+$ always acts by automorphisms of the lattice $L_v$. This result would already give further support to the analysis of section \ref{s:symmorbmoduli}. Determining which precise subgroup of $O(L_v)$ corresponds to the duality group seems more difficult.
	\item The precise location of the locus $\M^{sym}_n$ inside the moduli space $\M_n$ seems to be still an open problem. This is probably related with the precise identification of the holographic duals of symmetric orbifolds (see for example \cite{Argurio:2000tb,Giribet:2018ada,Gaberdiel:2018rqv,Eberhardt:2018ouy,Eberhardt:2019qcl} for older and more recent results on this subject). It would be interesting to combine the results of our article, which arise only from considerations about the boundary CFT, with the analyses in these papers.
	\item The study of symmetric orbifolds leads to some puzzles about the location of singular models in the moduli space.  Some puzzling features of symmetric orbifolds were already stressed in \cite{Seiberg:1999xz}. For example, as observed in \cite{Seiberg:1999xz}, it seems natural to expect that a system of $N$ fundamental strings and one $NS5$-brane (or, dually, an analogous D1-D5 brane system) in the absence of R-R flux be described by a non-linear sigma model on the symmetric orbifold of $N$ copies of K3. On the other hand, the very analysis in \cite{Seiberg:1999xz} suggests that at such a point in the moduli space the SCFT should be singular,  while the symmetric orbifold seems to be perfectly consistent. Thus,  the physical intuition seems to lead to wrong conclusions in this case, but is not clear what the loophole is. Clarifying these issues might lead to a major improvement of our results.
	\item The classification of symmetries in section \ref{s:callifsymm} is conditional to certain assumptions and subject to some restrictions. Can we obtain a more general statement by removing some of these restriction? In particular, if the models corresponding to four-planes $\Pi\subset L_v\otimes \RR$ that are orthogonal to roots $r\in L_v$ , $r^2=-2$, are consistent, one might want to include their groups of symmetries in a more general classification. This problem is superficially similar to the one studied in \cite{Cheng:2016org}, where the results of \cite{K3symm} concerning non-linear sigma models on K3 were extended to include the `singular' K3 models, in order to obtain the full group of symmetries of the type IIA string theory on K3. However, there are some technical difficulties in extending the results of the present paper to include those models -- in particular, it is not obvious that one can always embed the resulting symmetry groups in the groups of automorphisms of Niemeier lattices, which was one of the main tools used in \cite{Cheng:2016org}.
	\item As explained in section \ref{s:elevenandfriends}, our results suggest that there exist a number of highly non-trivial identities among Siegel modular forms that can be written as Borcherds products. From a mathematical perspective, it would be nice if one could prove those identities rigorously, maybe along the lines described in section \ref{s:elevenandfriends}. From a physicists' viewpoint, those Siegel forms represent the generating functions for the multiplicities of 1/4 BPS dyons in compactifications of type II string theory on $K3\times T^2$ or their orbifolds \cite{Strominger:1996sh,Dijkgraaf:1996it,Dijkgraaf:1996xw,gaiotto2005re, Shih:2005uc, shih2006exact,  Jatkar:2005bh,DS, DJS1, DJS2}. The fact that two such generating functions differ only for the constant term means that, for certain  pairs of duality orbits of charges, the 1/4 BPS multiplicities are different only when the electric and magnetic charges $Q$  and $P$ are null and orthogonal $P^2=Q^2=P\cdot Q=0$. It would be interesting to understand the physical meaning of this difference.\\
	These identities might be relevant in trying to understand the relationship between the K3 twining genera and the weak Jacobi forms appearing in Mathieu \cite{EOT,Cheng:2010pq,Gaberdiel:2010ch,Gaberdiel:2010ca,Eguchi:2010fg,Gannon:2012ck} and Umbral moonshine \cite{Cheng:2012tq,Cheng:2013wca,Cheng:2014zpa,DuncanGriffinOno2015}. The Siegel modular forms appear in certain low energy effective couplings in compactifications of type II  superstrings on CHL models \cite{Bossard:2016zdx,Bossard:2017wum,Bossard:2018rlt}. Such effective couplings might be important in trying to relate the moonshine conjectures to low dimensional string models whose symmetry groups are exactly the umbral groups. \cite{Kachru:2016ttg,Zimet:2018dev}.
	Finally, similar identities might exist also for the Siegel modular forms related to the second quantized twisted-twining genera \cite{Persson:2013xpa}, which play a role in generalized moonshine \cite{Gaberdiel:2012gf,Gaberdiel:2013nya,Cheng:2016nto}.
\end{itemize}

\bigskip
{\bf Acknowledgements.}
I would like to thank Sarah Harrison, Shamit Kachru, Natalie Paquette,  Daniel Persson, Anne Taormina, Katrin Wendland, and especially Max Zimet for useful discussions. I would like to thank the organizers and participants of the workshop on Moonshine, organized at ESI Vienna, in September 2018, and of the program Automorphic Structures in String Theory at Simons Center for Geometry and Physics, Stony Brook, in March 2019, where some of the ideas that led to this article were developed. This work is supported by a grant from Programma per Giovani Ricercatori Rita Levi Montalcini, from MIUR.

\appendix
\section*{Appendices}

\section{Moduli spaces}
\subsection{The  moduli space of type II on K3}\label{s:stringsonK3}

In this section, we  describe the moduli space of type IIA and type IIB string theories compactified on a K3 surface. It is useful to consider first the moduli space of supersymmetric NLSM on K3, describing the world-sheet of a fundamental type II superstring in weak string coupling limit. Our main references are \cite{Aspinwall:1996mn,Dijkgraaf:1998gf}.

The moduli space of NLSM on a K3 surface $X$ corresponds to a choice of an hyperk\"ahler structure and a closed B-field on $X$. For fixed volume $V$, the hyperk\" ahler structure is parametrized by a choice of an oriented positive definite 3-plane
\be U=\langle \omega_1,\omega_2,\omega_3\rangle\subset H^2(X,\RR)\cong \RR^{3,19}\ ,
\ee  spanned by three positive pairwise orthogonal $\omega_1,\omega_2,\omega_3\in \RR^{3,19}$.    The 3-plane is spanned by the real and imaginary part of $\Omega\in H^{2,0}(X,\CC)$, and the K\"ahler class of the metric. For example, one can take $\Omega=\omega_2+i\omega_3$ and $\omega_1$ as the K\"ahler class. The volume form can be seen as an element $V\in H^4(X,\RR)$, and it is useful to normalize $\omega_1,\omega_2,\omega_3$ by
\be \frac{1}{2}(\omega_1^2+\omega_2^2+\omega_3^2)=V\ .
\ee The B-field takes values in $H^2(X,\RR)$ (modulo $H^2(X,\ZZ)$). This determines an oriented positive definite\footnote{The metric on $H^*(X,\RR)$ is induced by linearity by Mukai pairing on the lattice $H^*(X,\ZZ)$; this differs from the intersection form by the sign of the $H^0(X,\RR)\oplus H^4(X,\RR)$ component. The lattice $H^*(X,\ZZ)$ is even unimodular of signature $(4,20)$. Notice that all such lattices are isomorphic to $\Gamma^{4,20}=U^{\oplus 4}\oplus E_8\oplus E_8$, where $U$ is generated by $u,v$ with $u^2=v^2=0$, $u\cdot v=1$, and $E_8$ is the root lattice of $E_8$.} $4$-plane $\Pi$ in $H^*(X,\RR)=H^0(X,\RR)\oplus H^2(X,\RR)\oplus H^4(X,\RR)$ spanned by
\be \Pi=\langle 1+B-V+\frac{1}{2}B\wedge B, \omega_1+B\wedge \omega_1, \omega_1+B\wedge \omega_1, \omega_1+B\wedge \omega_1\rangle\subset H^*(X,\RR)\ .
\ee The moduli space of NLSM on K3 is 
\be \M_{K3}=O^+(\Gamma^{4,20})\backslash Gr^+(4,20) \,\setminus \calS \ ,
\ee
where
\be  Gr^+(4,20)=O^+(4,20,\RR)/(SO(4,\RR)\times O(20,\RR))
\ee is the Grassmannian of oriented positive definite $4$-dimensional subspaces of $H^*(X,\RR)\cong \RR^{4,20}$ and
\be \calS:=\{\Pi\in Gr^+(4,20)\mid \Pi\subset u^\perp \text{ for some } u\in H^*(X,\ZZ),\ u^2=-2\}\ ,
\ee is the locus where the NLSM is not well-defined \cite{Aspinwall:1996mn}. In the limit where one approaches the locus $\calS$, some D-brane becomes massless and cannot be decoupled from the  dynamics of the fundamental string, even in the weak string coupling limit.

Type IIA string theory on K3 is perfectly well defined even on the locus $\calS$, and the moduli space of type IIA string theory on K3 is $\M_{IIA}=(O^+(\Gamma^{4,20})\backslash Gr^+(4,20))\times \RR_+$, where $\RR_+$ represents the string coupling. D-branes carry R-R charge taking values in $H^*(X,\ZZ)\cong \Gamma^{4,20}$. A D-brane carrying charge $v\in H^*(X,\ZZ)$, $v^2>0$, has mass proportional to $v_L^2=v^2+v_R^2$, where $v_L$ and $v_R$ are, respectively, the components of $v$ along $\Pi$ and orthogonal to $\Pi$. For a given D-brane charge $v$, there is a locus of attractor moduli in $\M_{K3}$ where the mass is minimized. This means that $v_R=0$, i.e. the four plane $\Pi$ must contain $v$.

For type IIB on K3, besides the metric and B-field on K3 and the string coupling, there are in addition RR moduli $\C\in H^*(X,\RR)$ (modulo $H^*(X,\ZZ)$). These moduli determine an oriented positive definite $5$-plane in $H^*(X,\RR)_{RR}\oplus H^0(X,\RR)_{NS}\oplus H^4(X,\RR)_{NS}$ (there are two copies of $H^0$ and $H^4$, one related to D1 and D5 branes, and the other to fundamental strings and NS5 branes) spanned by
\be (\Pi ,0,\C\cdot \Pi),\quad (\C,1,1/g_s)\quad \in H^*(X,\RR)_{RR}\oplus H^0(X,\RR)_{NS}\oplus H^4(X,\RR)_{NS}\ .
\ee More explicitly, if
\be \C=(\theta,\tilde B,G)\in H^0(X,\RR)\oplus H^2(X,\RR)\oplus H^4(X,\RR)
\ee then the $5$-plane is spanned by
\be \begin{matrix}
	H^*(X,\RR)_{RR}&\oplus& H^0(X,\RR)_{NS}&\oplus& H^4(X,\RR)_{NS}\\ \hline
	1+B-V+\frac{1}{2}B\wedge B,&&  0,&& -G+\tilde B\cdot B -\theta\cdot(\frac{1}{2}B\wedge B-V)
	\ ,\\
	\theta+\tilde B+G,&& 1,&& 1/g_s\\
	\omega_i+B\wedge \omega_i,&&0,&&\tilde B\cdot \omega_i-\theta\cdot (B\wedge \omega_i)
\end{matrix}
\ee where $i=1,2,3$. Consider a charge vector $(v,0,0)\in H^*(X,\ZZ)_{RR}\oplus H^0(X,\RR)_{NS}\oplus H^4(X,\RR)_{NS}$, with $v=(r,0,-p)\in  H^0(X,\ZZ)\oplus H^2(X,\ZZ)\oplus H^4(X,\ZZ)=H^*(X,\ZZ)_{RR}$, $r,p>0$. This charge corresponds to bound state of D1-D5 branes, carrying $r$ units of D5-brane charge and $p$ units of D1-brane charge. The attractor manifold for this system is given by the Grassmannian of five planes $\Pi$ that contain $v$. This means that $v$ must be a linear combination of the five vectors above. This gives the conditions
\be rB+\sum_i\zeta_i\omega_i=0\ ,\qquad V+\frac{1}{2}B^2=\frac{p}{r}\ ,\qquad rG+p\theta=0\ .
\ee The first condition means that the NS B-field must be contained in the $3$-dimensional subspace $\langle \omega_1,\omega_2,\omega_3\rangle\subset H^2(X,\RR)$, i.e. must be self-dual. The second condition fixes the volume $V$ in terms of the charges; for $B=0$, this is just the ratio $p/r$. Finally, the third condition implies that the RR moduli must have the form
\be \C=\beta (r,0,p)+(0,\tilde B,0)\in H^*(S,\RR)/H^*(S,\ZZ),
\ee where $\beta\in \RR$ and $\tilde B$ is an arbitrary class in $H^2(S,\RR)/H^2(S,\ZZ)$.

\subsection{Non-linear sigma models on $K3^{[n]}$}\label{a:Hilbertscheme}

The worldvolume theory of the D1-D5 system wrapping $K3\times S^1\times \RR_t$, in the limit where the volume of the K3 surface is much smaller than the volume of the circle $S^1$ ($V^{1/4}\ll R_{S^1}$), is described by a non-linear sigma model with target space a hyperk\"ahler manifold of dimension $4n$, $n=rp>1$, in the same moduli space of the symmetric product $\Sym^n S$ of the K3 surface $S$. Mathematically, these hyperk\"ahler manifolds can be described as deformations of $S^{[n]}$, the Hilbert scheme of $n$ points of $S$. The latter is a crepant resolution $\pi:S^{[n]}\to \Sym^nS$ of $\Sym^nS$, which is singular at the `small diagonal', the locus of points $(x_1,\ldots,x_n)\in Sym^n S$ where $x_i=x_j$ for some $i\neq j$.  

Let $X$ be a generic deformation of some Hilbert scheme $S^{[n]}$. The integral second cohomology $H^2(X,\ZZ)$ can be endowed with a non-degenerate even bilinear form (the Beaville-Bogomolov form); with this bilinear form $H^2(X,\ZZ)$ is isomorphic to an even lattice of signature $(3,20)$
\be H^2(X,\ZZ)\cong \Gamma^{3,19}\oplus \langle 2-2n\rangle\ .
\ee  
A non-linear sigma model on a hyperk\"ahler manifold $X$ of type K3$^{[n]}$, $n>1$, is specified by a hyperk\"ahler structure on $X$ and a closed B-field. For fixed volume $V$, the hyperk\" ahler structure is parametrized by a choice of an oriented positive definite 3-plane
\be U=\langle \omega^{[n]}_1,\omega^{[n]}_2,\omega^{[n]}_3\rangle\subset H^2(X,\RR)\cong \RR^{3,20}\ ,
\ee  spanned by three positive pairwise orthogonal $\omega^{[n]}_1,\omega^{[n]}_2,\omega^{[n]}_3\in \RR^{3,20}$. The B-field $B^{[n]}$ takes values in $H^2(X,\RR)$ (modulo $H^2(X,\ZZ)$).
When $X$ is a Hilbert scheme $X=S^{[n]}$ of $n$ points on a K3 surface $S$, there is a canonical lattice isomorphism (which extends by linearity to an isomorphism of real and complex cohomology)
\be H^2(S^{[n]},\ZZ)\cong H^2(S,\ZZ)\oplus \ZZ \delta\ ,
\ee where the first summand  is spanned by pullbacks $\pi^* f^{(n)}$ of symmetrization $f^{(n)}$ of integral cohomology classes $f\in H^2(S,\ZZ)$ on $S$, while the second summand is spanned by a class $\delta \in H^2(S^{[n]},\ZZ)$, $\delta^2=2-2n$, such that $2\delta$ is Poincar\'e dual to the exceptional divisor. The complex structure on $S$ induces a complex structure on $S^{[n]}$: if $\Omega=\omega_2+i\omega_3$ spans $H^{2,0}(S)$, then its image $\Omega^{[n]}=\pi^*\Omega^{(n)} \in H^{2,0}(S^{[n]})$ (the image under the isomorphism above) spans $H^{2,0}(S^{[n]})$. Note that, with respect to this complex structure, the class $\delta$ is always contained in the $H^{1,1}(S^{[n]})$ component of the Hodge cohomology. A generic manifold $X$ in the moduli space can be obtained by the deforming the complex structure of the Hilbert scheme, so that $H^{2,0}(X)$ is spanned by
\be \Omega^{[n]}=\pi^*\Omega^{(n)}+(\alpha_1+i\alpha_2)\delta\ ,
\ee for some $\alpha_1,\alpha_2\in \RR$. A K\"ahler class $\omega_1\in H^{1,1}(X)$ for the K3 surface determines a $1$-parameter family of Kahler classes on its Hilbert scheme $S^{[n]}$
\be \omega^{[n]}=\pi^* \omega_1^{(n)}+\lambda u\ ,\qquad \lambda\in \RR_+\ .
\ee In the limit $\lambda\to 0$, one obtains the singular geometry corresponding to the symmetric product $\Sym^n S$. Similarly,  a B-field $B$ on $S$ determines a $1$-parameter family of B-fields on $S^{[n]}$
\be B^{[n]}=\pi^* B^{(n)}+\beta  u\ ,\qquad \beta\in \RR/\ZZ\ .
\ee Here, $\beta$ represents the B-field flux through the the exceptional divisor of $S^{[n]}$. Finally, the volume of $S^{[n]}$ can be fixed to be simply the $n$-th power $V^n$ of the volume of $S$.


\section{Some basic facts about lattices}\label{a:lattices}

In this appendix, we review some useful facts about lattices. See \cite{ConwaySloane}, \cite{Nikulin} for more information.

A lattice $L$ of rank $n$ (over $\ZZ$) is a free $\ZZ$-module $L\cong \ZZ^n$ of rank $n$, with a quadratic form $Q:L\to \RR$, $$Q(av)=a^2Q(v)\ ,\qquad a\in \ZZ\ ,$$ and a symmetric $\ZZ$-bilinear form $B:L\times L\to \RR$  with
\be Q(v)=B(v,v)\qquad v\in L
\ee so that
\be  B(v,w):=\frac{1}{2}(Q(v+w)-Q(v)-Q(w))\ ,\qquad v,w\in L\ .
\ee  We will also use the notation $v\cdot w$ or $(v,w)$ for the bilinear form $B$. A vector $v\in L$ is called primitive if it is not an integral multiple $mw$, $m>1$,  of another vector $w\in L$. An embedding $S\hookrightarrow L$ of a lattice $S$ into a lattice $L$ is called primitive if $L/S$ is free; equivalently, if every primitive vector of $S$ is mapped to a primitive vector of $L$.

One can think of a lattice $L$ as being embedded in the real space $L\otimes \RR$  with bilinear form induced by extending $B$ by $\RR$-linearity. The signature of $L$ is the signature of the bilinear form $B$ on the real vector space $L\otimes \RR$. The lattice is non-degenerate if the only vector $v\in L$ such that $B(v,w)=0$ for all $w$ is $v=0$; in the following,  we implicitly assume that lattices are non-degenerate.  

The lattice $L$ is called integral if the bilinear form $B$ takes values in $\ZZ$, and even if the quadratic form $Q$ takes values in $2\ZZ$; notice that even lattices are also integral. The dual lattice $L^*$ is the lattice $\Hom_\ZZ(L,\ZZ)$ of $\ZZ$-linear homomorphisms from $L$ to $\ZZ$, and can be identified with a lattice in $L\otimes \RR$, by
\be L^*=\{x\in L\otimes \RR\mid B(x,v)\in \ZZ,\ \forall z\in L\}\ .
\ee A lattice $L$ is called self-dual (or unimodular), if $L^*\cong L$. An even unimodular lattice of signature $(a,b)$ exists if and only if $a-b\equiv 0\mod 8$. Two indefinite even unimodular lattices  with the same signature are always isomorphic; we will denote by $\Gamma^{a,b}$ such a lattice. Given a basis of generators $v_1,\ldots,v_n$ of $L$, the Gram matrix of $L$ is the $n\times n$ real symmetric matrix with entries $B_{ij}=B(v_i,v_j)$. Notice that the Gram matrix of $L^*$ is the inverse of the one of $L$.   The group $O(L)$ of automorphisms of the lattice $L$ is the subgroup of the real orthogonal group on $\RR^n$ preserving the lattice $L$.

If $L$ is integral, then $L\subset L^*$, and $A_L:=L^*/L$ is a finite abelian group called the discriminant group. The order $|A_L|$ of $A_L$ is called the discriminant of $L$, and is equal to the determinant of the Gram matrix. We will mostly be interested in discriminant groups of even lattices. The discriminant group $A_L$ of an even lattice $L$ has a quadratic form (the discriminant form) $q_L:A_L\to \QQ/2\ZZ$ induced by the quadratic form on $L^*$ and $L$.  The group $O(A_L)\equiv O(q_L)$ is the group of automorphisms of the abelian group $A_L$ that preserve the quadratic form $q_L$. There is a natural homomorphism $O(L)\to O(q_L)$, which is in general neither injective nor surjective. The kernel of this homomorphism is the subgroup $O^0(L)\subset O(L)$ acting trivially on $L^*/L$. As shown in \cite{Nikulin} (theorem 1.14.2), if an even lattice $L$ contains a copy of the even unimodular lattice $\Gamma^{1,1}$, then $O(L)\to O(q_L)$ is surjective. Two even lattices have the same signature and isomorphic discriminant forms if and only if they are in the same genus (see \cite{Nikulin} for a proof and \cite{ConwaySloane} for the definition of genus). More information about discriminant forms will be given in appendix \ref{a:discriminants}.

Let $M$ be an even lattice, $L$ be an even unimodular lattice, and $M\hookrightarrow L$ a primitive embedding. Then, if $K:=M^\perp \cap L$ is the sublattice of $L$ orthogonal to $M$ (here, identified with the image in $L$ of the primitive embedding), one has the inclusions
\be M\oplus  K\subseteq L\subseteq M^*\oplus K^*\ .
\ee
Furthermore, there is an isomorhpism $\gamma:A_M\stackrel{\cong}{\rightarrow} A_K$ of abelian groups that inverts the discriminant form, i.e. 
\be q_K\circ \gamma=-q_M\ ,
\ee and such that 
\be\label{glue} L=\{(x,y)\in M^*\oplus K^*\mid \bar y\equiv \gamma(\bar x)\}\ ,
\ee where $\bar x$ and $\bar y$ are the images of $x\in M^*$ and $y\in K^*$ in the discriminant groups $A_M=M^*/M$ and $A_K=K^*/K$, respectively (see \cite{Nikulin}, proposition 1.6.1). Vice versa, if $M$ and $K$ are two even lattices for which there exists an isomorphism $\gamma:A_M\stackrel{\cong}{\rightarrow} A_K$ that inverts the discriminant form, then the lattice $L$ defined by \eqref{glue} is even unimodular and contains $M$ and $K$ as primitive sublattices. In this case, we say that $M$ and $K$ have been `glued' to obtain $L$.

 Let $L$ be an even lattice. Given a primitive vector $x\in L$, the divisor $\Div(x)$ is the maximal positive integer such that $\frac{x}{\Div(x)}\in L^*$ (and this implies that $\frac{x}{\Div(x)}$ is primitive in $L^*$). Thus $\frac{x}{\Div(x)}$ modulo $L$ is a generator of order $\Div(x)$ of the discriminant group. This implies that $\Div(x)$ divides the order of the discriminant group. Furthermore, since $\frac{x}{\Div(x)}\in L^*$, we have $\frac{(x,x)}{\Div(x)}\in \ZZ$, so that $\Div(x)$ divides also $(x,x)$.

Now, given $v\in L$ of norm $2n-2$, one has that $\Div(v)$ divides both $2n-2$ and the order of $A_L$. In particular, when $2n-2$ is coprime to $|A_L|$, one has $\Div(v)=1$.

Clearly, a necessary condition for two primitive vectors $u,v\in L$ to be related by some automorphism in $O(L)$ is that they have the same length and the same divisor. For certain lattices $L$, this condition is also sufficient. The following criterion is well-known.

%

\begin{proposition} (Eichler criterion). Let $L$ be an even lattice containing two copies $U, U_1\cong \Gamma^{1,1}$ of the even unimodular lattice of signature $(1,1)$, so that $L=U\oplus U_1\oplus L_0$. Then any two vectors $u,v\in L$ that are primitive, with the same length $(u,u)=(v,v)$ and such that $\frac{u}{\Div(u)}\equiv \frac{v}{\Div(v)}\mod L$ are related by an $O^+(L)$ transformation acting trivially on the discriminant form.
\end{proposition}  See \cite{GHS2009} proposition 3.3 for a proof.  We notice that a primitive null vector $e\in L$ can be a generator of a hyperbolic lattice $U\subset L$ (i.e. there exist a null $f\in L$ such that $(e,f)=1$) if and only if $\Div(e)=1$. 

Let $A_L=L^*/L$ be the discriminant group of the even lattice $L$; there is an orthogonal (with respect to the induced bilinear form $A_L\times A_L\to \QQ/\ZZ$) decomposition $A_L=\oplus_{p\text{ prime}} A_L^{(p)}$ of $A_L$ into its Sylow $p$-subgroups $A_L^{(p)}$ (the maximal subgroups of order a power of $p$). Let $l(A_L^{(p)})$ be the minimal number of generators of $A_L^{(p)}$. By Corollary 1.13.5 of \cite{Nikulin}, for a lattice $L$ of signature $(t^+,t^-)$ with $t^+\ge 1$, $t^-\ge 1$ and such that $l(A_L^{(p)})\le t^++t^--3$ for all $p$, then the lattice is of the form $L=U+T$, where $U$ is the hyperbolic lattice. This means that if $t^+\ge 2$, $t^-\ge 2$ and $l(A_L^{(p)})\le t^++t^--5$ then $L=U\oplus U_1\oplus T$ and one can apply the Eichler criterion.

\section{Some technical proofs.}
\subsection{No symmetries acting trivially on exactly marginal operators.}\label{a:nokernel}

In this section, we prove the statement in section \ref{s:callifsymm}, that the K3 NLSM $\C'$ with `the largest symmetry group' \cite{Gaberdiel:2013psa}, and all its symmetric products $\Sym^n \C'$, admit no symmetries acting trivially on the space of exactly marginal operators.

Let us first prove the statement above for the case $n=1$, i.e. for the model $\C'$ itself. This model is a rational CFT, with the bosonic chiral and antichiral algebras both isomorphic to $(\hat{su}(2)_1)^{\oplus 6}$. Let us denote by $[0]$ and $[1]$ the modules of $\hat{su}(2)_1$ with conformal weights $0$ and $1/4$ and by $[a_1,\ldots,a_6;b_1,\ldots,b_6]\in \ZZ_2^6\oplus\ZZ_2^6$ the representations of $(\hat{su}(2)_1)^{\oplus 6}\oplus (\hat{su}(2)_1)^{\oplus 6}$. Then, the theory contains one copy of each of the representations $[a_1,\ldots,a_6;b_1,\ldots,b_6]$ with $\sum_ia_i\equiv 0\mod 2$ and either $a_i=b_i$ for all $i$, or $a_i=b_i+1$ for all $i$. The holomorphic and antiholomorphic supercurrents are contained in the $[1,1,1,1,1,1;0,0,0,0,0,0]$ and $[0,0,0,0,0,0;1,1,1,1,1,1]$ modules, respectively. The RR ground fields are contained in the six modules where $a_i=b_i$ for all $i$ and only one of them is non-zero, that is
\be [1,0,0,0,0,0;1,0,0,0,0,0],\ [0,1,0,0,0,0;0,1,0,0,0,0],\ \ldots,\ [0,0,0,0,0,1;0,0,0,0,0,1]\ .
\ee
From the description of the spectrum, it is clear that, by taking OPEs of the RR ground fields, of the $\N=(4,4)$ supercurrents, and of the $(\hat{su}(2)_1)^{\oplus 6}\oplus (\hat{su}(2)_1)^{\oplus 6}$ currents, one can obtain every field in the theory. Thus, we just need to prove that a symmetry acting trivially on $\N=(4,4)$ and on RR ground fields acts trivially also on the Kac-Moody algebra currents. The groups of automorphisms of $(\hat{su}(2)_1)^{\oplus 6}$ is $SO(3)^6:S_6$, where $SO(3)^6=(SU(2)/(-1))^6$ is the group of inner automorphisms generated by the zero modes of the currents. Considering both the chiral and anti-chiral algebras, one gets a group $(SO(3)^6\times SO(3)^6):S_6$; notice that the spectrum is only invariant under the diagonal $S_6$ acting simultaneously on the left- and right-movers. Clearly, a symmetry exchanging the different $\hat{su}(2)_1$ subalgebras must also exchange the corresponding RR ground fields, so it cannot act trivially on them. As for symmetries acting by $SO(3)^6\times SO(3)^6$  automorphisms on the current algebra, compatibility of the OPE between currents and RR ground states require the latter to transform under the corresponding lift in the covering $SU(2)^6\times SU(2)^6$. The latter group does not act faithfully on the RR ground states, but the kernel of the action is contained in the $\ZZ_2^6\times \ZZ_2^6$ central subgroup, which is the kernel of the covering $SU(2)^6\times SU(2)^6\to SO(3)^6\times SO(3)^6$. Thus every non-trivial element in $SO(3)^6\times SO(3)^6$ lifts to an element acting non-trivially on the RR ground fields. It follows that symmetries acting trivially on the ground fields must act trivially also on the currents and we conclude. 

The extension of this proof to the case of $\Sym^n\C'$ is simple. The untwisted sector of the symmetric orbifold $\Sym^n \C'$ contains all the purely holomorphic and purely anti-holomorphic fields of the model, since all twisted sectors ground states have strictly positive $L_0$ and $\bar L_0$ eigenvalue. The whole NS-NS untwisted sector of $\Sym^n \C'$  is generated (via OPE) by the fields of the form $\chi^{[n]}:=\sum_{\sigma\in S_n} \sigma(\chi\otimes 1\otimes\cdots\otimes 1)$, where $\chi$ is a holomorphic or a antiholomorphic field in $\C'$. A symmetry of $\Sym^n \C'$ must preserve the sets of purely holomorphic and purely anti-holomorphic fields, so in particular it must preserve the whole untwisted sector. Furthermore, if a symmetry acts trivially on the exactly marginal operators in the untwisted sector of $\Sym^n\C'$, then it must act trivially also on the generators $\chi^{[n]}$; if this were not true, then an analogous symmetry would exist for $\C'$. Therefore, a symmetry acting trivially on the exactly marginal operators must act trivially on the whole untwisted sector of the theory. Finally, the whole model $\Sym^n \C'$ is generated by the untwisted sector and by the twisted ground fields in the $[\sigma]$-twisted sector, where $[\sigma]$ is the $S_n$ conjugacy class corresponding to a single transposition. But the supersymmetric descendants of these twisted ground states are also exactly marginal operators. Therefore, a symmetry preserving the $\N=(4,4)$ superconformal algebra, the spectral flow, and acting trivially on all exactly marginal operators must necessarily be trivial.

\subsection{Proof of proposition \ref{th:lemma}}\label{a:pflemma}
	
	Choose a primitive vector $v\in \Gamma^{5,21}$ of norm $v^2=2n-2$ and identify $L_n$ with $L_v=v^\perp\cap \Gamma^{5,21}$. We have to prove that the subgroup $G_\Pi$ of $Stab^+(v)\subset O(\Gamma^{5,21})$ fixing $Z=\langle v\rangle\oplus \Pi$ pointwise, is isomorphic to the group $O^0(L_\Pi)$ of automorphisms of $L_\Pi$ acting trivially on $L_\Pi^*/L_\Pi$. Recall that $L_\Pi=L_v\cap \Pi^\perp=\Gamma^{5,21}\cap Z^\perp$. Let $L^\Pi:=\Gamma^{5,21}\cap (L_\Pi)^\perp$ and consider $g\in G_\Pi$. Then, $g$ acts trivially on $L^\Pi$, because, for every $v\in L^\Pi$, $v-g(v)\in L^\Pi $ must be orthogonal to $Z$ (since $X$ is fixed by $g$), so it must be contained in $L_\Pi$; but $L^\Pi\cap L_\Pi=0$, so that $g(v)=v$. Therefore, $g$ preserves the lattice $L_\Pi$ and must be an element of $O(L_\Pi)$. Now, $L_\Pi$ and $L^\Pi$ are primitive sublattices of $\Gamma^{5,21}$, that are the orthogonal complements of each other.  As explained in appendix \ref{a:lattices} (see \cite{Nikulin}), this implies there is an isomorphism $\gamma: L_\Pi^*/L_\Pi\to (L^\Pi)^*/L^\Pi$ such that $(x,y)\in L_\Pi^*\oplus (L^\Pi)^*$ is a vector of $\Gamma^{5,21}$ if and only if the images $\bar x$, $\bar y$ in the discriminant groups are related by $\bar y=\gamma(\bar x)$. Since $g$ is an automorphism of $\Gamma^{5,21}$, it must preserve this condition and therefore the action of $g$ on $ L_\Pi^*/L_\Pi$ and $(L^\Pi)^*/L^\Pi$ must be compatible with $\gamma$. Since $g$ acts trivially on $L^\Pi$, and therefore also on $(L^\Pi)^*/L^\Pi$, then it must act trivially also on $L_\Pi^*/L_\Pi$. Therefore, $g\in O^0(L_\Pi)$, and $G_\Pi\subseteq O^0(L_\Pi)$.
	
	Vice versa, suppose that $g\in O^0(L_\Pi)$. Then, we can extend the action on $\Gamma^{5,21}$ by requiring it to act trivially on $L^\Pi$. This action is compatible with the isomorphism $\gamma$, and therefore it is an automorphism of $\Gamma^{5,21}$. Furthermore, it acts trivially on the space $L^\Pi\otimes \RR=L_\Pi^\perp\subset \Gamma^{5,21}\otimes\RR $, which contains $Z=\langle v\rangle\oplus \Pi$. It follows that $g\in G_\Pi$ and $O^0(L_\Pi)\subseteq G_\Pi$.  We conclude that $G_\Pi=O^0(L_\Pi)$.
\subsection{Proof of proposition \ref{th:main}}\label{a:pfthmain}

Choose a primitive vector $v\in \Gamma^{5,21}$ of norm $v^2=2n-2$ and identify $L_n$ with $L_v=v^\perp\cap \Gamma^{5,21}$. Let \be \Gamma^G:=\{ v\in \Gamma^{5,21}\mid g(v)=g,\forall g\in g\}\ee
\be \Gamma_G:=(\Gamma^G)^\perp\cap \Gamma^{5,21}=\{v\in \Gamma^{5,21}\mid (v,w)=0\ \forall w\in\Gamma^G\}\ .
\ee Since $\Gamma^G\otimes \RR$ contains the positive definite $5$-dimensional subspace $Z$, it must have signature $(5,d)$ for some $0\le d\le 21$, so that $\Gamma_G$ is negative definite of rank $21-d$. By \cite{Nikulin} 1.12.3, $\Gamma^G$ can be primitively embedded in the even unimodular lattice $\Gamma^{9+d,1+d}$ of signature $(9+d,1+d)$. If $S:=(\Gamma_G)^\perp \cap \Gamma^{9+d,1+d}$ is the orthogonal complement of $\Gamma^G$ in $\Gamma^{9+d,1+d}$, then $S(-1)$ has signature $(1,4+d)$ and the discriminant form $q_{S(-1)}$ satisfies
\be q_{S(-1)}\cong -q_S\cong q_{\Gamma^G}\cong -q_{\Gamma_G}\ .
\ee
Therefore $S(-1)$ and $\Gamma_G$ can be glued to get the even unimodular lattice $\Gamma^{1,25}$. This means that both $S(-1)$ and $\Gamma_G$ can be primitively embedded in $\Gamma^{1,25}$. Furthermore, since the induced action of $G$ on the the discriminant group $A_{\Gamma_G}\cong A_{\Gamma^G}$ is trivial, the action of $G$ on $\Gamma_G$ can be extended to an action on $\Gamma^{1,25}$ acting trivially on the orthogonal complement of $\Gamma_G$. We conclude that $G$ is a subgroup of the group of authomorphisms $O^+(\Gamma^{1,25})$ acting trivially on a sublattice $S(-1)\subset \Gamma^{1,25}$ of signature $(1,4+d)$.

The lattice $\Gamma^{1,25}$ can be described as
\be \Gamma^{1,25}:=\{(x_0;x_1,\ldots,x_{25})\in \ZZ^{26}\cup (\ZZ+\frac{1}{2})^{26},\ \sum_i x_i\in 2\ZZ\}\subset \RR^{1,25}
\ee i.e. the lattice of vectors $(x_0;x_1,\ldots,x_{25})\in \RR^{1,25}$ with all integral or all half-integral entries whose sum is even, with quadratic form $(x_0;x_1,\ldots,x_{25})^2=x_0^2-\sum_{i=1}^{25} x_i^2$.
The group $O^+(\Gamma^{1,25})$ is known (see \cite{ConwaySloane}, chapter 27) to be isomorphic to $W\rtimes Co_\infty$, where $W$ is the (infinite) Weyl group generated by reflections with respect to hyperplanes $r^\perp\subset \Gamma^{1,25}\otimes \RR\cong \RR^{1,25}$ orthogonal to the roots $r\in \Gamma^{1,25}$, $r^2=-2$, and $Co_\infty=Leech\rtimes Co_0$ is the affine Conway group. The union $\bigcup_{r\in\Gamma^{1,25},r^2=-2}r^\perp$ of hyperplanes orthogonal to the roots divides the space $\Gamma^{1,25}\otimes \RR\cong \RR^{1,25}$ into infinitely many connected components;  the closure of any such component is called a Weyl chamber. The Weyl group acts transitively by permutation on the set of Weyl chambers, and if $D$ is the interior of any Weyl chamber, then for any $t\in W$, $t\neq 1$, one has $t(D)\cap D=\emptyset$. We call the Weyl chamber containing the null vector $w=(70;0,1,2,3,\ldots,24)\in \RR^{1,25}$, $w^2=0$ (Weyl vector) the fundamental Weyl chamber, and the roots corresponding to the hyperplanes delimiting the fundamental chamber are called fundamental roots. The Weyl vector $w$ has the property that $(w,r)=-1$ for every fundamental root $r$. The group $Co_\infty$ can be identified with the subgroup of $O^+(\Gamma^{1,25})$ fixing the fundamental chamber (setwise), and it fixes the Weyl vector $w$. Finally, notice that the lattice $(\Gamma^{1,25}\cap w^\perp)/\langle w\rangle$ is isomorphic to the negative definite Leech lattice $\Lambda$.

Since $\Gamma_G$ contains no roots, then $S(-1)\otimes \RR =(\Gamma_G)^\perp$ cannot be contained in $r^\perp$ for any root $r$, and therefore cannot be contained in the union $\bigcup_{r\in\Gamma^{1,25},r^2=-2}r^\perp$ (given any two vectors of $S(-1)\otimes \RR$ not contained in the same hyperplane, a suitable linear combination of them is not contained in the union). Therefore, the group $G\subset O^+(\Gamma^{1,25})$ fixes a vector in the interior of a Weyl chamber. The embedding of $\Gamma_G$ in $\Gamma^{1,25}$ is defined up to $O^+(\Gamma^{1,25})$ transformations, and one can use this freedom to make sure that $G$ fixes a vector in the interior of the fundamental chamber. This implies that $G$ is a subgroup of the group $Co_\infty$ preserving this chamber -- any automorphism outside $Co_\infty$ maps the fundamental chamber to a different chamber, so it cannot fix a vector in its interior. It follows that $G$ fixes the Weyl vector $w$, so that $\Gamma_G\subset \Gamma^{1,25} \cap w^\perp$. Composing the embedding $\Gamma_G\hookrightarrow \Gamma^{1,25}\cap  w^\perp$ with the map  $\Gamma^{1,25} \cap w^\perp\to (\Gamma^{1,25} \cap w^\perp)/\langle w\rangle$, one obtains a primitive embedding of $\Gamma_G$ in the Leech lattice $\Lambda$. By the same argument as above, the action of $G$ on $\Gamma_G$ can be extended to an action on the Leech lattice $\Lambda$ acting trivially on the orthogonal complement of $\Gamma_G$. Therefore, $G$ is contained in the  subgroup of the Conway group $Co_0$ acting trivially on the orthogonal complement $\Lambda^G:=(\Gamma_G)^\perp \cap \Lambda$. In fact, $G$ coincides with the subgroup of $Co_0$ fixing pointwise the lattice $\Lambda^G$: any generator  in this group has a well defined action on $(\Gamma^G)^\perp \cap \Lambda\cong \Gamma_G$, and the action can be extended to an automorphism of $\Gamma^{5,21}$ acting trivially on the orthogonal complement $\Gamma^G$, and therefore fixing $Z$ pointwise.

For the vice versa, let $\tilde G\subset Co_0$ be the pointwise stabilizer of a sublattice $\Lambda^{\tilde G}\subset \Lambda$ of rank $3+d$, $d\ge 0$, and let $\Lambda_{\tilde G}:=(\Lambda^{\tilde G})^\perp\cap \Lambda$ be its orthogonal complement in $\Lambda$. If $d>0$, then the results in \cite{K3symm} imply that there exists a well defined non-linear sigma model $\C$ on K3 with symmetry group $\tilde G$. Therefore, the $n$-th symmetric product $\Sym^n \C$ of this non-linear sigma model has a group of symmetries containing $\tilde G$. In fact, the symmetry group of $\Sym^n\C$ is strictly larger, because of the $\ZZ_2$ symmetry acting by $(-1)^{\sgn(\sigma)}$ on the $\sigma$-twisted sector. On the other hand, the group $\tilde G$ of symmetries is preserved by deformations induced by exactly marginal operators in the twisted sector, so that we have a continuous family of models whose symmetry group contains $\tilde G$. A generic element in this family is non-singular and its symmetry group is exactly $\tilde G$.

Let us now consider the case where $d=0$. Consider a lattice $\Lambda^G(-1)\oplus_\perp \langle v \rangle$ of signature $(4,0)$ given by the orthogonal sum of $\Lambda^G(-1)$ and a $1$-dimensional lattice $ \langle v \rangle$ generated by a vector $v$ of square norm $2n-2$. By theorem 1.12.3 of \cite{Nikulin}, this lattice can be primitively embedded in the $E_8$ lattice. Let $S:=(\Lambda^G(-1))^\perp \cap E_8$ and $S_0:=(\Lambda^G(-1)\oplus \langle v\rangle )^\perp \cap E_8$ be the orthogonal complements of $\Lambda^G(-1)$ and $\Lambda^G(-1)\oplus_\perp \langle v\rangle$ in $E_8$. Then, $S$ has rank $5$ and $q_S\cong -q_{\Lambda^G(-1)}=q_{\Lambda^G}\cong -q_{\Lambda_G}$, so that $S$ and $\Lambda_G$ can be glued to obtain the even unimodular lattice $\Gamma^{5,21}$. By taking $Z\equiv S\otimes \RR\subset \RR^{5,21}$, we obtain that the stabilizer of $Z$ in $O^+(\Gamma^{5,21})$ is exactly the group $\tilde G$. Furthermore, $Z$ contains the vector $v$ of length $2n-2$, so that it is an attractor point for the stringy-like object of the corresponding charge. It remains to check that the point in the moduli space of NLSM on $K3^{[n]}$ determined by $Z$ is non-singular. By construction, $Z^\perp \cap \Gamma^{5,21}=\Lambda_G$, which contains no roots, since it is a sublattice of the Leech lattice. We need to prove that there are no vectors $v',v''\in \Gamma^{5,21}$ such that $v'+v''= v$ and $v'_+=\alpha v$ for some $0<\alpha<1$. Let us suppose by absurdity that such a vector exists. Since $\Gamma^{5,21}\subset S^*\oplus \Lambda_G^*$, and since $v'_+$ is the component of $v'$ along $S^*\otimes \RR\subset \Gamma^{5,21}\otimes \RR$, we obtain that $v'_+\in S^*$. Furthermore, since $v'_+$ is proportional to $v $, it must be orthogonal to $S_0\subset S^*\cap v^\perp$. Recall that $S$ and $\Lambda^G$ can be `glued' to get the $E_8$ lattice, so that for every vectors in $S^*$, and in particular $v'_+$, there is a vector $\lambda\in(\Lambda^G)^*$ such that $(v'_+,\lambda)\in E_8\subset S^*\oplus (\Lambda^G)^*$. Since both $v'_+$ and $\Lambda^G$ are orthogonal to $S_0\subset S\subset E_8$, we have $(v'_+,\lambda)\in E_8\cap (S_0)^\perp$. But $E_8\cap (S_0)^\perp=\Lambda^G\oplus \langle v\rangle$, so $v'_+$ must be an integral multiple of $v$, contradicting the expectation that $0<\alpha<1$. We conclude that the $K3^{[n]}$ model specified by $Z=S\otimes \RR \in \Gamma^{5,21}\otimes \RR$ is non-singular.

\section{Discriminant forms}\label{a:discriminants}

In this section, we describe the main facts about discriminant forms and a dictionary between the notations used by different authors. Furthermore, we describe the computations that lead to the results reported in table \ref{tab:big}.

A discriminant form is a finite abelian group $A_q$ (we use additive notation for this group) with a quadratic form $q:A\to \QQ/2\ZZ$. If $L$ is an even lattice, then $L^*/L$ with the induced quadratic form is a discriminant form, see appendix \ref{a:lattices}. (Some authors, including Miranda and Morrison \cite{MirandaMorrison1,MirandaMorrison2}, define $q$ as $1/2$ of this, so that it is a map into $\QQ/\ZZ$ instead of $\QQ/2\ZZ$. I use convention of \cite{Nikulin} here).
The rank of a discriminant form is the minimal number of generators of the group.


Two discriminant forms $A_q$ and $A'_{q'}$ are equivalent if there is an isomorphism of abelian groups $\gamma:A_q\stackrel{\cong}{\to} A'_{q'}$  that preserves the quadratic forms $q'\circ \gamma=q$.

One can prove that any discriminant form decomposes into orthogonal components
\be A=\oplus_{p\text{ prime}} A^{(p)}\ ,
\ee where each component $A^{(p)}$ has order a prime power $p^k$. So, it is sufficient to study discriminant forms whose order is a prime power. Furthermore, if $p$ is odd, $A^{(p)}$ can always be diagonalized, i.e. written (not uniquely, in general) as an orthogonal sum of Jordan components of rank $1$. For $p=2$, any discriminant form $A^{(2)}$ can be decomposed into Jordan components of rank at most $2$.

\medskip
Let us describe the elementary Jordan components and their notation in Miranda-Morrison's (MM)  \cite{MirandaMorrison1,MirandaMorrison2} and Nikulin's (N) article \cite{Nikulin}.

For $p$ odd, the inequivalent discriminant forms of rank $1$ are denoted by $w^\epsilon_{p,k}$, $\epsilon=\pm 1$ (MM) or $q^{(p)}_\theta(p^k)$, $\theta\in(\ZZ/{p^k}\ZZ)^\times$ (considered mod $((\ZZ/{p^k}\ZZ)^\times )^2$) and  are given by a cyclic group of order $p^k$, $A\cong \ZZ/p^k\ZZ$, with generator $x$ such that $q(x)=\frac{\theta}{p^k}$, where  $\left(\frac{\theta}{p}\right)=\epsilon$ (Jacobi symbol). Our notation here is as follows: $\theta$ is an element in $(\ZZ/{p^k}\ZZ)^\times$, the multiplicative group of integers modulo $p^k$ that are coprime to $p$.  $((\ZZ/{p^k}\ZZ)^\times )^2$ is the subgroup of elements in $(\ZZ/{p^k}\ZZ)^\times$ that are squares mod $p^k$ (i.e., squares modulo $p$). Two discriminant forms $q^{(p)}_\theta(p^k)$ and $q^{(p)}_{\theta'}(p^k)$ are equivalent iff $\theta'= \alpha^2\theta\mod p^k$, for some $\alpha^2\in ((\ZZ/{p^k}\ZZ)^\times )^2$, so that the form depends on $\theta$ only modulo $((\ZZ/{p^k}\ZZ)^\times )^2$. In fact, for $p$ odd, there are only two cosets in $(\ZZ/{p^k}\ZZ)^\times/((\ZZ/{p^k}\ZZ)^\times )^2$, namely the squares mod $p^k$ (i.e., the trivial coset) and the non-squares. The Jacobi symbol $\left(\frac{\theta}{p}\right)$, for $\theta$ coprime to $p$, equals $+1$ is $\theta$ is a square mod $p$ and $-1$ if it is not. To summarize,
\be
q^{(p)}_{\theta}(p^k)\cong \frac{\theta}{p^k}\cong \begin{cases}w^{+1}_{p,k} & \text{if }\theta\text{ is a square}\mod p\ ,\\
	w^{-1}_{p,k} & \text{if }\theta\text{ is not a square}\mod p\ .\end{cases}
\ee

\medskip

For $p=2$, the irreducible Jordan components can have rank $1$ or $2$. The inequivalent components of rank $1$ are denoted by $w^\epsilon_{2,k}$, $\epsilon\in \ZZ_8^\times $ (MM) or $q^{(2)}_\theta(2^k)$, $\theta\in (\ZZ/{2^k}\ZZ)^\times$ (considered modulo $((\ZZ/{2^k}\ZZ)^\times)^2$) and  are given by a cyclic group of order $2^k$, $A\cong \ZZ/2^k\ZZ$, with a generator $x$ such that $q(x)=\frac{\theta}{2^k}\mod 2\ZZ$, where $\theta\equiv \epsilon\mod 8$. For $k=1$, $w^\epsilon_{2,1}$ really depends on $\epsilon \mod 4$ rather than $\mod 8$ (see the relations below). To summarize,
\be w^\epsilon_{2,k}=q_2^\theta(2^k)=\frac{\theta}{2^k}\qquad \theta=\epsilon\mod 8\ .
\ee
The different range of $\epsilon$ with respect to the odd case depends on the fact that the quotient $(\ZZ/{2^k}\ZZ)^\times/((\ZZ/{2^k}\ZZ)^\times )^2$ has four cosets rather than $2$.

The inequivalent $2$-components of rank $2$ are given by: 
\begin{itemize}
	\item $u_k$ (MM) or $u^{(2)}_+(2^k)$ (N) is given by a group $(\ZZ/2^k\ZZ)\times (\ZZ/2^k\ZZ)$ with generators $x,y$ and quadratic form $q(x)=q(y)=0\mod 2\ZZ$ and $q(x+y)=2^{-k+1}\mod 2\ZZ$, i.e.
	\be u_k=u^{(2)}_+(2^k)=\begin{pmatrix}
		0 & 2^{-k}\\ 2^{-k} & 0
	\end{pmatrix}
	\ee
	\item $v_k$ (MM) or $v^{(2)}_+(2^k)$ (N) is given by a group $(\ZZ/2^k\ZZ)\times (\ZZ/2^k\ZZ)$ with generators $x,y$ and quadratic form $q(x)=q(y)=q(x+y)=2^{-k+1}\mod 2\ZZ$, i.e.
	\be v_k=v^{(2)}_+(2^k)=\begin{pmatrix}
		2^{-k+1} & 2^{-k}\\ 2^{-k} & 2^{-k+1}
	\end{pmatrix}
	\ee
\end{itemize}

In general, there are many equivalent ways to decompose a discriminant form into  elementary Jordan components. All these equivalences are consequence of the following relations (see \cite{Nikulin}):
\be\label{rel1} (w_{p,k}^{\epsilon})^2\cong (w_{p,k}^{\epsilon'})^2\qquad p\text{ odd}\ ,
\ee
\be\label{rel2} w_{2,1}^{\epsilon}=w_{2,1}^{5\epsilon}
\ee
\be\label{rel3} (u_k)^2=(v_k)^2\ ,
\ee
\be\label{rel4} (w_{2,k}^{\epsilon})^2\cong (w_{2,k}^{5\epsilon})^2\ ,
\ee
\be\label{rel5} (w_{2,k}^{\epsilon})^2\oplus (w_{2,k}^{\epsilon'})\cong \begin{cases} v_k\oplus (w_{2,k}^{-5\epsilon'})\ , & \text{if }\epsilon\equiv \epsilon'\mod 4\ ,\\
	u_k\oplus (w_{2,k}^{-\epsilon'})\ , & \text{if }\epsilon\equiv -\epsilon'\mod 4\ ,\end{cases}
\ee
\be\label{rel6} v_k\oplus w_{2,k+1}^\epsilon=u_k\oplus w_{2,k+1}^{5\epsilon}\ ,
\ee
\be\label{rel7} w_{2,k}^\epsilon \oplus v_{k+1} =w_{2,k}^{5\epsilon}\oplus u_{k+1}\ ,
\ee
\be\label{rel8} w_{2,k}^\epsilon\oplus w_{2,k+1}^{\epsilon'}=w_{2,k}^{\epsilon+2\epsilon'}\oplus w_{2,k+1}^{5(\epsilon'-2\epsilon)}\ ,
\ee
\be\label{rel9} w_{2,k}^\epsilon\oplus w_{2,k+2}^{\epsilon'}=w_{2,k}^{5\epsilon}\oplus w_{2,k+2}^{5\epsilon'}\ .
\ee
 In \cite{MirandaMorrison1,MirandaMorrison2}, Miranda and Morrison use these relations to decompose every 2-group discriminant form as
\be A^{(2)}= \oplus_{k\ge 1} (u_k^{n(k)}\oplus v_k^{m(k)}\oplus w(k))\ ,
\ee where $m(k)\le 1$, rank$(w(k))\le 2$ and $w(k)$ is a sum of components $w_{2,k}^\epsilon$.

\bigskip

In the following, we will be interested in calculating the opposite of a discriminant form. This can be found using the rules
\be -w_{p,k}^\epsilon =\begin{cases} w_{p,k}^{\epsilon}\ , & \text{if }p\equiv 1\mod 4\ ,\\ w_{p,k}^{-\epsilon} & \text{if }p\equiv -1\mod 4\ ,
\end{cases}
\ee
\be -u_k=u_k\ ,
\ee
\be -v_k=v_k\ ,
\ee
\be -w_{2,k}^\epsilon=w_{2,k}^{-\epsilon}\ .
\ee

\bigskip

The genus of a lattice is determined by its signature and its discriminant form  (see \cite{Nikulin}, section 1.9). 
A compact symbol for the genus was given in Conway and Sloane's book \cite{ConwaySloane}. From that symbol, one can derive the corresponding discriminant form using the following dictionary:
\be (p^k)^{\pm n}=\oplus_{i=1}^n w_{p,k}^{\epsilon_i}\ ,\quad  \prod_i \epsilon_i=\pm 1\qquad \qquad p\text{ odd}\ ,
\ee
\be (2^k)^{+2n}_{II}=(u_k)^n\ ,
\ee 
\be (2^k)^{-2n}_{II}=(v_k)^n\ ,
\ee 
\be (2^k)^{\pm n}_t=\oplus_{i=1}^n w_{2,k}^{\epsilon_i}\ ,\qquad \sum_i\epsilon_i\equiv t\mod 8\ ,\qquad \prod_i\left(\frac{\epsilon_i}{2}\right)=\pm 1
\ee where 
\be \left(\frac{\epsilon}{2}\right)=\begin{cases}+1 & \text{if } \epsilon=\pm 1\mod 8\ ,\\ -1 & \text{if } \epsilon=\pm 3\mod 8\ .\end{cases}
\ee In \cite{HohnMason}, H\"ohn and Mason provide the genus of the sublattice $Leech^g$ of the positive definite\footnote{We recall that $\Lambda$ denotes the \emph{negative definite} Leech lattice, so that there is a canonical isomorphism $\Lambda\cong Leech(-1)$ given by flipping the sign of the quadratic form. This also induces an isomorphism $O(\Lambda)\cong O(Leech)$.} Leech lattice $Leech$ fixed by some $g\in Co_0$. The genus is given using the symbol in\cite{ConwaySloane}, so that one can reconstruct the corresponding discriminant form using the dictionary above.

\bigskip

	In the following table, we provide for each Frame shape $\pi_g$, the genus of the  fixed sublattice $Leech^g$ in the (positive definite) Leech lattice (as reported in \cite{HohnMason}), the number of the corresponding group in the list in \cite{HohnMason}, and the \emph{opposite} $-q_{Leech^g}$ of the discriminant form of $Leech^g$. This will be useful for computing the data in table \ref{tab:big}.
	  Notice that
	$$ -q_{Leech^g}\cong q_{\Lambda^g}=-q_{\Lambda_g}=-q_{\Gamma_g}\ .
	$$  where  $q_{\Gamma_g}$ of the coinvariant sublattice $\Gamma_g:=(\Gamma^g)^\perp\cap \Lambda$. The discriminant form is given in Miranda and Morrison (MM) notation \cite{MirandaMorrison1,MirandaMorrison2}, so that it is easier to apply their theorems.
	\begin{tabularx}{\linewidth}{CCCCC}
	\pi_g & \text{HM group \#}  &	\text{genus }Leech^g &-q_{Leech^g}\\
	\midrule\endhead
 1^{24} &  1 & 1
\\
	\rowcolor{gray!11}{}
	1^82^8 & 2 & 2_{II}^{+8} &  u_1^{\oplus 4} \\
	1^{-8}2^{16}	& 14		&	2_{II}^{+8}		& u_1^{\oplus 4}\\
	\rowcolor{gray!11}{}
	2^{12} 				& 5			&	2_4^{+12}	& u_1^{\oplus 4}\oplus v_1\oplus w_{2,1}^1\oplus w_{2,1}^{-1}\\
	1^6 3^6			   & 4 			& 3^{+6} 				& (w_{3,1}^{-1})^{\oplus 6}\\
	\rowcolor{gray!11}{}
	1^{-3}3^9 		  & 35		& 3^{+5}			& (w_{3,1}^{-1})^{\oplus 5}\\
	3^8					& 22		& 3^{+8}			& (w_{3,1}^{-1})^{\oplus 8}\\
	\rowcolor{gray!11}{}
	1^42^24^4		& 9			& 2_2^{+2}4_{II}^{+4}					& w_{2,1}^1\oplus w_{2,1}^{-1}\oplus u_2^{\oplus 2}\\
	1^82^{-8}4^8	& 14		& 2_{II}^{+8}		& u_1^{\oplus 4}\\
	\rowcolor{gray!11}{}
	1^{-4}2^64^4	& 41		& 2_6^{+2}4_{II}^{+4}		&(w_{2,1}^1)^{\oplus 2}\oplus (u_2)^{\oplus 2}\\
	2^{-4}4^8			& 99		& 2_{II}^{-2}4_{II}^{-2} 	& v_1\oplus v_2\\
	\rowcolor{gray!11}{}
	2^44^4				& 21		& 2_{II}^{+4}4_{II}^{+4}			& u_1^{\oplus 2}\oplus u_2^{\oplus 2}\\
	4^6						& 64		& 4_6^{+6}					& u_2^{\oplus 2}\oplus (w_{2,2}^5)^{\oplus 2}\\
	\rowcolor{gray!11}{}
	1^45^4				& 20				& 5^{+4}					& (w_{5,1}^1)^{\oplus 4}\\
	1^{-1}5^5			& 122			& 5^{+3}		& (w_{5,1}^1)^{\oplus 3}\\
	\rowcolor{gray!11}{}
	1^22^23^26^2	& 18			& 2_{II}^{+4}3^{+4}				& [u_1^{\oplus 2}]\oplus[(w_{3,1}^1)^{\oplus 4}]\\
	1^42^13^{-4}6^5 	& 33		& 2_{II}^{-6}3^{-1}		& [u_1^{\oplus 2}\oplus v_1]\oplus [w_{3,1}^1]\\
	\rowcolor{gray!11}{}
	1^52^{-4}3^16^4	& 35			& 3^{+5}		& (w_{3,1}^{-1})^{\oplus 5}\\
	1^{-2}2^43^{-2}6^4	& 104	& 2_{II}^{+4}3^{+2}	& [u_1^{\oplus 2}]\oplus[(w_{3,1}^1)^{\oplus 2}]\\
	\rowcolor{gray!11}{}
	1^{-1}2^{-1}3^36^3	& 114	& 2_{II}^{-2}3^{+4}		& [v_1]\oplus [(w_{3,1}^1)^{\oplus 4}]\\
	1^{-4}2^53^46^1		& 62	& 2_{II}^{-6}3^{-5}		& [u_1^{\oplus 2}\oplus v_1]\oplus [(w_{3,1}^1)^{\oplus 5}]\\
	\rowcolor{gray!11}{}
	2^36^3				& 63		& 2_4^{-6} 3^{-3}		& [u_1^{\oplus 2}\oplus w_{2,1}^{-1}\oplus w_{2,1}^1]\oplus [(w_{3,1}^1)^{\oplus 3}]\\
	6^4						& 161		& 2_4^{+4}3^{+4}		& [v_1\oplus w_{2,1}^{-1}\oplus w_{2,1}^1]\oplus [(w_{3,1}^1)^{\oplus 4}]\\
	\rowcolor{gray!11}{}
	1^37^3					& 52					& 7^{+3}					& (w_{7,1}^{-1})^{\oplus 3}\\
	1^22^14^18^2 		& 55				& 2_5^{+1}4_1^{+1}8_{II}^{+2}					& v_3\oplus w_{2,1}^{-1}\oplus w_{2,2}^{-1}**\\
	 \rowcolor{gray!11}{}
	 1^42^{-2}4^{-2}8^4		& 99	& 2_{II}^{-2}4_{II}^{-2}	& v_1\oplus v_2\\
	 1^{-2}2^34^18^2	& 143		& 2_3^{+1}4_1^{+1}8_{II}^{+2}	& v_3\oplus w_{2,2}^1\oplus w_{2,1}^3**\\
	 \rowcolor{gray!11}{}
	 2^44^{-4}8^4		& 107			& 4_4^{+4}			& v_2\oplus w_{2,2}^5\oplus w_{2,2}^{-1}\\
	 4^28^2				& 147				& 4_{II}^{-2}8_{II}^{-2}	& v_2\oplus v_3\\
	 \rowcolor{gray!11}{}
	 1^3 3^{-2}9^3	&	101			& 3^{+2}9^{+1}			& (w_{3,1}^1)^{\oplus 2}\oplus w_{3,2}^{-1}\\
	 1^22^15^{-2}10^3	& 100	& 2_{II}^{-4}5^{-1}		& [v_1^{\oplus 2}]\oplus [w_{5,1}^{-1}]\\
	 \rowcolor{gray!11}{}
	 1^32^{-2}5^110^2	& 122		& 5^{+3}			& (w_{5,1}^1)^{\oplus 3}\\
	 1^{-2}2^35^210^1	& 159		& 2_{II}^{-4}5^{-3}		& [v_1^{\oplus 2}]\oplus[(w_{5,1}^{-1})^{\oplus 3}]\\
	  \rowcolor{gray!11}{}
	  2^210^2			& 149			& 2_4^{+4}5^{+2}	& [v_1\oplus w_{2,1}^1\oplus w_{2,1}^{-1}]\oplus [(w_{5,1}^1)^{\oplus 2}]\\
	  1^211^2			& 120 			& 11^{+2}			& (w_{11,1}^1)^{\oplus 2}\\
	  \rowcolor{gray!11}{}
	  1^22^{-2}3^24^26^{-2}12^2		& 104	& 2_{II}^{+4}3^{+2}		& [u_1^{\oplus 2}]\oplus [(w_{3,1}^1)^{\oplus 2}]\\
	  1^12^23^14^{-2}12^2	& 109			& 2_2^{+2}3^{+3}		& [(w_{2,1}^{1})^{\oplus 2}]\oplus [(w_{3,1}^{-1})^{\oplus 3}]\\
	   \rowcolor{gray!11}{}
	   1^23^{-2}4^16^212^1	& 123		& 2_2^{-2}4_{II}^{+2}3^{+1}	& [u_2\oplus w_{2,1}^1\oplus w_{2,1}^5]\oplus[w_{3,1}^{-1}]\\
	   1^{-2}2^23^24^112^1	& 157		& 2_6^{-2}4_{II}^{+2}3^{+3}	& [u_2\oplus w_{2,1}^{-1}\oplus w_{2,1}^3]\oplus[(w_{3,1}^{-1})^{\oplus 3}]\\
	   \rowcolor{gray!11}{}
	   2^14^16^112^1	& 135			& 2_{II}^{-2}4_{II}^{-2}3^{-2}		& [v_1\oplus v_2]\oplus[w_{3,1}^1
	   \oplus w_{3,1}^{-1}]\\
	   1^12^17^114^1	& 129		& 2_{II}^{+2}7^{+2}		& [u_1]\oplus [(w_{7,1}^{1})^{\oplus 2}]\\
	   \rowcolor{gray!11}{}
	   1^13^15^115^1	& 128		& 3^{-2}5^{-2}			& [w_{3,1}^{1}\oplus w_{3,1}^{-1}]\oplus [w_{5,1}^{1}\oplus w_{5,1}^{-1}]\\
	\caption{\small  For each Frame shape $\pi_g$, we provide: the number of the smallest group containing $g$, as listed by H\"ohn and Mason in \cite{HohnMason}; the genus of the fixed sublattice $Leech^g$, using the conventions of \cite{HohnMason}; and the opposite $-q_{Leech^g}$ of the discriminant form of $Leech^g$ in Miranda and Morrison's notation \cite{MirandaMorrison1,MirandaMorrison2}.}\label{tab:discr}\end{tabularx}

The discriminant form of $\Lambda_{g,n}=\langle v\rangle \oplus \Gamma_g$ is the orthogonal sum $q_v\oplus q_{\Gamma_g}$ of the discriminant form of $\Gamma_g$, and the discriminant form $q_v$ of $\langle v\rangle$. The latter can be obtained as follows. Since $v^2=2n-2$, the lattice dual to $\langle v\rangle$ has generator $\frac{v}{2n-2}$, of norm $\frac{1}{2n-2}$. Therefore, the corresponding discriminant group is isomorphic to $\ZZ_{2n-2}$, with a generator of norm $\frac{1}{2n-2}\mod 2\ZZ$. Let $2n-2=\prod_p p^{r_p}$ the prime power decomposition of $2n-2$. Then, 
\be\label{quvi} q_v=\oplus_p w^{\epsilon_p}_{p,r_p}\ ,
\ee where the sum runs over the primes dividing $2n-2$, and
\be \epsilon_p=\left( \frac{\prod_{p'\neq p} p'^{r_{p'}}}{p}\right)\ ,\qquad p\text{ odd}\ ,
\ee and
\be \epsilon_2=\prod_{p'\neq 2} p'^{r_{p'}}\mod 8\ .
\ee 
The number of embeddings of $\Lambda_{g,n}=\langle v\rangle\oplus \Gamma_g$ in $\Gamma^{5,21}$ modulo $O^+(\Gamma^{5,21})$ can be computed using the algorithm described in \cite{MirandaMorrison2}, once the discriminant form $q_N$ of $N:=(\Lambda_{g,n})^\perp\oplus \Gamma^{5,21}$ is known. This is given by
$$ q_N=-q_{\Lambda_{g,n}}=(-q_v)\oplus (-q_{\Gamma_g})= (-q_v)\oplus (-q_{Leech^g})\ .
$$ For each Frame shape $\pi_g$, the discriminant form $-q_{Leech^g}$ is reported in the table above, while $-q_v$ can be easily determined for each $m$ from \eqref{quvi}. Given these data and the algorithm in \cite{MirandaMorrison2}, computing the number of cosets in the third column of table \ref{tab:big} is a tedious but straightforward exercise.

\section{Modular forms from Borcherds products}

In this appendix, we show that the infinite product \eqref{PhiSieg} defines a meromorphic Siegel modular form of genus two with respect to a congruence subgroup of $Sp(4,\ZZ)$. The proof is a straightforward application of \cite{Borcherds98}, along the lines  \cite{GritClery}. 

Given a non-linear sigma model on $K3$ with symmetry group $G$, for any two \emph{commuting} elements $g,h\in G$ ($gh=hg$) one can define the \emph{twisted-twining genus}
\be\label{twtwdef} \phi_{g,h}(\tau,z):=\Tr_{RR,g-twisted}(hq^{L_0-\frac{c}{24}}\bar q^{\bar L_0-\frac{\bar c}{24}}y^{J_0^3}(-1)^{F+\bar F})\ .
\ee This is just a generalization of the twining genera defined in eq.\eqref{twindef}, where the trace is taken in the $g$-twisted RR sector of the $K3$ model, rather than in the untwisted sector; eq.\eqref{twindef} is recovered when $g$ is the identity. Standard path integral arguments show that each twisted-twining genus $\phi_{g,h}$ is a weak Jacobi forms of weight $0$ and index $1$ for some congruence subgroup $\Gamma_{g,h}\subseteq SL(2,\ZZ)$ depending on $g$ and $h$. More generally, under   an $SL(2,\ZZ)$ transformation\footnote{In our conventions, the $g$-twisted sector is defined so that the symmetry $g$ acts by $e^{2\pi i(\bar L_0-L_0)}$. The convention where $g\equiv e^{2\pi i(L_0-\bar L_0)}$ in the $g$-twisted sector is also quite common; the two conventions are related by exchanging the $g$ and $g^{-1}$-twisted sector. The detailed form of the modular transformations \eqref{twtwmod} depends on this choice.} 
	\be\label{twtwmod} \phi_{g,h}\left(\frac{a\tau+b}{c\tau+d},\frac{z}{c\tau+d}\right)=e^{\frac{2\pi i z^2}{c\tau+d}}\phi_{g^ah^c,g^bh^d}(\tau,z)\ ,\qquad \begin{pmatrix}
		a & b\\ c & d
	\end{pmatrix}\in SL(2,\ZZ)\ ,
	\ee i.e. they are the components of a vector valued weak Jacobi form for $SL(2,\ZZ)$. For simplicity, we assumed that the multiplier is trivial; our analysis can be easily extended to the general case, as explained at the end of this section.
	
 Let us focus on twisted twining genera of the form $\phi_{g^s,g^k}$, where both the `twist' and the `twining'  are powers of the same element of $g\in G$, and therefore generate a cyclic group $\ZZ_N$. We can define the discrete Fourier transforms
 \be\label{FourTr} \hat \phi_{s,t}(\tau,z):=\frac{1}{N}\sum_{k=1}^N e^{-2\pi i \frac{kt}{N}} \phi_{g^s,g^k}(\tau,z)=\sum_{n,l} \hat c_{s,t}(n,l) q^ny^l\ ,\qquad s,t\in \ZZ/N\ZZ\ ,
 \ee where now $n\in \frac{1}{N}\ZZ$, $n\ge 0$, and $\ell\in \ZZ$. The Fourier coefficients $\hat c_{s,t}(n,l)$ of $\hat \phi_{s,t}$ are the $\ZZ_2$-graded dimensions of the $g$-eigenspace with eigenvalue $e^{\frac{2\pi it}{N}}$ in the $g^s$-twisted sector, and in particular they are all integral. 
 
 The elliptic property of weak Jacobi forms of index $1$ implies that, for each $s$ and $t$, $\hat c_{s,t}(n,l)$ depends only on the discriminant $4n-l^2$ and on $l \mod 2$
 \be \hat c_{s,t}(n,l)=\hat c_{s,t,\ell}(4n-l^2)\ ,\qquad s,t\in \ZZ/N\ZZ,\ \ell\in \ZZ/2\ZZ\ ,
 \ee where $\ell\equiv l\mod 2$ \cite{eichler_zagier}. As a consequence, every weak Jacobi form of weight $1$ admits a  decomposition \cite{eichler_zagier}
 \be\label{thetadec} \hat \phi_{s,t}(\tau,z)=\sum_{\ell\in\ZZ/2\ZZ} F_{s,t,\ell}(\tau)\theta_{2,\ell}(\tau,z)\ ,
 \ee where all the $z$-dependence is encoded in the level two theta functions
 \be \theta_{2,\ell}(\tau,z)=\sum_{r\in \ZZ+\frac{\ell}{2}} q^{r^2}y^{2r}\ ,\qquad \ell\in \ZZ/2\ZZ
 \ee  and
 \be F_{s,t,\ell}(\tau)=\sum_{n} \hat c_{s,t,\ell}(n)q^n\ ,\qquad s,t\in \ZZ/N\ZZ,\ \ell\in \ZZ/2\ZZ\ ,
 \ee are holomorphic functions on the upper half-plane. The theta functions $ \theta_{2,\ell}(\tau,z)$ transform as
 \be\label{thetamod}  \theta_{2,\ell}(\tau+1,z)=e^{\frac{\pi i\ell^2}{2}} \theta_{2,\ell}(\tau,z)\ ,\qquad \theta_{2,\ell}(-\frac{1}{\tau},\frac{z}{\tau})=\sqrt{\frac{\tau}{2i}}e^{2\pi i \frac{z^2}{\tau}}\sum_{\ell'\in \ZZ/2\ZZ} (-1)^{\ell\ell'} \theta_{2,\ell'}(\tau,z)\ .
 \ee 
  Eqs.\eqref{twtwmod}, \eqref{FourTr}, \eqref{thetadec}, and \eqref{thetamod} imply that the functions $F_{s,t,\ell}$  transform as \begin{align*} F_{s,t,\ell}(\tau+1)&=e^{2\pi i \frac{st}{N}} e^{-\frac{\pi i\ell^2}{2}} F_{s,t,\ell}(\tau)\ ,\\ 
 F_{s,t,\ell}(-\frac{1}{\tau})&=\frac{\sqrt{2i}}{2N\sqrt{\tau}}\sum_{s',t'\in \ZZ/N\ZZ}\sum_{\ell'\in \ZZ/2\ZZ}e^{-2\pi i \frac{st'+s't}{N}} e^{\frac{2\pi i\ell\ell'}{2}} F_{s',t',\ell'}(\tau)\ .
 \end{align*} This mean that the $F_{s,t,\ell}$, $s,t\in \ZZ/N\ZZ$, $\ell\in \ZZ/2\ZZ$, are the components of a vector valued weakly holomorphic (i.e. holomorphic in the interior of the upper half-space and possibly with poles at the cusps) modular form $F$ of weight $-1/2$ for the metaplectic group $Mp(2,\ZZ)$. The latter is a double cover of $SL(2,\ZZ)$ whose elements can be written as
 \be \left(\begin{pmatrix}
 	a & b\\ c & d
 \end{pmatrix},\pm \sqrt{c\tau+d}\right)\ ,\qquad \begin{pmatrix}
 	a & b\\ c & d
 \end{pmatrix}\in SL(2,\ZZ)\ .
 \ee Let us analyse this representation of $Mp(2,\ZZ)$ more in detail. Given an even lattice $M$ of signature $(b^+,b^-)$, one can define the Weil representation $\rho_M$ of $Mp(2,\ZZ)$ on the group ring $\CC[M^*/M]$, with generators $e_\delta$, $\delta\in M^*/M$ (see for example section $4$ of \cite{Borcherds98}). The representation $\rho_M$ is defined by
 $$ \rho_M(T)e_\delta=e^{\pi i q_M(\delta)}e_\delta\ ,\qquad \rho_M(S)e_\delta=\frac{\sqrt{i}^{b^--b^+}}{\sqrt{|M^*/M|}}\sum_{\delta'\in M^*/M} e^{-2\pi i (\delta,\delta')}e_{\delta'}
 $$ where $q_M$ is the discriminant form on $M^*/M$, $(\cdot,\cdot)$ the corresponding $\QQ/\ZZ$-valued bilinear form, and $$ T=\left(\begin{pmatrix}
 1 & 1\\ 0 & 1
 \end{pmatrix},1\right)\ ,\qquad  S=\left(\begin{pmatrix}
 0 & -1\\ 1 & 0
 \end{pmatrix},\sqrt{\tau}\right)
 $$ are the generators of $Mp_2(\ZZ)$. In particular, let us consider the even lattice $M:= \Gamma^{1,1}\oplus \Gamma^{1,1}(N)\oplus \langle -2\rangle$ of signature $(b^+,b^-)=(2,3)$, so that $M^*\cong \Gamma^{1,1}\oplus \Gamma^{1,1}(\frac{1}{N})\oplus \langle -\frac{1}{2}\rangle$. The group $M^*/M\cong \ZZ/N\ZZ\times \ZZ/N\ZZ\times \ZZ/2\ZZ$ of order $|M^*/M|=2N^2$ has generators $\alpha,\beta,\gamma$ with discriminant form
 $$ q_M(s\alpha+t\beta+\ell\gamma)=\frac{2st}{N}-\frac{\ell^2}{2}\mod 2\ZZ\ ,
 $$
  for $s,t\in \ZZ/N\ZZ$ and $\ell\in \ZZ/2\ZZ$.  Then, $F$ can be seen as a $\CC[M^*/M]$-valued function $F\equiv(F_\delta)_{\delta\in M^*/M}$, with components
 $$ F_{s\alpha+t\beta+\ell\gamma}(\tau):=F_{s,t,\ell}(\tau)\ ,\qquad s,t\in \ZZ/N\ZZ,\ \ell\in \ZZ/2\ZZ\ ,
 $$  transforming as a modular form of weight $(b^+-b^-)/2=-1/2$ with respect to the Weil representation $\rho_M$, i.e.
 $$ F(\frac{a\tau+b}{c\tau+d})=\frac{1}{\sqrt{c\tau+d}}\,\rho_M\left(\begin{pmatrix}
 a & b\\ c & d
 \end{pmatrix},\sqrt{c\tau+d}\right)  F(\tau)\ .
 $$

 Theorem 13.3 of \cite{Borcherds98} then associates with the modular form $F$ an automorphic form $\Psi_M(F)$ under a certain group $Aut(M,F)\subseteq O(M)$ of automorphisms of $M$ (see \cite{Borcherds98} for details). In particular, for a lattice $M$ of signature $(b^+,b^-)=(2,3)$, the domain of definition of $\Psi_M(F)$ is the Siegel upper half-space of genus $2$, and $Aut(M,F)\subset O(2,3,\RR)$ can be seen as a discrete subgroup of $ Sp(4,\RR)\cong O(2,3,\RR)$, whose intersection with $Sp(4,\ZZ)$ is a suitable congruence subgroup (see \cite{GritClery} for more details). The automorphic form $\Psi_M(F)$, therefore, is a meromorphic Siegel modular form of genus two, whose zeroes and poles are described in point 2 of Theorem 13.3 of \cite{Borcherds98}. Finally, for each primitive null vector $z\in M$, the form $\Psi_M(F)$ admits an infinite product expansion convergent in a suitable region of the Siegel upper half-space  (see point 5 of Theorem 13.3 of \cite{Borcherds98}). In particular, if we take $z$ to be a primitive null vector in the component $\Gamma^{1,1}(N)$ of $M=\Gamma^{1,1}\oplus \Gamma^{1,1}(N)\oplus \langle -2\rangle$, the infinite product becomes exactly the one in eq.\eqref{PhiSieg}. This proves our claims in section \ref{s:elevenandfriends}. Expansions of $\Psi_M(F)$ at other cusps (i.e. different primitive null vectors) can also be related to `second quantized twisted-twining genera', i.e. generating functions for twisted-twining genera in $K3^{[n]}$ models. 
 
 More generally, for any commuting $g,h\in G$, generating a group $\ZZ_N\times \ZZ_M$, one can obtain a `second quantized twisted-twining genus' starting from the vector valued Jacobi form whose components are the twisted-twining genera $\phi_{g^sh^r,g^kh^j}$, $s,k\in \ZZ/N\ZZ$, $r,j\in \ZZ/M\ZZ$ where the group generated by $g$ and $h$ is not necessarily cyclic. The procedure is analogous: one considers a double `discrete Fourier transform' $\hat \phi_{s,r,t,u}$, analogous to \eqref{FourTr}, whose Fourier coefficients are the $\ZZ_2$-graded dimensions of simultaneous $g$- and $h$-eigenspaces in the $g^sh^r$-twisted sector. Then, one performs a theta decomposition analogous to \eqref{thetadec} to get some functions $F_{s,r,t,u,\ell}$. One can easily show that these functions are the components of a vector-valued modular form of weight $-1/2$ for the Weil representation $\rho_M$ of $Mp(2,\ZZ)$ with $M=\Gamma^{1,1}(M)\oplus \Gamma^{1,1}(N)\oplus \langle -2\rangle$. By applying Theorem 13.3 of \cite{Borcherds98}, one gets a meromorphic Siegel modular form of genus two for a congruence subgroup of $Sp(4,\ZZ)$, which admits several different infinite product expansions near its cusps. These expansions are related to generating functions for twisted-twining genera of symmetric products of the K3 model. 
 
 The construction also generalizes to the case where the multiplier is not trivial, at least in the case where $g,h$ generate a cyclic group. The presence of a non-trivial multiplier can be interpreted as the fact that the order $\tilde N$ of the symmetry $g$, when acting on the twisted sectors, is larger (in fact, a multiple) than the order $N$ of $g$ on the untwisted sector. The construction then works simply by replacing $N$ with $\tilde N$ everywhere. In the context of orbifolds of vertex operator algebras, similar ideas were used in an example in section 3.5 of \cite{Carnahan2014}, and were described in some detail in appendix B of \cite{Paquette:2016xoo}; we refer to these references for more details. Second quantized twisted-twining genera when the group generated by $g,h$ is not cyclic and the multiplier is not trivial were considered in \cite{Persson:2013xpa}.
 
One of the hypotheses in Theorem 13.3 of \cite{Borcherds98} is that the Fourier coefficients $c_\delta(n)$ of $F_\delta$ are integral for all $n<0$ and for all $\delta\in M^*/M$. In our case, the `proof' that this condition is satisfied (for all $n$, not just for the negative ones) comes from physics. Indeed, it follows from the fact that the coefficients $c_\delta(n)$  can be interpreted as $\ZZ_2$-graded dimensions of certain eigenspaces in the twisted sector of the theory. In turn, this interpretation is based on the claim, justified by path integral arguments, that the $SL(2,\ZZ)$-transformations of  the twining genus $\phi_g$  correspond to traces of the form \eqref{twtwdef}. In order to prove that eq.\eqref{PhiSieg} is a Siegel modular form without using any input from physics, one needs to verify  this assumption through a (tedious) case by case analysis -- notice that the functions $F_\delta$ are known explicitly, and that there are only a finite number of coefficients $c_\delta(n)$ with negative $n$.

\bibliographystyle{utphys}
	\bibliography{Refs}

\providecommand{\href}[2]{#2}\begingroup\raggedright\begin{thebibliography}{10}

\bibitem{Maldacena:1997re}
J.~M. Maldacena, ``{The Large N limit of superconformal field theories and
  supergravity},'' \href{http://dx.doi.org/10.1023/A:1026654312961,
  10.4310/ATMP.1998.v2.n2.a1}{{\em Int. J. Theor. Phys.} {\bfseries 38} (1999)
  1113--1133}, \href{http://arxiv.org/abs/hep-th/9711200}{{\ttfamily
  arXiv:hep-th/9711200 [hep-th]}}.
[Adv. Theor. Math. Phys.2,231(1998)].

\bibitem{K3symm}
M.~R. Gaberdiel, S.~Hohenegger, and R.~Volpato, ``{Symmetries of K3 sigma
  models},'' \href{http://dx.doi.org/10.4310/CNTP.2012.v6.n1.a1}{{\em
  Commun.Num.Theor.Phys.} {\bfseries 6} (2012) 1--50},
\href{http://arxiv.org/abs/1106.4315}{{\ttfamily arXiv:1106.4315 [hep-th]}}.

\bibitem{Dijkgraaf:1998gf}
R.~Dijkgraaf, ``{Instanton strings and hyperKahler geometry},''
  \href{http://dx.doi.org/10.1016/S0550-3213(98)00869-4}{{\em Nucl. Phys.}
  {\bfseries B543} (1999) 545--571},
\href{http://arxiv.org/abs/hep-th/9810210}{{\ttfamily arXiv:hep-th/9810210
  [hep-th]}}.

\bibitem{Seiberg:1999xz}
N.~Seiberg and E.~Witten, ``{The D1 / D5 system and singular CFT},''
  \href{http://dx.doi.org/10.1088/1126-6708/1999/04/017}{{\em JHEP} {\bfseries
  04} (1999) 017},
\href{http://arxiv.org/abs/hep-th/9903224}{{\ttfamily arXiv:hep-th/9903224
  [hep-th]}}.

\bibitem{Mongardi2015}
G.~Mongardi, ``A note on the {K}\"{a}hler and {M}ori cones of hyperk\"{a}hler
  manifolds,'' \href{http://dx.doi.org/10.4310/AJM.2015.v19.n4.a1}{{\em Asian
  J. Math.} {\bfseries 19} no.~4, (2015) 583--591}.
  \url{https://doi.org/10.4310/AJM.2015.v19.n4.a1}.

\bibitem{Mongardi2016}
G.~Mongardi, ``Towards a classification of symplectic automorphisms on
  manifolds of {$K3^{[n]}$} type,''
  \href{http://dx.doi.org/10.1007/s00209-015-1557-x}{{\em Math. Z.} {\bfseries
  282} no.~3-4, (2016) 651--662}.
  \url{https://doi.org/10.1007/s00209-015-1557-x}.

\bibitem{huybrechts}
D.~Huybrechts, ``On derived categories of {K}3 surfaces, symplectic
  automorphisms and the {C}onway group,'' in {\em Development of moduli
  theory---{K}yoto 2013}, vol.~69 of {\em Adv. Stud. Pure Math.}, pp.~387--405.
\newblock Math. Soc. Japan, [Tokyo], 2016.
\newblock \href{http://arxiv.org/abs/1309.6528}{{\ttfamily 1309.6528
  [math.AG]}}.

\bibitem{HohnMason2019}
G.~H\"{o}hn and G.~Mason, ``Finite groups of symplectic automorphisms of
  hyperk\"{a}hler manifolds of type {$K3^{[2]}$},'' {\em Bull. Inst. Math.
  Acad. Sin. (N.S.)} {\bfseries 14} no.~2, (2019) 189--264,
  \href{http://arxiv.org/abs/1409.6055}{{\ttfamily arXiv:1409.6055 [math.AG]}}.

\bibitem{mukai1988finite}
S.~Mukai, ``{Finite groups of automorphisms of K3 surfaces and the Mathieu
  group},'' {\em Inventiones mathematicae} {\bfseries 94} no.~1, (1988)
  183--221.

\bibitem{Cheng:2016org}
M.~C.~N. Cheng, S.~M. Harrison, R.~Volpato, and M.~Zimet, ``{K3 String Theory,
  Lattices and Moonshine},''
  \href{http://dx.doi.org/10.1007/s40687-018-0150-4}{{\em Res. Math. Sci.}
  {\bfseries 5} no.~3, (2018) 45},
  \href{http://arxiv.org/abs/1612.04404}{{\ttfamily arXiv:1612.04404
  [hep-th]}}.
\url{https://doi.org/10.1007/s40687-018-0150-4}.

\bibitem{Paquette:2017gmb}
N.~M. Paquette, R.~Volpato, and M.~Zimet, ``{No More Walls! A Tale of
  Modularity, Symmetry, and Wall Crossing for 1/4 BPS Dyons},''
  \href{http://dx.doi.org/10.1007/JHEP05(2017)047}{{\em JHEP} {\bfseries 05}
  (2017) 047},
\href{http://arxiv.org/abs/1702.05095}{{\ttfamily arXiv:1702.05095 [hep-th]}}.

\bibitem{EOT}
T.~Eguchi, H.~Ooguri, and Y.~Tachikawa, ``{Notes on the K3 Surface and the
  Mathieu group $M_{24}$},''
  \href{http://dx.doi.org/10.1080/10586458.2011.544585}{{\em Exper.Math.}
  {\bfseries 20} (2011) 91--96},
\href{http://arxiv.org/abs/1004.0956}{{\ttfamily arXiv:1004.0956 [hep-th]}}.

\bibitem{Cheng:2012tq}
M.~C. Cheng, J.~F. Duncan, and J.~A. Harvey, ``{Umbral Moonshine},''
  \href{http://dx.doi.org/10.4310/CNTP.2014.v8.n2.a1}{{\em
  Commun.Num.TheorPhys.} {\bfseries 08} (2014) 101--242},
\href{http://arxiv.org/abs/1204.2779}{{\ttfamily arXiv:1204.2779 [math.RT]}}.

\bibitem{Cheng:2013wca}
M.~C. Cheng, J.~F. Duncan, and J.~A. Harvey, ``{Umbral moonshine and the
  Niemeier lattices},'' \href{http://dx.doi.org/10.1186/2197-9847-1-3}{{\em
  Research in the Mathematical Sciences} {\bfseries 1} no.~1, (2014) 1--81},
  \href{http://arxiv.org/abs/1307.5793}{{\ttfamily arXiv:1307.5793 [math.RT]}}.
\url{http://dx.doi.org/10.1186/2197-9847-1-3}.

\bibitem{Cheng:2014zpa}
M.~C.~N. Cheng and S.~Harrison, ``{Umbral Moonshine and K3 Surfaces},''
  \href{http://dx.doi.org/10.1007/s00220-015-2398-5}{{\em Commun. Math. Phys.}
  {\bfseries 339} no.~1, (2015) 221--261},
\href{http://arxiv.org/abs/1406.0619}{{\ttfamily arXiv:1406.0619 [hep-th]}}.

\bibitem{Aspinwall:1996mn}
P.~S. Aspinwall, ``{K3 surfaces and string duality},'' in {\em {Fields, strings
  and duality. Proceedings, Summer School, Theoretical Advanced Study Institute
  in Elementary Particle Physics, TASI'96, Boulder, USA, June 2-28, 1996}},
  pp.~421--540.
\newblock 1996.
\newblock
\href{http://arxiv.org/abs/hep-th/9611137}{{\ttfamily arXiv:hep-th/9611137
  [hep-th]}}.
\newblock

\bibitem{Nahm:1999ps}
W.~Nahm and K.~Wendland, ``{A Hiker's guide to K3: Aspects of N=(4,4)
  superconformal field theory with central charge c = 6},''
  \href{http://dx.doi.org/10.1007/PL00005548}{{\em Commun. Math. Phys.}
  {\bfseries 216} (2001) 85--138},
\href{http://arxiv.org/abs/hep-th/9912067}{{\ttfamily arXiv:hep-th/9912067
  [hep-th]}}.

\bibitem{Aspinwall:1995td}
P.~S. Aspinwall and D.~R. Morrison, ``{U duality and integral structures},''
  \href{http://dx.doi.org/10.1016/0370-2693(95)00745-7}{{\em Phys. Lett.}
  {\bfseries B355} (1995) 141--149},
\href{http://arxiv.org/abs/hep-th/9505025}{{\ttfamily arXiv:hep-th/9505025
  [hep-th]}}.

\bibitem{Johnson:1999qt}
C.~V. Johnson, A.~W. Peet, and J.~Polchinski, ``{Gauge theory and the excision
  of repulson singularities},''
  \href{http://dx.doi.org/10.1103/PhysRevD.61.086001}{{\em Phys. Rev.}
  {\bfseries D61} (2000) 086001},
\href{http://arxiv.org/abs/hep-th/9911161}{{\ttfamily arXiv:hep-th/9911161
  [hep-th]}}.

\bibitem{Witten:1995zh}
E.~Witten, ``{Some comments on string dynamics},'' in {\em {Future perspectives
  in string theory. Proceedings, Conference, Strings'95, Los Angeles, USA,
  March 13-18, 1995}}, pp.~501--523.
\newblock 1995.
\newblock
\href{http://arxiv.org/abs/hep-th/9507121}{{\ttfamily arXiv:hep-th/9507121
  [hep-th]}}.
\newblock

\bibitem{Lunin:2001pw}
O.~Lunin and S.~D. Mathur, ``{Three point functions for M(N) / S(N) orbifolds
  with N=4 supersymmetry},''
  \href{http://dx.doi.org/10.1007/s002200200638}{{\em Commun. Math. Phys.}
  {\bfseries 227} (2002) 385--419},
\href{http://arxiv.org/abs/hep-th/0103169}{{\ttfamily arXiv:hep-th/0103169
  [hep-th]}}.

\bibitem{Volpato:2014zla}
R.~Volpato, ``{On symmetries of $\mathcal{N}=(4,4)$ sigma models on $T^4$},''
  \href{http://dx.doi.org/10.1007/JHEP08(2014)094}{{\em JHEP} {\bfseries 1408}
  (2014) 094},
\href{http://arxiv.org/abs/1403.2410}{{\ttfamily arXiv:1403.2410 [hep-th]}}.

\bibitem{Gaberdiel:2013psa}
M.~R. Gaberdiel, A.~Taormina, R.~Volpato, and K.~Wendland, ``{A K3 sigma model
  with $\mathbb{Z}^8_2$ : $\mathbb{M}_{20}$ symmetry},''
  \href{http://dx.doi.org/10.1007/JHEP02(2014)022}{{\em JHEP} {\bfseries 1402}
  (2014) 022},
\href{http://arxiv.org/abs/1309.4127}{{\ttfamily arXiv:1309.4127 [hep-th]}}.

\bibitem{HohnMason}
G.~Hoehn and G.~Mason, ``{The 290 fixed-point sublattices of the Leech
  lattice},'' \href{http://dx.doi.org/10.1016/j.jalgebra.2015.08.028}{{\em J.
  Algebra} {\bfseries 448} (2016) 618--637},
  \href{http://arxiv.org/abs/1505.06420}{{\ttfamily arXiv:1505.06420
  [math.GR]}}.
\url{https://doi.org/10.1016/j.jalgebra.2015.08.028}.

\bibitem{eichler_zagier}
M.~Eichler and D.~Zagier, {\em {The theory of Jacobi forms}}.
\newblock Birkh{\"a}user, 1985.

\bibitem{Conway:1985vn}
J.~H. Conway, R.~T. Curtis, S.~P. Norton, R.~A. Parker, and R.~A. Wilson, {\em
  Atlas of finite groups}.
\newblock Oxford University Press, Eynsham, 1985.
\newblock Maximal subgroups and ordinary characters for simple groups, With
  computational assistance from J. G. Thackray.

\bibitem{HaradaLang1990}
K.~Harada and M.-L. Lang, ``On some sublattices of the {L}eech lattice,''
  \href{http://dx.doi.org/10.14492/hokmj/1381517491}{{\em Hokkaido Math. J.}
  {\bfseries 19} no.~3, (1990) 435--446}.
  \url{http://dx.doi.org/10.14492/hokmj/1381517491}.

\bibitem{Nikulin}
V.~V. Nikulin, ``Integer symmetric bilinear forms and some of their geometric
  applications,'' {\em Izv. Akad. Nauk SSSR Ser. Mat.} {\bfseries 43} no.~1,
  (1979) 111--177, 238.

\bibitem{MirandaMorrison1}
R.~Miranda and D.~R. Morrison, ``The number of embeddings of integral quadratic
  forms. {I},'' {\em Proc. Japan Acad. Ser. A Math. Sci.} {\bfseries 61}
  no.~10, (1985) 317--320.
  \url{http://projecteuclid.org/euclid.pja/1195514534}.

\bibitem{MirandaMorrison2}
R.~Miranda and D.~R. Morrison, ``The number of embeddings of integral quadratic
  forms. {II},'' {\em Proc. Japan Acad. Ser. A Math. Sci.} {\bfseries 62}
  no.~1, (1986) 29--32. \url{http://projecteuclid.org/euclid.pja/1195514495}.

\bibitem{Dijkgraaf:1996it}
R.~Dijkgraaf, E.~P. Verlinde, and H.~L. Verlinde, ``{Counting dyons in N=4
  string theory},'' \href{http://dx.doi.org/10.1016/S0550-3213(96)00640-2}{{\em
  Nucl. Phys.} {\bfseries B484} (1997) 543--561},
\href{http://arxiv.org/abs/hep-th/9607026}{{\ttfamily arXiv:hep-th/9607026
  [hep-th]}}.

\bibitem{Dijkgraaf:1996xw}
R.~Dijkgraaf, G.~W. Moore, E.~P. Verlinde, and H.~L. Verlinde, ``{Elliptic
  genera of symmetric products and second quantized strings},''
  \href{http://dx.doi.org/10.1007/s002200050087}{{\em Commun. Math. Phys.}
  {\bfseries 185} (1997) 197--209},
\href{http://arxiv.org/abs/hep-th/9608096}{{\ttfamily arXiv:hep-th/9608096
  [hep-th]}}.

\bibitem{Borcherds98}
R.~E. Borcherds, ``Automorphic forms with singularities on {G}rassmannians,''
  \href{http://dx.doi.org/10.1007/s002220050232}{{\em Invent. Math.} {\bfseries
  132} no.~3, (1998) 491--562}. \url{http://dx.doi.org/10.1007/s002220050232}.

\bibitem{Persson:2015jka}
D.~Persson and R.~Volpato, ``{Fricke S-duality in CHL models},''
  \href{http://dx.doi.org/10.1007/JHEP12(2015)156}{{\em JHEP} {\bfseries 12}
  (2015) 156},
\href{http://arxiv.org/abs/1504.07260}{{\ttfamily arXiv:1504.07260 [hep-th]}}.

\bibitem{Argurio:2000tb}
R.~Argurio, A.~Giveon, and A.~Shomer, ``{Superstrings on AdS(3) and symmetric
  products},'' \href{http://dx.doi.org/10.1088/1126-6708/2000/12/003}{{\em
  JHEP} {\bfseries 12} (2000) 003},
\href{http://arxiv.org/abs/hep-th/0009242}{{\ttfamily arXiv:hep-th/0009242
  [hep-th]}}.

\bibitem{Giribet:2018ada}
G.~Giribet, C.~Hull, M.~Kleban, M.~Porrati, and E.~Rabinovici, ``{Superstrings
  on AdS$_{3}$ at $k =$ 1},''
  \href{http://dx.doi.org/10.1007/JHEP08(2018)204}{{\em JHEP} {\bfseries 08}
  (2018) 204},
\href{http://arxiv.org/abs/1803.04420}{{\ttfamily arXiv:1803.04420 [hep-th]}}.

\bibitem{Gaberdiel:2018rqv}
M.~R. Gaberdiel and R.~Gopakumar, ``{Tensionless string spectra on
  AdS$_{3}$},'' \href{http://dx.doi.org/10.1007/JHEP05(2018)085}{{\em JHEP}
  {\bfseries 05} (2018) 085},
\href{http://arxiv.org/abs/1803.04423}{{\ttfamily arXiv:1803.04423 [hep-th]}}.

\bibitem{Eberhardt:2018ouy}
L.~Eberhardt, M.~R. Gaberdiel, and R.~Gopakumar, ``{The Worldsheet Dual of the
  Symmetric Product CFT},''
  \href{http://dx.doi.org/10.1007/JHEP04(2019)103}{{\em JHEP} {\bfseries 04}
  (2019) 103},
\href{http://arxiv.org/abs/1812.01007}{{\ttfamily arXiv:1812.01007 [hep-th]}}.

\bibitem{Eberhardt:2019qcl}
L.~Eberhardt and M.~R. Gaberdiel, ``{String theory on
  $\boldsymbol{\text{AdS}_{\mathbf{3}}}$ and the symmetric orbifold of
  Liouville theory},''
\href{http://arxiv.org/abs/1903.00421}{{\ttfamily arXiv:1903.00421 [hep-th]}}.

\bibitem{Strominger:1996sh}
A.~Strominger and C.~Vafa, ``{Microscopic origin of the Bekenstein-Hawking
  entropy},'' \href{http://dx.doi.org/10.1016/0370-2693(96)00345-0}{{\em Phys.
  Lett.} {\bfseries B379} (1996) 99--104},
\href{http://arxiv.org/abs/hep-th/9601029}{{\ttfamily arXiv:hep-th/9601029
  [hep-th]}}.

\bibitem{gaiotto2005re}
D.~Gaiotto, ``{Re-recounting dyons in N=4 string theory},''
\href{http://arxiv.org/abs/hep-th/0506249}{{\ttfamily arXiv:hep-th/0506249
  [hep-th]}}.

\bibitem{Shih:2005uc}
D.~Shih, A.~Strominger, and X.~Yin, ``{Recounting Dyons in N=4 string
  theory},'' \href{http://dx.doi.org/10.1088/1126-6708/2006/10/087}{{\em JHEP}
  {\bfseries 10} (2006) 087},
\href{http://arxiv.org/abs/hep-th/0505094}{{\ttfamily arXiv:hep-th/0505094
  [hep-th]}}.

\bibitem{shih2006exact}
D.~Shih and X.~Yin, ``{Exact black hole degeneracies and the topological
  string},'' \href{http://dx.doi.org/10.1088/1126-6708/2006/04/034}{{\em JHEP}
  {\bfseries 04} (2006) 034},
\href{http://arxiv.org/abs/hep-th/0508174}{{\ttfamily arXiv:hep-th/0508174
  [hep-th]}}.

\bibitem{Jatkar:2005bh}
D.~P. Jatkar and A.~Sen, ``{Dyon spectrum in CHL models},''
  \href{http://dx.doi.org/10.1088/1126-6708/2006/04/018}{{\em JHEP} {\bfseries
  04} (2006) 018},
\href{http://arxiv.org/abs/hep-th/0510147}{{\ttfamily arXiv:hep-th/0510147
  [hep-th]}}.

\bibitem{DS}
J.~R. David and A.~Sen, ``{CHL Dyons and Statistical Entropy Function from
  D1-D5 System},'' \href{http://dx.doi.org/10.1088/1126-6708/2006/11/072}{{\em
  JHEP} {\bfseries 11} (2006) 072},
\href{http://arxiv.org/abs/hep-th/0605210}{{\ttfamily arXiv:hep-th/0605210
  [hep-th]}}.

\bibitem{DJS1}
J.~R. David, D.~P. Jatkar, and A.~Sen, ``{Dyon spectrum in generic N=4
  supersymmetric Z(N) orbifolds},''
  \href{http://dx.doi.org/10.1088/1126-6708/2007/01/016}{{\em JHEP} {\bfseries
  01} (2007) 016},
\href{http://arxiv.org/abs/hep-th/0609109}{{\ttfamily arXiv:hep-th/0609109
  [hep-th]}}.

\bibitem{DJS2}
J.~R. David, D.~P. Jatkar, and A.~Sen, ``{Dyon Spectrum in N=4 Supersymmetric
  Type II String Theories},''
  \href{http://dx.doi.org/10.1088/1126-6708/2006/11/073}{{\em JHEP} {\bfseries
  11} (2006) 073},
\href{http://arxiv.org/abs/hep-th/0607155}{{\ttfamily arXiv:hep-th/0607155
  [hep-th]}}.

\bibitem{Cheng:2010pq}
M.~C. Cheng, ``{K3 Surfaces, N=4 Dyons, and the Mathieu Group M24},''
  \href{http://dx.doi.org/10.4310/CNTP.2010.v4.n4.a2}{{\em
  Commun.Num.Theor.Phys.} {\bfseries 4} (2010) 623--658},
\href{http://arxiv.org/abs/1005.5415}{{\ttfamily arXiv:1005.5415 [hep-th]}}.

\bibitem{Gaberdiel:2010ch}
M.~R. Gaberdiel, S.~Hohenegger, and R.~Volpato, ``{Mathieu twining characters
  for K3},'' \href{http://dx.doi.org/10.1007/JHEP09(2010)058}{{\em JHEP}
  {\bfseries 1009} (2010) 058},
\href{http://arxiv.org/abs/1006.0221}{{\ttfamily arXiv:1006.0221 [hep-th]}}.

\bibitem{Gaberdiel:2010ca}
M.~R. Gaberdiel, S.~Hohenegger, and R.~Volpato, ``{Mathieu Moonshine in the
  elliptic genus of K3},''
  \href{http://dx.doi.org/10.1007/JHEP10(2010)062}{{\em JHEP} {\bfseries 1010}
  (2010) 062},
\href{http://arxiv.org/abs/1008.3778}{{\ttfamily arXiv:1008.3778 [hep-th]}}.

\bibitem{Eguchi:2010fg}
T.~Eguchi and K.~Hikami, ``{Note on Twisted Elliptic Genus of K3 Surface},''
  \href{http://dx.doi.org/10.1016/j.physletb.2010.10.017}{{\em Phys.Lett.}
  {\bfseries B694} (2011) 446--455},
\href{http://arxiv.org/abs/1008.4924}{{\ttfamily arXiv:1008.4924 [hep-th]}}.

\bibitem{Gannon:2012ck}
T.~Gannon, ``{Much ado about Mathieu},''
\href{http://arxiv.org/abs/1211.5531}{{\ttfamily arXiv:1211.5531 [math.RT]}}.

\bibitem{DuncanGriffinOno2015}
J.~F.~R. Duncan, M.~J. Griffin, and K.~Ono, ``Proof of the umbral moonshine
  conjecture,'' \href{http://dx.doi.org/10.1186/s40687-015-0044-7}{{\em Res.
  Math. Sci.} {\bfseries 2} (2015) Art. 26, 47}.
  \url{https://doi.org/10.1186/s40687-015-0044-7}.

\bibitem{Bossard:2016zdx}
G.~Bossard, C.~Cosnier-Horeau, and B.~Pioline, ``{Protected couplings and BPS
  dyons in half-maximal supersymmetric string vacua},''
  \href{http://dx.doi.org/10.1016/j.physletb.2016.12.035}{{\em Phys. Lett.}
  {\bfseries B765} (2017) 377--381},
\href{http://arxiv.org/abs/1608.01660}{{\ttfamily arXiv:1608.01660 [hep-th]}}.

\bibitem{Bossard:2017wum}
G.~Bossard, C.~Cosnier-Horeau, and B.~Pioline, ``{Four-derivative couplings and
  BPS dyons in heterotic CHL orbifolds},''
  \href{http://dx.doi.org/10.21468/SciPostPhys.3.1.008}{{\em SciPost Phys.}
  {\bfseries 3} no.~1, (2017) 008},
\href{http://arxiv.org/abs/1702.01926}{{\ttfamily arXiv:1702.01926 [hep-th]}}.

\bibitem{Bossard:2018rlt}
G.~Bossard, C.~Cosnier-Horeau, and B.~Pioline, ``{Exact effective interactions
  and 1/4-BPS dyons in heterotic CHL orbifolds},''
\href{http://arxiv.org/abs/1806.03330}{{\ttfamily arXiv:1806.03330 [hep-th]}}.

\bibitem{Kachru:2016ttg}
S.~Kachru, N.~M. Paquette, and R.~Volpato, ``{3D String Theory and Umbral
  Moonshine},'' \href{http://dx.doi.org/10.1088/1751-8121/aa6e07}{{\em J.
  Phys.} {\bfseries A50} no.~40, (2017) 404003},
\href{http://arxiv.org/abs/1603.07330}{{\ttfamily arXiv:1603.07330 [hep-th]}}.

\bibitem{Zimet:2018dev}
M.~Zimet, ``{Umbral Moonshine and String Duality},''
\href{http://arxiv.org/abs/1803.07567}{{\ttfamily arXiv:1803.07567 [hep-th]}}.

\bibitem{Persson:2013xpa}
D.~Persson and R.~Volpato, ``{Second Quantized Mathieu Moonshine},''
  \href{http://dx.doi.org/10.4310/CNTP.2014.v8.n3.a2}{{\em
  Commun.Num.TheorPhys.} {\bfseries 08} (2014) 403--509},
\href{http://arxiv.org/abs/1312.0622}{{\ttfamily arXiv:1312.0622 [hep-th]}}.

\bibitem{Gaberdiel:2012gf}
M.~R. Gaberdiel, D.~Persson, H.~Ronellenfitsch, and R.~Volpato, ``{Generalized
  Mathieu Moonshine},''
  \href{http://dx.doi.org/10.4310/CNTP.2013.v7.n1.a5}{{\em Commun.Num.Theor
  Phys.} {\bfseries 07} (2013) 145--223},
\href{http://arxiv.org/abs/1211.7074}{{\ttfamily arXiv:1211.7074 [hep-th]}}.

\bibitem{Gaberdiel:2013nya}
M.~R. Gaberdiel, D.~Persson, and R.~Volpato, ``{Generalised Moonshine and
  Holomorphic Orbifolds},'' in {\em {Proceedings, String-Math 2012, Bonn,
  Germany, July 16-21, 2012}}, pp.~73--86.
\newblock 2013.
\newblock
\href{http://arxiv.org/abs/1302.5425}{{\ttfamily arXiv:1302.5425 [hep-th]}}.
\newblock

\bibitem{Cheng:2016nto}
M.~C.~N. Cheng, P.~de~Lange, and D.~P.~Z. Whalen, ``{Generalised Umbral
  Moonshine},''
\href{http://arxiv.org/abs/1608.07835}{{\ttfamily arXiv:1608.07835 [math.RT]}}.

\bibitem{ConwaySloane}
J.~H. Conway and N.~J.~A. Sloane,
  \href{http://dx.doi.org/10.1007/978-1-4757-6568-7}{{\em Sphere packings,
  lattices and groups}}, vol.~290 of {\em Grundlehren der Mathematischen
  Wissenschaften [Fundamental Principles of Mathematical Sciences]}.
\newblock Springer-Verlag, New York, third~ed., 1999.
\newblock \url{http://dx.doi.org/10.1007/978-1-4757-6568-7}.

\bibitem{GHS2009}
V.~Gritsenko, K.~Hulek, and G.~K. Sankaran, ``Abelianisation of orthogonal
  groups and the fundamental group of modular varieties,''
  \href{http://dx.doi.org/10.1016/j.jalgebra.2009.01.037}{{\em J. Algebra}
  {\bfseries 322} no.~2, (2009) 463--478}.
  \url{https://doi.org/10.1016/j.jalgebra.2009.01.037}.

\bibitem{GritClery}
F.~Cl\'{e}ry and V.~Gritsenko, ``Siegel modular forms of genus 2 with the
  simplest divisor,'' \href{http://dx.doi.org/10.1112/plms/pdq036}{{\em Proc.
  Lond. Math. Soc. (3)} {\bfseries 102} no.~6, (2011) 1024--1052}.
  \url{https://doi.org/10.1112/plms/pdq036}.

\bibitem{Carnahan2014}
S.~Carnahan, ``Generalized moonshine, {II}: {B}orcherds products,''
  \href{http://dx.doi.org/10.1215/00127094-1548416}{{\em Duke Math. J.}
  {\bfseries 161} no.~5, (2012) 893--950},
  \href{http://arxiv.org/abs/0908.4223}{{\ttfamily arXiv:0908.4223}}.
  \url{http://dx.doi.org/10.1215/00127094-1548416}.

\bibitem{Paquette:2016xoo}
N.~M. Paquette, D.~Persson, and R.~Volpato, ``{Monstrous BPS-Algebras and the
  Superstring Origin of Moonshine},''
  \href{http://dx.doi.org/10.4310/CNTP.2016.v10.n3.a2}{{\em Commun. Num. Theor.
  Phys.} {\bfseries 10} (2016) 433--526},
\href{http://arxiv.org/abs/1601.05412}{{\ttfamily arXiv:1601.05412 [hep-th]}}.

\end{thebibliography}\endgroup

	
\end{document}